\newif\ifnotes
\newif\ifsubmit     %
\newif\ifllncs      %
\newif\ifexabs      %
\newif\ifblind  

\notesfalse
\submittrue

\ifllncs
\documentclass{llncs}

\else 
\documentclass[11pt]{article}
\usepackage{libertine}
\usepackage{fullpage}
\usepackage[letterpaper,top=2cm,bottom=2cm,left=3cm,right=3cm,marginparwidth=1.75cm]{geometry}
\fi

\usepackage[T1]{fontenc}

\usepackage{packages}
\usepackage{bpotp-bibliography}
\usepackage{general-macros}
\usepackage{project-macros}
\hypersetup{pdftitle={On Best-Possible One-Time Programs}}
\ifblind
\hypersetup{pdfauthor={}}
\fi

\title{On Best-Possible One-Time Programs}
\ifblind
\author{}
\ifllncs
\institute{}
\fi
\else
\author{Aparna Gupte\\ MIT \and Jiahui Liu\\Fujitsu Research \and Luowen Qian\\Northeastern University \and Justin Raizes\\NTT Research \and Bhaskar Roberts\\University of California, Berkeley \and Mark Zhandry\\Stanford University}
\fi
\date{}

\begin{document}

\maketitle
\thispagestyle{plain}

\begin{abstract}
One-time programs (OTPs) aim to let a user evaluate a program on a single input while revealing nothing else.
Classical OTPs require hardware assumptions.
Quantum measurement offers a possible physical basis for one-time use: obtaining an output may irreversibly disturb the program needed for further computation.
Yet deterministic functionalities remain impossible due to gentle-measurement attacks (Broadbent, Gutoski and Stebila, 2013).
While recent works achieve positive results for randomized functionalities with high-entropy outputs (Gunn, Movassagh 2025; Gupte, Liu, Raizes, Roberts, Vaikuntanathan 2025), the fundamental limits and the strongest achievable security notions remain poorly understood.

Inspired by classical obfuscation, we ask for a ``best-possible'' one-time compiler: a generic transformation that, for any functionality, leaks the least amount of information compared with any other one-time implementation.
We show that a generic best-possible one-time compiler
cannot exist even for classical randomized functionalities assuming $\mathsf{SZK} \not\subseteq \mathsf{BQP}$.

Given this impossibility, we identify a natural subclass of one-time compilers called ``testable one-time program'' compilers, which output quantum states augmented with reflection programs for themselves.
(1) We formulate a simplified, generalized \emph{Single-Effective-Query (SEQ)} simulation security notion for quantum channels.
The definition uses self-adjoint implementations whose behavior under arbitrary quantum interactions depends only on the channel.
We show that SEQ security implies best-possible testable one-time security.
(2) We construct SEQ-secure OTPs for \emph{all} quantum functionalities in the classical oracle model, yielding the first positive results for arbitrary quantum channels beyond classical randomized functionalities.
Thus, SEQ security could serve as a testable one-time analogue of virtual black-box (VBB) security in the many-time obfuscation setting.

Finally, we propose stateful quantum indistinguishability obfuscation (stateful quantum iO) --- quantum state obfuscation for stateful quantum programs.
It implies best-possible testable OTPs and is achievable in the classical oracle model, thus giving us a promising approach towards best-possible testable OTPs in the plain model.
\end{abstract}

\ifllncs \else
\clearpage
\tableofcontents
\clearpage

\pagebreak
\fi

\ifllncs
\let\bpotpSubsection\subsection
\let\subsection\section
\section{Introduction}

One-time programs (OTPs), introduced by Goldwasser, Kalai, and Rothblum~\cite{GKR08-otp}, aim to give a user a program that can be evaluated on one input of the user's choice and then becomes useless afterwards.
In the classical plain model this is impossible, since the user can simply copy the program before running it.
Quantum information offers a possible way around this copying attack: unknown quantum states cannot in general be cloned, and extracting an output may disturb the state that encodes a program.
However, Broadbent, Gutoski, and Stebila~\cite{BGS13-otp} showed that quantum states do not generically solve the one-time program problem: for deterministic classical functionalities, gentle-measurement techniques still allow many evaluations.

\paragraph{The Cryptographic Perspective.}
The central cryptographic challenge for OTPs is to identify the most appropriate security goal in the presence of unavoidable attacks.

Ben-David and Sattath~\cite{BDS23-tokens} constructed ``quantum one-time signature tokens'', which can be viewed as quantum OTPs for the signing functionality.
This circumvents the impossibility of~\cite{BGS13-otp} because the signing functionality is \emph{randomized}, with high-entropy classical outputs.
More recently, two works~\cite{GLRRV25-otp,GM24-otp} extended this approach to give quantum OTP constructions for general randomized classical functionalities.
They proposed security notions and proved guarantees for functionalities with entropic outputs, obtaining the first general positive results for OTPs.

Randomness alone, however, does not make the intuitive ideal of ``revealing the output of a single input'' achievable.
As pointed out in~\cite{GLRRV25-otp}, if some piece of information can be deterministically extracted from the probabilistic output, then a generalization of the attack of~\cite{BGS13-otp} extracts that information on multiple inputs.
For example, this applies when all possible outputs on a given input have the same parity.
Such attacks seem to violate the very intuition behind a one-time program.

The prior works \cite{GLRRV25-otp,GM24-otp} do provide one-time security guarantees for their constructions, but these come in several forms: some are tailored to particular applications like signing, while the simulation-based definitions are very involved.
In particular, \cite{GLRRV25-otp} propose a strong simulation-based notion called \emph{Single-Effective Query} (SEQ) security, which respects a quantum adversary's ability to uncompute its own evaluations.
It remains unclear both what these guarantees mean for the one-time security of general programs and whether they are the strongest security one could hope for.
Ideally, we would have a guarantee that a program is one-time except for certain classes of unavoidable attacks.
We discuss the prior security definitions further in \Cref{sec:more_discussions_prior_defs}.

This leads to a fundamental question:
\begin{center}
    \emph{For a given family of functionalities,\\ what is the strongest achievable one-time security notion for quantum OTPs?}
\end{center}

For inspiration, consider classical obfuscation, which, like OTPs, seeks to protect the implementation details, but unlike OTPs, permits arbitrarily many evaluations.
A similar difficulty arises even with this weaker goal.
Virtual black-box (VBB) obfuscation requires that the behavior of any adversary receiving an obfuscated program be indistinguishable from that of a simulator with black-box oracle access to the functionality.
This is provably impossible for certain functionalities~\cite{BGIRSV01-obfuscation}.

Two approaches have emerged from this impossibility.
The first is the weaker notion of \emph{indistinguishability obfuscation} (iO), proposed in~\cite{BGIRSV01-obfuscation}, which roughly requires that obfuscations of two equivalent programs be indistinguishable.
This notion avoids the impossibility and is potentially achievable (e.g.\ \cite{GGHRSW16-io,JLS26-io} in the pre-quantum setting, or~\cite{BGMZ18-zeroizing,HJL25-io}).
A priori, it is unclear what iO guarantees about a particular program: we usually care about the security of that program, and comparing it with equivalent programs does not directly explain what information is hidden.
Indeed, historically we did not know how to construct useful applications of iO until more than a decade after the definition was proposed \cite{SW14-io}.
The other approach retains the VBB goal and shows that it \emph{is} achievable for some very specific functionalities~\cite{CD08-point,LPS04-obfuscation,WZ17-cac,GKW17-lockable}.
Unfortunately, there is a wide gulf between what is known to be VBB obfuscatable and what is known to be un-obfuscatable even after decades of research.
Thus both approaches have major limitations.

Goldwasser and Rothblum~\cite{GR07-best} provide a satisfying way to unify them through the lens of ``best-possible'' obfuscation.
Their starting question is: \emph{what is the least information any obfuscated program can reveal?}
A best-possible obfuscator reveals no more than any other program computing the same functionality, except for the unavoidable leakage of an upper bound on the program length.
Informally, best-possible security requires that whatever can be efficiently learned from the obfuscated program must also be learnable from every such alternative implementation (of comparable length).
In particular, if a functionality can be VBB-obfuscated securely by \emph{any} obfuscator, a best-possible obfuscator must also provide VBB security for it.
Goldwasser and Rothblum show that best-possible obfuscation is actually \emph{equivalent} to iO.
With the subsequent emergence of iO constructions, this gives best-possible obfuscation for all programs.
To establish VBB security for a new functionality, it therefore suffices to show that \emph{some} VBB obfuscator exists: any iO scheme then provides that guarantee.

\paragraph{The Quantum Information Perspective.}
A fundamental problem at the intersection of quantum information and computation is to understand how destructive measurements limit the amount of information extractable from a quantum state.
In the information-theoretic setting, Holevo's theorem~\cite{Holevo73-bound} gives a foundational limitation: measuring an $n$-qubit encoding of a classical message can reveal at most $n$ bits of information about that message.
This bound on the total information extracted does not rule out giving the recipient a choice of \emph{which} information to extract.
Different choices can require incompatible measurements, so extracting one part of the information may destroy access to another.

This possibility motivates the study of quantum random access codes in quantum one-way communication complexity~\cite{ANTV02-qrac,Nayak99-qrac}.
This study asks how compactly a classical truth table can be encoded into a quantum state while allowing any chosen entry to be recovered with bounded error.
Such an encoding serves as a quantum program for evaluating the corresponding function on a single input of the user's choice.
For example, two bits can be encoded into one qubit so that either chosen bit can be recovered with probability about $0.85$~\cite{ANTV02-qrac}.
The measurements for the two choices are incompatible: decoding one bit erases the information needed to decode the other.
The asymptotic savings are limited, however: encoding an arbitrary $N$-bit truth table so that any chosen entry can be recovered with constant bounded error still requires $\Omega(N)$ qubits~\cite{Nayak99-qrac}.

Quantum complexity theory generalizes quantum random access codes to evaluating arbitrary programs through the perspective of quantum advice~\cite{Aaronson05-advice}.
The comparison between e.g.\ $\mathsf{BQP/poly}$ and $\mathsf{BQP/qpoly}$ asks whether access to a quantum program can enable a quantum computer to perform computations beyond those possible with a classical program.
Here, as with random access codes, correctness is required only for a single evaluation, and the pursuit of more efficient encodings naturally leads to examples of one-time behavior.
For certain search problems, which allow multiple valid outputs, prior work has shown that using a quantum program can be exponentially more efficient than a classical one \cite{Liu23-advice,ABK24-relational}.
Broadbent, Karvonen, and Lord~\cite{BKL24-advice} further show that some promise problems admit \emph{uncloneable quantum programs}: their computational usefulness cannot be duplicated by uniform efficient strategies for two separate recipients.
Gavinsky~\cite{Gavinsky12-qmoney} makes this one-time behavior explicit: a two-qubit instance of the hidden-matching encoding allows either of two chosen relational queries to be answered with certainty, while any strategy answering both succeeds with probability at most $3/4$ on a random instance.
While the efficiency and computational power of these programs have been extensively studied, their one-time behavior remains less systematically understood.

A separate line of work investigates the role of destructive measurement through a counterfactual question: what would change if measurement outcomes could be sampled without disturbing the quantum state?
The resulting models of computation with \emph{non-collapsing measurements} have been studied for their computational power~\cite{ABFL16-noncollapse,MP25-noncollapse} and the power of quantum verification~\cite{AGIMR25-noncollapse,BM25-noncollapse}.
For the random-access task above, Aaronson~\cite{Aaronson18-pdqp} shows that with these measurements, an arbitrary truth table can instead be encoded in a quantum state of polylogarithmic size while allowing efficient recovery of any chosen entry, in contrast to the linear lower bound for ordinary measurements.
These models have also been used to formulate cryptographic hardness assumptions~\cite{MSY25-noncollapse} and to separate cryptographic primitives~\cite{CKMSTY26-noncollapse}.

Quantum one-time programs offer a complementary and operational approach to study this question, by bringing the destructive nature of quantum measurement into the spotlight:
\begin{center}
    \emph{How and when can obtaining the result of a chosen computation come at an unavoidable cost to the ability to perform others?}
\end{center}

Measurement can make information extraction destructive, but coherent quantum evolution is reversible.
The recipient of a quantum program can control how it is evaluated: they may postpone measurements, evaluate on superpositions, or reverse an evaluation instead of following the intended procedure.
As mentioned above, techniques like gentle measurement allow repeated evaluation of deterministic classical functionalities, ruling out one-time programs for them \cite{Aaronson05-advice,BGS13-otp}.
For a randomized output with sufficient entropy, retaining the observed sample may instead obstruct this recovery.
The positive results for randomized functionalities discussed above show that restrictions on further information extraction can indeed be enforced.
However, the positive results are still limited, and there is a wide gap in our understanding of how this phenomenon manifests for other randomized functionalities.

These results leave open how strongly information extraction can be restricted for a general computation.
While the recipient (or adversary) is free to choose how to evaluate a given program, the designer (or cryptographer) is free to choose which correct implementation to provide.
Information exposed by one quantum program may be hidden by another program for the same functionality.
To understand the limits of one-time computation, we therefore need to distinguish information that every correct implementation must reveal from information exposed by a particular implementation.
We would like to understand the limits of the strongest one-time programs and find a generic way to achieve this security.

\paragraph{Our Question.}
Both perspectives therefore motivate the study of the best achievable security for quantum OTPs: programs that reveal no more information than any correct one-time implementation of comparable size for the same functionality.
More succinctly,
\begin{center}
    \emph{Is there a best-possible one-time compiler?}
\end{center}

\begin{remark}[Deterministic Functionalities]
The question remains meaningful even for deterministic functionalities, an important special case and the original setting considered by Goldwasser, Kalai, and Rothblum.
Here, best-possible one-time security is collapses to best-possible obfuscation.
Existing attacks make repeated evaluation unavoidable, so revealing the least amount of information among all one-time correct implementations is equivalent to revealing the least among all \emph{many-time} correct implementations.

Since our comparison class of possible obfuscations of a program includes quantum programs, the corresponding notion is \emph{quantum-state} obfuscation; restricting the comparison to classical programs recovers the classical notion of Goldwasser and Rothblum.
Assuming quantum state indistinguishability obfuscation~\cite{CG24-qio}, best-possible compilers are achievable for deterministic functionalities; existing ideal obfuscation results also give constructions in the classical oracle model~\cite{BBV24-qso}.
However, this is an unsatisfying answer to the one-time question: achieving best-possible security here means giving up one-time security altogether.

For the rest of this work, we therefore focus on the more general setting of randomized and even quantum functionalities, where learning an output may irreversibly disturb the program and allow stronger security than best-possible many-time obfuscation.
\end{remark}

\paragraph{Organization.}
Comparing implementations first requires specifying which behavior makes them implementations of the same functionality.
We therefore begin in \Cref{sec:intro-correctness} by clarifying the definition of correctness and addressing a common confusion.
\Cref{sec:intro-impossibility,sec:intro-impossibility-overview} then state and explain our impossibility result for unrestricted best-possible one-time programs.
This motivates the notion of \emph{testable} programs, introduced in \Cref{sec:intro-testable} to capture a natural subclass of one-time compilers.
We introduce SEQ for quantum channels, show that it is best possible among testable programs, and present our constructions in \Cref{sec:intro-quantum-functionality,sec:intro-seq-best-possible,sec:intro-constructions}, respectively.
\Cref{sec:intro-positive-overview} sketches the proofs of the positive results, and \Cref{sec:intro-discussion} discusses their implications and open questions.

\subsection{Defining Correctness}
\label{sec:intro-correctness}

In this section, we clarify our choice of correctness definition for obfuscation and OTPs to address a common confusion.
Recall the original, informal goal of obfuscation: preserve a functionality's input--output behavior while hiding everything else as much as possible~\cite{BGIRSV01-obfuscation}.
Correctness specifies the behavior that must be preserved; security concerns what a recipient should not learn from the implementation.
For one-time programs, correctness concerns only the first evaluation.
The \emph{ideal} security then should require that anything that is not preserved by the correctness requirement should be hidden.

\paragraph{For Deterministic Functionalities.}
We first recall the existing definitions for deterministic functionalities.
Consider a deterministic classical function $f:\calX\to\calY$.
A quantum program $P=(\rho,\Eval)$ consists of a possibly mixed state $\rho$ and a public evaluation circuit $\Eval$.
The recipient can run $\Eval$ coherently and in reverse\footnote{We take $\Eval$ to be unitary without loss of generality.}.
Directly formalizing the prescribed input--output behavior gives the following definition.

\begin{definition}[Classical correctness]\label{def:intro-classical-correctness}
    A program $P$ is (one-time) \emph{classically correct} for $f$ if evaluation returns $f(x)$ on each classical input $x$ with certainty.
\end{definition}

Another prevalent definition in the literature requires coherent evaluation.

\begin{definition}[Coherent correctness]\label{def:intro-coherent-correctness}
    A program $P$ is (one-time) \emph{coherently correct} for $f$ if its evaluator implements the unitary map
    \[
        U_f:\ket{x}\ket{z}\longmapsto\ket{x}\ket{z\oplus f(x)}.
    \]
\end{definition}
We set aside other subtleties of defining correctness here to focus on the distinction between these two notions.

In the black-box setting, the recipient can only query the supplied interface.
Here the two notions are inequivalent: coherent queries can reveal properties of $f$ that cannot be learned efficiently through classical queries.
For example, recovering the hidden period in Simon's problem takes polynomially many coherent queries but exponentially many classical queries~\cite{Simon94-power}.

Obfuscation instead gives the recipient white-box access to an implementation $P=(\rho,\Eval)$.
The public description of $\Eval$ permits coherent execution of both $\Eval$ and $\Eval^{-1}$, even when correctness is promised only on classical inputs.
The standard uncomputation argument underlying the impossibility of deterministic OTPs (see, e.g.\ \cite{BGS13-otp}) gives the following equivalence.

\begin{fact}[White-box equivalence of correctness notions]\label{claim:intro-correctness-equivalence}
    For deterministic functionalities, classical and coherent correctness are equivalent up to a uniform efficient transformation of the public evaluator using black-box calls to $\Eval$ and $\Eval^{-1}$.
\end{fact}

Thus coherent access is unavoidable for a quantum recipient of a deterministic program, motivating the use of coherent oracle access in models such as the QROM~\cite{BDFLSZ11-qrom}.

For deterministic functionalities, we draw three conclusions:
\begin{enumerate}
    \item \Cref{def:intro-classical-correctness,def:intro-coherent-correctness} are equivalent in the standard white-box model.
    \item \Cref{def:intro-classical-correctness} is \emph{the} definition of correctness that follows from first principles and should be the default.
    \item In the black-box model, however, \Cref{def:intro-coherent-correctness} more faithfully captures what a quantum recipient can learn from a correct implementation. This unavoidable access is \emph{the only} reason to prefer \Cref{def:intro-coherent-correctness} in that setting.\footnote{For deterministic functionalities, the white-box equivalence makes coherent correctness a convenient default. The literature can thus leave implicit the distinction between what correctness requires and what a recipient can derive. This convention can obscure the significance of that distinction when moving to sampling functionalities.}
\end{enumerate}

\paragraph{For Sampling Functionalities.}
Next, we turn to what it means to correctly implement a sampling functionality.
For a classical sampling task represented by a randomized circuit $f:\calX\times\calR\to\calY$, the specified behavior is
\begin{equation*}
    \label{eq:classical-sampling-interface}
    x\longmapsto f(x;r),\qquad r\getsr\calR.
\end{equation*}
Write $D_x$ for the resulting output distribution.
The user chooses one classical input and receives one classical sample; the random coins remain internal to the functionality.

As with \Cref{def:intro-classical-correctness}, the following definition arises from first principles.

\begin{definition}[Sampling correctness]\label{def:intro-sampling-correctness}
    A program $P$ is (one-time) \emph{sampling correct} for $f$ if evaluation produces exactly the distribution $D_x$ on each classical input $x$.
\end{definition}
This definition does not require preserving the output's dependence on a particular random tape.
Similarly, deterministic correctness does not require preserving the original circuit's internal wire values.
For obfuscation, correctness requires preserving the map from inputs to output distributions.
The security goal is to hide everything else as much as possible, including the output's dependence on the random tape.
We emphasize that \Cref{def:intro-sampling-correctness} has long been the default in the classical literature.
It dates back to the original definition of sampling obfuscation by Barak et al.\ \cite{BGIRSV01-obfuscation}.
Later examples include the probabilistic obfuscation notions of Canetti et al.\ \cite{CLTV15-pio} and Dodis et al.\ \cite[Definition~4.1]{DHRW16-spooky}.
Classical sampling is also the correctness interface used in prior work on randomized classical OTPs~\cite{GLRRV25-otp,GM24-otp}.
It appears in applications such as signature tokens and one-time NIZKs~\cite{BDS23-tokens,GLRRV25-otp}.

The uncomputation argument in \Cref{claim:intro-correctness-equivalence}, however, breaks down for sampling functionalities.
Retaining a random output can correlate it with the program and workspace, preventing correct uncomputation.
The argument therefore does not yield a well-defined analogue of \Cref{def:intro-coherent-correctness} for sampling functionalities.\footnote{Prior work by Gupte et al.\ \cite{GLRRV25-otp} prescribed a particular coherent interface for certain sampling functionalities in its simulation security notion. Extending the QROM is a natural starting point for a construction that makes extensive use of a random oracle. As we show later, their notion is a special case of our generalized SEQ framework. It corresponds to one particular quantum-channel extension of a sampling functionality; other extensions can also be used. This interface is therefore not intrinsic to the classical sampling task.}

One might nevertheless wonder whether there is an unknown generic attack for sampling functionalities, analogous to \Cref{claim:intro-correctness-equivalence}.
Such an attack might suggest a more ``appropriate'' analogue of \Cref{def:intro-coherent-correctness}.
It would not, however, change the best-possible comparison in the standard white-box model since the attacker is still permitted to carry out such attacks under \Cref{def:intro-sampling-correctness}.
More precisely, suppose a uniform efficient procedure converts any one-time sampling-correct program into one satisfying a stronger coherent correctness requirement.
The conversion must work from a single program instance and preserve the sampling functionality.
Now consider an implementation that reveals no more than any program satisfying the stronger requirement.
A recipient of any one-time sampling-correct program can first convert it into such a program.
The implementation therefore reveals no more than the original sampling-correct program either.
If the proposed notion is white-box equivalent to \Cref{def:intro-sampling-correctness}, these conversions work in both directions.
The two notions of best-possible one-time security are then equivalent as well.
Therefore, for our purposes, using \Cref{def:intro-sampling-correctness} is both justified by the intended use of obfuscation and technically sufficient.

Finally, a future OTP application may require a particular coherent extension of a sampling functionality.
We believe the cleanest approach is then to specify that extension as a quantum channel.
Looking ahead, our positive results construct OTPs for arbitrary quantum channels.

\subsection{Impossibility of Best-Possible One-Time Programs}
\label{sec:intro-impossibility}

For a classical sampling functionality, the most literal definition of best-possible one-time compilers says the following.
A one-time compiler $\OTP^*$ is best possible if, for every adversary attacking $\OTP^*(f)$ and for every efficient one-time program $P$ that implements the same sampling functionality $f$, there is a simulator given $P$ that reproduces the adversary's view.
In other words, $\OTP^*(f)$ should leak no more than any possible one-time implementation of $f$, since if $P$ achieves a certain security guarantee, then $\OTP^*(f)$ should also achieve that guarantee by the existence of the simulator.

\begin{definition}[Informal; \Cref{def:bp-sim}]
    A one-time compiler $\OTP^*$ is \emph{best-possible} for a family $\calF$ of sampling functionalities if there exists a simulator $\Sim$ such that for any sampling functionality $f \in \calF$, and for any quantum program $P$ that \emph{one-time} implements $f$, $\OTP^*(f, 1^{|P|})$ is computationally indistinguishable from $\Sim(P, 1^{|f|})$.
\end{definition}
The length of the program is also unavoidably leaked similar to indistinguishability obfuscation and best-possible obfuscation \cite{GR07-best}.

This definition is deliberately unrestricted.
A competing implementation $P$ may be a mixed quantum state, and it only needs to implement the sampling functionality in the one-query classical sense as in \Cref{def:intro-sampling-correctness}.
We emphasize that this is not an arbitrary technical choice.
Our central question of investigating ``the least amount of information leaked for any one-time implementation'' should include reasonable implementation that correctly samples $\{D_x\}$.
Thus the first-principle best-possible definition therefore has to compare against all of them.

Our first result shows that an unrestricted best-possible one-time compiler is impossible.
Following the notation of Goldreich, Sahai, and Vadhan~\cite{GSV98-szk}, let $\left.\mathsf{HVPZK}\right|_{\mathrm{am}}$ denote the class of promise problems that admit public-coin honest-verifier perfect zero knowledge proof systems, with strict polynomial-time simulation as specified in \Cref{sec:impossibility-from-hvpzk}.
Similarly, $\left.\mathsf{HVSZK}\right|_{\mathrm{am}}$ only requires statistical zero knowledge, and $\left.\mathsf{HVSZK}\right|_{\mathrm{am}} = \mathsf{SZK}$~\cite[Theorem~2]{GSV98-szk}, where $\mathsf{SZK}$ is the standard class without either the honest-verifier or public-coin restriction.

\begin{theorem}[Informal; \Cref{thm:impossibility-of-bp-otp-from-szk,thm:impossibility-of-bp-otp-from-hvpzk}]
    If $\mathsf{SZK} \not\subseteq \mathsf{BQP}$ (respectively, $\left.\mathsf{HVPZK}\right|_{\mathrm{am}} \not\subseteq \mathsf{BQP}$), no best-possible one-time compiler with statistical (respectively, perfect) equivalence exists for all classical sampling functionalities with a single input bit.
\end{theorem}

These assumptions are also implied by various post-quantum cryptographic assumptions, giving the following concrete consequences.

\begin{corollary}[Informal; \Cref{thm:impossibility-of-bp-otp}]
    Assuming the post-quantum hardness of LWE, best-possible one-time compilers do not exist for classical sampling functionalities under statistical equivalence.
    Assuming post-quantum weakly pseudorandom effective group actions, best-possible one-time compilers do not exist even under perfect equivalence.
\end{corollary}

To establish the second consequence, we construct a perfectly lossy encryption scheme from the same group action assumption, which may be of independent interest (\Cref{thm:lossy-PKE-from-group-actions}).

\subsection{Why Unrestricted Best-Possible One-Time Compilers Are Impossible}
\label{sec:intro-impossibility-overview}

We now explain the intuition behind the impossibility result.
Looking ahead, this section will also help motivate the positive results. We refer the reader to \Cref{sec:impossibility-of-bp-otp-proof} for a formal proof.

At a high level, the key insight relies on identifying two kinds of sampling functionalities that a best-possible compiler would have to treat differently.
\begin{enumerate}
    \item The first kind is essentially input-independent: the functionality samples from a distribution that does not depend on the input.
    Such a task has an extremely secure (similar to virtual black-box) one-time implementation: sample once and hardwire the sample into a constant program.
    Any best-possible compiler for this task must leak no more than this constant implementation.
    We call this the \emph{lossy} case, since the output carries essentially no information about the input.

    \item The second kind is input-dependent: the output uniquely identifies the input except with negligible probability, so intuitively a correct implementation cannot simply return a sample fixed before seeing the input.
    We call this the \emph{injective} case, where injectivity refers to recovering the input from the output, possibly inefficiently.
\end{enumerate}

We construct these two kinds of functionalities assuming the worst-case hardness of the $\mathsf{SZK}$-complete problem \emph{Statistical Difference}~\cite{SV03-szk}.
After polarization, a hard instance gives two classical sampling circuits $C_0,C_1$ whose output distributions $D_0,D_1$ are either $\exp(-n)$-close or $(1 - \exp(-n))$-far.
Consider the sampling functionality
\[
    f(b;r):=C_b(r), \qquad b\in\bit.
\]
In the close case, sampling from $D_0$ regardless of $b$ implements $f$ up to negligible error, giving the lossy case above.
In the far case, the two distributions put almost all their mass on disjoint sets, so the output identifies $b$, giving the injective case.
Thus, if an efficient best-possible compiler lets us distinguish these two kinds of functionalities, we obtain an efficient quantum algorithm for Statistical Difference, contradicting $\mathsf{SZK}\not\subseteq\mathsf{BQP}$.
\Cref{fig:intro-impossibility} summarizes this contradiction.

\begin{figure}[htbp]
\centering
\begin{tikzpicture}[x=\linewidth,y=1cm,
    box/.style={draw,rounded corners=2pt,align=center,font=\small,
        text width=.35\linewidth,inner sep=7pt},
    arrow/.style={->,>=stealth,semithick},
    explanation/.style={align=center,font=\small,text width=.38\linewidth},
    relation/.style={align=center,font=\small,text width=.15\linewidth}]
    \node[box] (d) at (.23,3.0)
        {Close sampler $f_{\mathsf{close}}$\\nearly input-independent};
    \node[box] (e) at (.77,3.0)
        {Far sampler $f_{\mathsf{far}}$\\output identifies the input bit};
    \node[relation] at (.5,3.0) {Hard to\\distinguish};
    \node[box,fill=black!4] (od) at (.23,.85) {$\OTP^*(f_{\mathsf{close}})$};
    \node[box,fill=black!4] (oe) at (.77,.85) {$\OTP^*(f_{\mathsf{far}})$};
    \draw[arrow] (d) -- node[right] {$\OTP^*$} (od);
    \draw[arrow] (e) -- node[right] {$\OTP^*$} (oe);
    \node[relation] at (.5,.85) {Easy to\\distinguish};
    \node[explanation,anchor=north] at (.23,.30)
        {Must act like constant\\by best-possible security};
    \node[explanation,anchor=north] at (.77,.30)
        {Cannot be constant\\by one-time correctness};
\end{tikzpicture}
\caption{Distinguishing the close and far cases from sampler descriptions is hard under $\mathsf{SZK}\not\subseteq\mathsf{BQP}$. The same best-possible compiler $\OTP^*$ followed by the EPR test would distinguish these cases, implying that $\OTP^*$ cannot be efficient.}
\label{fig:intro-impossibility}
\end{figure}

\begin{remark}[One-time VBB when possible]
\label{rem:intro-vbb-when-possible}
    Under the same assumptions, the same proof idea rules out an even weaker security notion.
    We call this notion ``one-time VBB when possible'': if the family of functionalities admits a one-time VBB compiler then we also require one-time VBB security; otherwise, there is no security required.
    We formalize this notion in \Cref{app:vbb-when-possible}.
\end{remark}

\paragraph{Proof idea in slightly more detail.}
Suppose we have a best-possible one-time compiler $\OTP^*$ and are given a polarized Statistical Difference instance $(C_0,C_1)$.
We compile $f(b;r)=C_b(r)$ and test whether evaluating the resulting program reveals the input bit to the output.
To do so, prepare an EPR pair $(\ket{00}+\ket{11})/\sqrt{2}$ on the input and an untouched reference register, evaluate the program, copy its output to a fresh register, and uncompute the evaluation.
Finally, test whether the input and reference have returned to their initial EPR state.
The test uses the public evaluation circuit and its inverse, so it is efficient even though the program state may be unknown.

In the injective case, correctness forces the copied output to identify the original input bit in the reference register, except with negligible probability.
This effectively measures the reference bit, and uncomputation cannot undo its correlation with the copied output because it touches neither register.
Consequently, the EPR test fails with probability at least $1/2-\negl$.

In the lossy case, correctness alone does not force $\OTP^*(f)$ to return a fixed sample.
The essential observation is that $f$ has a one-time implementation that is a mixture of constant programs: sample $y\gets D_0$ once, store $y$, and return it on either input.
For each fixed $y$, let $c_y(b)=y$.
Correctness forces $\OTP^*(c_y)$ to output $y$, up to negligible error, so copying its output causes negligible disturbance and uncomputation restores the EPR pair.
Thus the test passes with overwhelming probability on every compiled constant program.

To relate this behavior back to the behavior of $\OTP^*(f)$, we invoke best-possible security.
Let $P_y$ be the program storing $y$, and let $P$ be its mixture for $y\gets D_0$, with a common evaluator that copies the stored value.
Security supplies one simulator for $P$ and every $P_y$, and linearity gives
\[
    \OTP^*(f) \approx_c \Sim(P)
    = \{\Sim(P_y)\}_{y\gets D_0}
    \approx_c \{\OTP^*(c_y)\}_{y\gets D_0}.
\]
Since the EPR test is efficient, it must pass for $\OTP^*(f)$ with overwhelming probability as well.
Repeating the test with freshly compiled programs therefore distinguishes the close and far cases of Statistical Difference.

For perfect equivalence, we only need the two distributions in the lossy case to be exactly equal, so that the mixture of constant programs implements $f$ perfectly.
Malka's reduction~\cite[Lemma~4.4 of the full version]{Malka15-perfect} shows that the identical-versus-far variant of Statistical Difference is $\left.\mathsf{HVPZK}\right|_{\mathrm{am}}$-hard under randomized polynomial-time reductions.
The same EPR argument therefore gives the perfect-equivalence impossibility.
We formalize the perfect variant in \Cref{sec:impossibility-from-hvpzk}.

\begin{remark}
    We note that our impossibility can be dequantized to state that best-possible one-time programs are also impossible with classical obfuscators assuming the existence of the same (but classically secure) lossy encryptions used for the concrete consequences in \Cref{sec:lossy-pke}.
    The only step that needs to be changed is that rather than using a quantum distinguisher that tests if the input superposition collapses, we can simply rewind the classical program to test if it is a constant function.

    This exact step is also the place where the impossibility breaks down when hardware assumptions or memory coherence assumptions are introduced, and these are the two settings where prior works have successfully constructed strong one-time programs.
\end{remark}

\subsection{From Impossibility to Testable Programs}
\label{sec:intro-testable}

We established that unrestricted best-possible one-time compilers are impossible to hope for.
The natural next step is to ask whether there is a weaker but still meaningful notion of best-possible one-time compilers that can be achieved.
One possible direction is to restrict the class of functionalities to avoid the examples used by the impossibility result.
In this rest of this work, we instead restrict the comparison class of one-time compilers.

A valid takeaway from the impossibility result is that mixed one-time implementations are very secure since they can sometimes even achieve the ideal one-time functionality, but on the other hand, it is also impossible to efficiently identify when mixed-output techniques can be applied given only the circuit description as input.
Thus, a natural idea to get around the impossibility is to restrict the comparison class to compilers that output pure-state programs only.

\begin{remark}
    Restricting to pure-state programs may be very tempting because this is closer in spirit to how best-possible obfuscation is traditionally defined \cite{GR07-best}: one compares against fixed obfuscated programs rather than arbitrary mixtures.
    However, notice that for perfectly correct obfuscation of deterministic functionalities, the distinction between fixed programs vs mixture of programs collapses: if a mixed state (many-time) implements a deterministic functionality with zero error, then every pure state in its support implements the same functionality, so quantifying over mixed states only makes the definition conceptually more involved.
    Turning back to one-time programs of randomized functionalities, the distinction need not collapse: a mixed implementation can be correct only in aggregate while its pure components are not individually correct.
    
    With that said, the mixed state definition is still the most reasonable definition for best-possible one-time compilers since that is the most literal interpretation of the ideal goal of comparing against all one-time correct implementations.
\end{remark}

However, simply restricting the comparison class to pure-state implementations does not evade the impossibility.
To see the issue, let $C(r)$ be a classical circuit sampling the input-independent distribution in the impossibility, using $\ell$ uniformly random bits $r$.
For any $s\in\bit^\ell$, the pure state
\begin{equation}
\label{eq:C-alternative-decomposition}
\ket{\psi_s}
=
2^{-\ell/2}\sum_{r\in\bit^\ell}(-1)^{s\cdot r}\ket{r,C(r)}
\end{equation}
together with the evaluation procedure that copies the $C(r)$ register to the output register is already a pure-state implementation of this sampling task.
Each state has a polynomial-size exact preparation circuit: prepare the signed uniform superposition over $r$ and compute $C(r)$ reversibly, retaining the random tape.
Thus the constant-sampler side of the impossibility can be represented using efficiently preparable pure program states.

Indeed, averaging uniformly over $s$ cancels every off-diagonal term between distinct random tapes, giving exactly the mixed state that samples $r$ and stores $\ket{r,C(r)}$.
By linearity, any simulator's average behavior on the pure implementations $\ket{\psi_s}$ is the same as its behavior on this mixed implementation.
But this mixed implementation is also the average of the basis-state programs $\ket{r,C(r)}$, each of which is a pure implementation of the constant function with value $C(r)$.
Consequently, the same averaging argument used in the impossibility still forces the EPR test to pass on the compiled functionality in the lossy case, while correctness forces it to fail with constant probability in the injective case.
The test would therefore still distinguish the two cases of Statistical Difference.
Bare purity therefore does not give a meaningful escape from the unrestricted impossibility.

There is also a representation issue: if a compiler outputs a mixed state, there is no canonical decomposition of that mixed state into pure states.
A definition that asks whether the ``pure components'' are correct therefore depends on a representation that is not part of the physical object.

These two issues motivate the notion of \emph{testability}.
Informally, a testable quantum program consists of a pure program state $\ket{\psi}$ (which can be interpreted as the input to a universal quantum circuit), and a reflection program for $\ket{\psi}$, i.e.\ it implements the unitary $I - 2\ketbra\psi$.

\begin{definition}[Informal; \Cref{def:testable-prog}]
    A quantum program is \emph{testable} if its program state is a pure state $\ket{\psi}$ accompanied by a reflection program for $\ket{\psi}$.
    A one-time compiler is testable if its outputs are testable programs and every $\ket\psi$ in the support one-time implements the input functionality.
\end{definition}

The reflection program in the testability definition gives a physical and well-defined way to pin down the pure program state output by the compiler.
Crucially, testability is not just purity: changing the hidden phases in \eqref{eq:C-alternative-decomposition} also changes the reflection interface.
Thus the random-phase averaging trick is no longer merely a different decomposition of the same program object.
This is the extra structure that breaks the impossibility argument and lets the simulator coherently recognize whether the program has returned to its designated initial state.

We highlight that testability turns out to be not artificial from the perspective of known constructions of one-time compilers.
There are concrete and recurring ways to modify existing one-time compiler constructions to become testable.
For starters, a trivial case is the classical-state case: if each issued program is individually correct and $\ket{\tilde P}$ is always a computational-basis string, then one can first read out that string (equivalently, measure in the computational basis) and then efficiently implement the exact reflection $I-2\ketbra{\tilde P}$ coherently via an equality check and phase kickback.
The most common pattern for designing a one-time compiler is to issue an uncloneable token that collapses upon use, along with an obfuscated classical program that only works when the two-basis measurement outcome on the token is correct.
For subspace and coset states, coherent membership tests in complementary bases suffice to implement reflection about the designated token~\cite{BDS23-tokens,GLRRV25-otp}.
In particular, both the plain-model coset-state construction of~\cite{GLRRV25-otp} and the construction of~\cite{GM24-otp} use a distinguished abort output, making these acceptance tests accessible through the evaluator.
Thus, these existing constructions admit testable variants preserving their established correctness and security guarantees.

We further conjecture that the program implicitly having a reflection oracle might be an unavoidable property for many natural function classes, since outputs such as signatures and one-time NIZKs admit efficient verification.
The SEQ interface becomes testable as soon as we know an efficiently preparable state that is orthogonal to every channel output.
This structure can be generically added by extending the query register with a success bit whose incoming value is ignored and whose output value is always set to 1.
We refer the readers to \Cref{sec:seq-recognizable-outputs} for more details.

\subsection{One-Time Security for Quantum Functionalities}
\label{sec:intro-quantum-functionality}
Before we investigate best-possible one-time compilers in the restricted testable setting, we first turn to extending the study of one-time security from classical sampling functionalities to all quantum (channel) functionalities.
As far as we are aware, our work is the first to study positive results for constructing one-time programs for arbitrary quantum channels without hardware assumptions.

The work of Gupte et al.\ \cite{GLRRV25-otp} proposed a simulation-based definition called Single Effective Query (SEQ) security for one-time programs for classical sampling functionalities, and showed that it is achievable in the classical oracle model.
The core idea is to require the adversary's view to be simulatable with only an ideal query interface for the protected functionality $\Psi$.
This interface permits computations and uncomputations that have not extracted information, but it tries its best to detect when the output is meaningfully disturbed and subsequently refuse to cooperate with a subsequent query.
This is in contrast to older simulation-based one-time security notions such as in \cite{GKR08-otp,BGS13-otp} where the interface allows literally one physical query, regardless what has been learned in the query.

In this section, we recast this idea in a simpler form that also applies to arbitrary quantum channels.
For the rest of this work, we call our generalized notion simply SEQ and refer to the original version from \cite{GLRRV25-otp} as classical SEQ (CSEQ) to distinguish.

Without hardware assumptions, an adversary could always compute and uncompute the unitary implementation of the interface.
This motivates the following definition of a ``self-adjoint implementation'' of a quantum channel.

\begin{definition}[Informal; \Cref{def:CSEQ-oracle}]
Consider a quantum channel $\Psi$ using input/output register $\calX$ which is implemented by some unitary $U_\Psi$ which might use a private auxiliary register $\calA$ that is initialized to 0.
Let $\calC$ be a qubit register acting as a counter.
Let the \emph{self-adjoint implementation} $S(U_\Psi)$ of $\Psi$ be the unitary that swaps
\[
\ket{x}_\calX\ket0_\calA \otimes \ket0_\calC \leftrightarrow (U_\Psi\ket{x}_\calX\ket0_\calA) \otimes \ket1_\calC
\]
and acts as identity everywhere else.
\end{definition}

\begin{fact}
    $S(U_\Psi)$ is a well-defined self-adjoint unitary and can be efficiently implemented given controlled query access to $U_\Psi$.
\end{fact}

\Cref{fig:seq-interface} gives this implementation and separates the public query register from the persistent private registers.
Its middle gate checks whether uncomputation has restored the auxiliary register before allowing another computation.

\begin{figure}[htbp]
\centering
\begin{tikzpicture}[x=\linewidth,y=1cm,
    wire/.style={semithick},
    gate/.style={draw,fill=white,minimum width=1.2cm,minimum height=1.45cm},
    control/.style={circle,fill=black,inner sep=1.5pt}]
    \node[anchor=east,font=\small] at (.20,.9) {public $\calX$};
    \node[anchor=east,font=\small] at (.20,0) {private $\calA$};
    \node[anchor=east,font=\small] at (.20,-.9) {private $\calC$};
    \foreach \y in {.9,0,-.9} {\draw[wire] (.22,\y) -- (.97,\y);}
    \node[gate] (udag) at (.35,.45) {$U_\Psi^\dagger$};
    \node[gate] (u) at (.85,.45) {$U_\Psi$};
    \node[control] (c1) at (.35,-.9) {};
    \node[control] (c2) at (.85,-.9) {};
    \draw[wire] (c1) -- (udag.south);
    \draw[wire] (c2) -- (u.south);
    \draw[wire,fill=white] (.60,0) circle[radius=2.5pt];
    \draw[wire] (.60,-.06) -- (.60,-.9);
    \coordinate (flip) at (.60,-.9);
    \draw[wire,fill=white] (flip) circle[radius=5pt];
    \draw[wire] ([xshift=-5pt]flip) -- ([xshift=5pt]flip);
    \draw[wire] ([yshift=-5pt]flip) -- ([yshift=5pt]flip);
    \node[font=\small,above] at (.60,.12) {$\calA=\ket{0}$};
    \node[font=\small,below] at (.60,-1.15) {flip the computed flag};
\end{tikzpicture}
\caption{The self-adjoint SEQ interface $S(U_\Psi)$. The caller controls only $\calX$; $\calA$ and $\calC$ persist privately between calls. The middle gate flips $\calC$ exactly when the inverse computation has restored $\calA$ to zero.}
\label{fig:seq-interface}
\end{figure}

\begin{definition}[Informal; \Cref{def:CSEQ-security}]
    A one-time compiler is SEQ-secure if  $ \OTP(P)$ can be efficiently (computationally indistinguishably) simulated by a simulator that only gets access to $P$ via querying the self-adjoint implementation for $P$ without having access to its $\calA, \calC$ registers.
\end{definition}

It is so called ``Single Effective Query'' because (1) to compute $\Psi$ a second time, you must first return to the state $\ket{x}_\calX\ket{0}_\calA\otimes \ket{0}_\calC$ by uncomputing it; but more importantly, (2) if the output is ``meaningfully'' disturbed, then further evaluations become effectively impossible because the evaluator cannot return to the initial state $\ket{x}_\calX\ket{0}_\calA\otimes \ket{0}_\calC$.

A helpful example is a channel that ignores the input and outputs $n$ uniformly random bits.
An untouched output can be uncomputed, but measuring all $n$ bits breaks the coherence with the hidden purification.
For large $n$, the next query then acts very close to identity.\footnote{In the canonical purification, the computed state contains $2^{-n/2}\sum_r\ket r_{\mathrm{out}}\ket r_{\mathrm{private}}$.
After measuring the output, its squared projection onto the computed subspace is $2^{-n}$, so the next SEQ query changes the state by trace distance $O(2^{-n/2})$.
This also matches the intuition that the higher the entropy of the measured output, the more destructive the effect on the quantum program.}

\begin{theorem}[Informal; \Cref{thm:cseq-seq-equivalence}]
\label{thm:cseq-seq-equivalence-informal}
    For every non-trivial randomized classical functionality $f:\calX\times\calR\to\calY$ satisfying the nontriviality condition that $|\calX|>1$, there is a quantum channel $\Phi_f$ such that the classical query interface $O_f^{\CSEQ}$ and the generalized query interface $O^{\SEQ}_{\Phi_f}$ are efficiently inter-simulatable.
    Consequently, classical \CSEQ{} security is a special case of generalized \SEQ{} security.
\end{theorem}
Thus our generalized SEQ definition should be read as extending the prior CSEQ framework rather than replacing it with an incomparable notion.

As the prior work \cite{GLRRV25-otp} has already shown, SEQ as a simulation-based security is unachievable for certain randomized circuits in the plain model by generalizing the impossibility of VBB \cite{BGIRSV01-obfuscation}.
Despite this VBB-style impossibility, we believe that SEQ on its own is still a very useful notion for the topic of one-time programs.
Of course, the VBB-style impossibility does not rule out SEQ for \emph{all} interesting functionalities; but more importantly, similar to the role VBB plays today, SEQ provides a relatively simple yet \emph{realistic ceiling} on the one-time security that one could hope for in the plain model for a given functionality using \emph{testable} OTP obfuscators.

\subsection{SEQ As Best-Possible Testable Security}
\label{sec:intro-seq-best-possible}

Next, we study best-possible security relative to the class of testable one-time compilers.
Similarly, we define a one-time compiler to be \emph{best possible among testable programs} for a class of quantum channels if whatever can be learned from its output can be simulated given any testable one-time implementation of the same channel.
Here the compiler being constructed need not itself be testable.

Our main positive theorem is that this simulation notion is best possible for the testable comparison class.

\begin{theorem}[Informal; \Cref{thm:cseq-testable-bpotp}]
    Any one-time compiler satisfying \SEQ{} security for a class of quantum channels is best possible among testable one-time programs for that class.
\end{theorem}

\paragraph{Recovering Prior Security Guarantees.}
Together, our two simulation theorems give SEQ a unifying role among the known approaches to one-time security.
First, the CSEQ guarantee of~\cite{GLRRV25-otp} is recovered directly as a special case of SEQ by \Cref{thm:cseq-seq-equivalence-informal}.
The other one-time constructions \cite{BDS23-tokens,GLRRV25-otp,GM24-otp} admit testable implementations through their existing acceptance checks.
Thus a one-time compiler satisfying SEQ security would also straightforwardly achieve the security guarantees from these works.
Taking a different perspective, our theorem also gives a way to establish new applications for SEQ security: it suffices to exhibit a testable implementation that provides the desired information-hiding guarantee, and any SEQ-secure compiler for the same channel inherits that guarantee.

For testable quantum one-time programs, SEQ is therefore a ceiling on achievable privacy, including that of future constructions.
This is analogous to how VBB provides a ceiling to security achievable by iO in the many-time obfuscation setting.
Complementarily, if one aims to construct a one-time compiler that achieves security unachievable by even SEQ, the OTP construction \emph{must actively} prevent its output program from being testable.

\subsection{Constructing SEQ-Secure Programs}
\label{sec:intro-constructions}

We next show that the SEQ benchmark is achievable in idealized models, generalizing \cite{GLRRV25-otp} from sampling functionalities to arbitrary non-oracular quantum functionalities.

Since we have already established that SEQ is sufficient for best-possible testable security, a natural high level strategy to construct one-time compilers with this security is to simply (many-time) obfuscate the SEQ interface and let the SEQ interface handle one-time security.
Since the SEQ interface maintains hidden registers across queries, the main tool our construction uses is a new notion called stateful quantum obfuscation.
We formalize a notion of stateful quantum indistinguishability obfuscation (iO).
Two stateful quantum programs are equivalent if they induce the same behavior over every sequence of evaluations.
Stateful quantum iO asks that equivalent stateful programs have indistinguishable obfuscations.

\begin{theorem}[Informal; \Cref{thm:siO-implies-bptotp}]
    Assuming stateful quantum indistinguishability obfuscation, there exists a one-time compiler for all quantum channels that is best possible among testable programs\footnote{As usual for positive results on oracle-model obfuscation, the oracle is an implementation aid for the compiler and evaluator.
Thus, when we say ``all quantum channels'' in the oracle model, we mean all non-oracular quantum channels; allowing the protected channel itself to query the same oracle would reintroduce the usual VBB-style impossibilities for obfuscating oracle-aided functionalities \cite{BGIRSV01-obfuscation,GLRRV25-otp}.
For comparison, our impossibility results for non-oracular sampling functionalities hold relative to any unitary oracle $O$, with $\mathsf{BQP}$ replaced by $\mathsf{BQP}^{O}$ in the hardness assumptions.
The oracle used here is efficiently simulatable, so the new assumptions are equivalent to the non-oracular assumptions.}.
\end{theorem}

We also give evidence that this is a reasonable idealized route by constructing stateful quantum obfuscation in an oracle model.

\begin{theorem}[Informal; \Cref{thm:ideal-sqO,cor:seq-otp-classical}]
    Ideal stateful quantum obfuscation exists in the classical oracle model.
    Consequently, there exist \SEQ-secure one-time programs for all non-oracular quantum channels in the classical oracle model.
\end{theorem}

This gives, to our knowledge, the first positive results for one-time programs for arbitrary non-oracular quantum channels without hardware assumptions, albeit in the classical oracle model.
It also suggests an approach for a plain-model instantiation of best-possible testable one-time programs by constructing stateful quantum obfuscation in the plain model.

\begin{figure}[pt]
\centering
\small
\setlength{\tabcolsep}{4pt}
\begin{tabular}{|>{\raggedright\arraybackslash}p{0.30\linewidth}|>{\raggedright\arraybackslash}p{0.30\linewidth}|>{\raggedright\arraybackslash}p{0.30\linewidth}|}
\hline
\textbf{Security notion} & \textbf{Impossibilities} & \textbf{Constructions} \\
\hline
Single physical query, simulation-based, OTP for quantum functionalities \cite{BGS13-otp}
&
Strong impossibility for most functionalities relative to all oracles
&
For single physical-query learnable (trivial) functions only \cite{BGS13-otp};
For constant-distribution functionalities (implicit in this work)
\\
\hline
Classical SEQ, simulation-based, OTP for classical functionalities \cite{GLRRV25-otp}
&
VBB-style impossibility for contrived functionalities in the plain model \cite{GLRRV25-otp}
&
For all classical functions, relative to a classical oracle \cite{GLRRV25-otp}
\\
\cline{1-1}\cline{3-3}
Generalized SEQ, simulation-based, OTP for quantum functionalities (\textcolor{red}{this work})
&
 
&
For all quantum functionalities relative to a classical oracle (\textcolor{red}{this work})
\\
\hline
Best-possible OTP compiler (\textcolor{red}{this work})
&
Impossibility for generic compilers under our zero-knowledge hardness assumptions (\textcolor{red}{this work})
&
 
\\
\cline{1-2}\cline{3-3}
Best-possible OTP compiler among only testable programs (\textcolor{red}{this work})
&
 
&
From either SEQ security or quantum stateful iO (\textcolor{red}{this work})
\\
\hline
\end{tabular}
\caption{One-time program security notions with impossibilities and constructions.}
\label{fig:definitions_impossibilities_figure}
\end{figure}

\subsection{Technical Overview of Positive Results}
\label{sec:intro-positive-overview}

We now give a high-level overview of the main arguments for the positive results in the last two sections.

\paragraph{Best-Possible Testable One-Time Programs.}
We first sketch that a one-time program achieving SEQ security is as secure as \emph{any} testable one-time program.

To show that SEQ security is best-possible for testable one-time programs, we must show how to simulate $\SEQ(f)$ using an alternative quantum program $g$ which is equivalent to $f$ for one query and has a reflection oracle $R$.
The idea of the simulator is straightforward: we simply implement $\SEQ(f)$ by implementing $\SEQ(g)$ instead.
The intuition is that the $\SEQ$ interface would hide the implementation details of whether $f$ or $g$ is being queried since they are one-time equivalent.

To formalize our intuition above, we establish the following lemma.
\begin{lemma}[Informal; \Cref{lem:cseq-sim-equals-cseq-canonical}]
\label{lem:informal-self-adjoint-hiding}
    For any two unitary implementations $U_\Psi, U'_\Psi$ of the same channel $\Psi$, $S(U_\Psi)$ and $S(U'_\Psi)$ are perfectly indistinguishable given unbounded query access to only the input/output register $\calX$.
\end{lemma}
Thus this lemma establishes that the self-adjoint implementation is a generic compiler to make quantum channel evaluation reversible without any dependence on the implementation of the channel itself.

This already suffices to prove the statement if $f$ and $g$ are classically described programs.
However, there is one subtlety we need to address, which is that $g$ may be a quantum state describing a program $\ket g$.
After evaluating it once, the program state may be disturbed.
Unlike classically, we cannot say that the program state is readonly due to no cloning.

This is where the reflection oracle $R$ plays a crucial role: the reflection oracle $R$ well defines the quantum program $\ket g$.
In particular, we modify the middle gate in \Cref{fig:seq-interface} to also control on the quantum program register being $\ket g$, which can be efficiently implemented using the reflection oracle.
The intuitive reason that this works is that this way, we are effectively implementing the following quantum program that does \emph{not} involve a quantum auxiliary input $\ket g$ and yet is perfectly equivalent to $g$:
\begin{enumerate}
    \item Perform a swap $\ket 0 \leftrightarrow \ket g$ in the program register.
    \item Carry out the original computation.
\end{enumerate}
Observe that the $\SEQ$ for this program would reflect around all zeroes on both the auxiliary register and the program register, which is equivalent to reflecting around all zeroes on the auxiliary register and $\ket g$ on the program register if we do not do the first step.

\paragraph{Achieving SEQ Security with Oracles.}
Note that our \Cref{lem:informal-self-adjoint-hiding} gives a straightforward way for instantiating SEQ-secure one-time programs with \emph{stateful unitary} ideal obfuscation.
In such an obfuscation scheme, the obfuscator takes as input a unitary $U$ that acts on a public register $\calX$ and an internal auxiliary register $\calA$ that is initialized to some state that is given; and it outputs $\tilde U$ which can only be queried in either the forward or backward direction where the internal register is inaccessible.
Then to one-time protect $f$, we can simply use the ideal obfuscation to obfuscate and give out $\SEQ(f)$.

Intuitively, such a stateful unitary ideal obfuscation scheme can be implemented in a (stateless) unitary oracle model, or by ideal unitary obfuscation.
The idea is that we can use a quantum authentication scheme to protect the internal state so that we can delegate the internal state management to the adversary without compromising security.
As shown in \Cref{fig:stateful-obfuscation}, the recipient holds the authenticated memory, while an obfuscated update circuit holds the authentication key.
On valid memory, the circuit coherently decodes, applies the SEQ update, and re-authenticates.
On invalid memory, it undoes the authentication check and acts as identity; both branches remain coherent so that the entire update is unitary.
Furthermore, such a scheme can be ported to the classical oracle model using the compiler of \cite{HT25-unitary}.
While this compiler introduces quantum states into the program, this does not matter for our application since our obfuscated programs are allowed to have quantum states.

\begin{figure}[htbp]
\centering
\begin{tikzpicture}[x=\linewidth,y=1cm,
    box/.style={draw,fill=white,align=center,font=\small,inner sep=6pt},
    arrow/.style={->,>=stealth,semithick},
    annotation/.style={font=\small,align=center}]
    \node[annotation,anchor=west] at (.02,4.1)
        {(a) The ideal interface keeps its memory between calls};
    \node[box,minimum height=1.5cm,text width=.18\linewidth] (ideal) at (.5,2.75)
        {$S(U_\Psi)$};
    \node[annotation] (qi) at (.10,3.2) {public $\calX_i$};
    \node[annotation] (qo) at (.90,3.2) {public $\calX_{i+1}$};
    \node[annotation] (si) at (.10,2.3) {private $\calS_i$};
    \node[annotation] (so) at (.90,2.3) {private $\calS_{i+1}$};
    \draw[arrow] (qi) -- ([yshift=.45cm]ideal.west);
    \draw[arrow] ([yshift=.45cm]ideal.east) -- (qo);
    \draw[arrow] (si) -- ([yshift=-.45cm]ideal.west);
    \draw[arrow] ([yshift=-.45cm]ideal.east) -- (so);
    \node[annotation] at (.5,1.65) {$\calS=(\calA,\calC)$ is the persistent private memory};

    \node[annotation,anchor=west] at (.02,.8)
        {(b) In the real world, the recipient holds authenticated memory};
    \node[annotation] (input) at (.075,-.05) {$(\calX,\widetilde{\calS})$};
    \node[box,text width=.13\linewidth] (dec) at (.29,-.05) {$\QAS.\Dec_k$};
    \node[box,text width=.15\linewidth] (update) at (.51,-.05) {$S(U_\Psi)$};
    \node[box,text width=.13\linewidth] (enc) at (.73,-.05) {$\QAS.\Enc_k$};
    \node[annotation] (output) at (.945,-.05) {$(\calX',\widetilde{\calS}')$};
    \draw[arrow] (input) -- (dec);
    \draw[arrow] (dec) -- node[above,font=\scriptsize] {valid} (update);
    \draw[arrow] (update) -- (enc);
    \draw[arrow] (enc) -- (output);
    \draw[arrow] (dec.south) -- (.29,-1.0) -- (.87,-1.0) -- (.87,-.05);
    \node[annotation,fill=white,inner sep=2pt] at (.58,-1.0)
        {invalid: undo the check, act as $I$};
    \draw[dashed] (.195,.48) rectangle (.885,-1.35);
    \node[annotation] at (.5,-1.85) {obfuscate the coherent update unitary with $k$ hardwired};
\end{tikzpicture}
\caption{Implementing the stateful SEQ interface. The compiler authenticates the initial memory and gives out an obfuscation of the update unitary. Authentication protects the memory held by the recipient, while ideal unitary obfuscation protects the key. Both branches of the authentication check are coherent. Ideal stateful obfuscation yields SEQ security; stateful iO suffices for best-possible security among testable programs.}
\label{fig:stateful-obfuscation}
\end{figure}

\paragraph{Towards the Plain Model via Stateful Quantum Obfuscation.} 
Even though SEQ security is achievable using ideal obfuscation or classical oracles, we already know that there are randomized classical functionalities that cannot be SEQ-securely one-time protected in the plain model \cite{GLRRV25-otp}.
A natural question then is whether we can nevertheless achieve best-possible testable one-time programs in the plain model.
Perhaps this is achievable through an alternative notion of obfuscation which plausibly exists in the plain model, similar to how for many-time security, iO and best-possible obfuscation are equivalent.

Previous works propose natural definitions for indistinguishability obfuscation of quantum circuits that compute either pseudodeterministic functions~\cite{BBV24-qso, CG24-qio} or unitaries~\cite{HT25-unitary} --- for every two quantum programs that compute the same pseudodeterministic function (or unitary), their obfuscations should be indistinguishable. A natural attempt at generalizing this to arbitrary quantum \emph{channels} would be to ask that if two quantum programs output similar mixed states, their obfuscations should look indistinguishable. However, if two programs produce similar outputs for the first evaluation, but then degrade so that subsequent evaluations are distinguishable, it should be hard to obfuscate them into programs that are indistinguishable. In fact, our previously discussed impossibility result formalizes this intuition.

This is a uniquely quantum issue for two reasons. First, quantum programs with auxiliary quantum states inherently maintain state over evaluations that involve non-gentle measurements. Second, quantum circuits have the ability to self-produce randomness, and there is no general way to ``de-randomize'' them, like in the classical case. Classical descriptions of circuits do not degrade with evaluations, and it is always possible to de-randomize a probabilistic circuit $C(x; r)$ and obfuscate $C(x; F_k(x))$, so that the obfuscations of two equivalent classical probabilistic circuits $C_1, C_2$ will always produce the same output across evaluations.
To avoid the impossibility, we explicitly consider obfuscating programs whose state may change over time.

\begin{definition}[Stateful iO, Informal; \Cref{def:stateful-io}]
    Stateful indistinguishability obfuscation allows two stateful quantum programs to be indistinguishably obfuscated if and only if they produce the same output distributions over \emph{any} sequence of evaluations.
\end{definition}

To show that stateful iO evades our impossibility result, we show that it is implied by ideal unitary obfuscation. Intuitively, a unitary program can authenticate a private register to itself in order to maintain state across multiple queries. Using ideal obfuscation, the authentication key is protected, so the register remains private from the evaluator.
Similarly as before, we can further port this to the classical oracle model using the compiler of \cite{HT25-unitary}.

Furthermore, unlike ideal obfuscation, there does not seem to be an inherent barrier to constructing stateful iO in the plain model.
Because the output distributions are required to be close even for query sequences that can depend on the programs being obfuscated, stateful iO avoids the self-referential techniques used to rule out ideal obfuscation~\cite{BGIRSV01-obfuscation}.

Finally, we show that stateful iO implies best-possible testable one-time programs. Thus, if one could construct stateful iO in the plain model, they would also construct best-possible testable one-time programs in the plain model.
This again uses \Cref{lem:informal-self-adjoint-hiding}, which when combined with the security guarantee of stateful iO, shows that whatever could be learned from the stateful obfuscation of $\mathsf{SEQ}(f)$ could also be learned from $g$ by statefully obfuscating $\mathsf{SEQ}(g)$.

\subsection{Discussion and Future Directions}
\label{sec:intro-discussion}

\paragraph{Prevalence of Our Best-Possible OTP Impossibility.}
Beyond the deterministic case, our best-possible impossibility identifies two classes of randomized functionalities that no single best-possible compiler can handle simultaneously.
A different idea for getting around the impossibility might be to restrict the functionality class to exclude one of these classes (this is orthogonal to restricting the comparison class of one-time compilers as we've done for testable one-time compilers).
Inspecting our impossibility further, the lossy class is probably unlikely to be of interest for one-time programs, where the randomized program samples the same distribution for every input.

While this could be a valid way of getting around the impossibility, we argue that this direction seems to require restricting the functionality class to be very small.
For starters, consider the randomized functionality $f(x; r) \to (y_1, y_2)$ where $y_1$ is sampled from a fixed distribution independent of $x$.
Then, this is a functionality that is not lossy yet (a suitably adapted version of) our impossibility still applies.
One could even further consider more involved variants of this where the lossy structure is even less apparent, such as applying a pseudorandom permutation to the output.
Given these examples, we suspect that it is unlikely to identify a meaningful and large class of genuinely randomized functionalities where best-possible OTPs are possible, since any class that is closed under augmentation of such structures is also susceptible to our impossibility.

\paragraph{Implications for Defining Obfuscation of Quantum Programs.}
Previous work propose natural definitions for indistinguishability obfuscation for quantum circuits that compute pseudodeterministic classical functions~\cite{BBV24-qso, CG24-qio} and unitaries~\cite{HT25-unitary} --- for every two quantum programs that compute the same pseudodeterministic function (or unitary), their obfuscations should be indistinguishable.
A natural attempt at generalizing this to quantum \textit{channels} would be to ask that if two quantum programs output (approximately) the same mixed state, their obfuscations should look indistinguishable.
Note that this natural attempt only asks the two programs to produce samples from same mixed state for the \emph{first} evaluation.
This is because for (approximately) classical functions or more generally unitaries, one-time equivalence is equivalent to many-time equivalence by gentle measurement.
However, when the functionality is not unitary and the program is described by a program state, the program state itself may evolve after the first query.
Crucially, we can rule out this one-time security notion by invoking our best-possible OTP impossibility, since intuitively if such an object exists, then it would give a best-possible OTP similar to how iO is equivalent to best-possible obfuscation.
We formalize this in \Cref{sec:sampling-qsio}.

Our definition of stateful iO gets around this problem by considering \emph{many-time} equivalence of two programs rather than \emph{one-time} equivalence.
Thus, when defining (many-time) obfuscation for quantum sampling programs or general quantum channels, a notion which takes into account behavior on sequential queries, like stateful iO, seems necessary.

\paragraph{Future Directions.}
Our results leave several directions open.
The most direct is to find plain-model one-time compilers with best-possible testable security, perhaps through a candidate for stateful quantum obfuscation.
On the (many-time) obfuscation front, does unitary indistinguishability obfuscation imply stateful iO?
The main difficulty in lifting our proof from ideal obfuscation is to adapt the argument that the obfuscator protects the authentication key.
Furthermore, does classical indistinguishability obfuscation imply unitary iO?

Another direction is to explore security beyond those achievable by testable compilers, in other words, to explore \emph{untestable} one-time program constructions.
Our impossibility already constructs a one-time VBB obfuscator for lossy-mode encryptors or constant-distribution samplers.
It is not hard to see that for these functionalities, testable implementations would yield strictly worse security.
The SEQ interface can also offer more privacy than testable implementations for certain functionalities (\Cref{prop:seq-testability-gap}), but extending the query register with a bit that is overwritten with a fixed success value generically makes the interface testable (\Cref{sec:seq-recognizable-outputs}).

To further motivate one-time program security beyond testable compilers, we present a generic attack against any testable one-time program $f$.
The attack takes a polynomial-size list of input--observable pairs $(x_i,O_i)_{i=1}^N$, where each input state is efficiently preparable and each observable is efficiently measurable with $\|O_i\|_\infty\leq 1$.
It estimates $\Tr(O_i f(x_i))$ for every $i$ up to inverse polynomial error in expected polynomial time.
(For example, we can estimate how often specified output bits are 0 on selected inputs.)
While this attack is not always useful (say for forging multiple signatures against a signature token), it allows a user to estimate output statistics on many chosen inputs from a single testable program.

This attack follows from the single-copy tomography result of Farhi et al.~\cite[Theorem~2]{FGHLNS10-restoration}.
For each input--observable pair, evaluation followed by the binary output measurement with effect $(I+O_i)/2$ defines an efficiently implementable POVM on the program state.
Their algorithm estimates its acceptance probability while restoring the program state using the reflection oracle, allowing the same copy to be reused for all pairs.

\ifllncs\else

\ifllncs
\section{One-Time Security in Previous Works}
\else
\subsection{One-Time Security in Previous Works}
\fi
\label{sec:more_discussions_prior_defs}

For completeness, we survey one-time security notions studied in prior works in this section.

\begin{itemize}

\item The work by Gupte et al.\ \cite{GLRRV25-otp} proposes a notion called ``single effective query (SEQ) simulation-based one-time security'' that circumvents the impossibility results (without hardware assumptions) of the
simulation-based one-time security defined in preceding works \cite{GKR08-otp,BJ15-qhe}. In this work, we revise and generalize their SEQ definition to all functionalities and, to distinguish, we refer to their original (classical) version as CSEQ.

More concretely, the CSEQ one-time security notion requires that the adversary's view after maliciously utilizing the one-time program of $f$ can be simulated by only querying a restricted stateful query interface to $f$.
        Such an interface attempts to record prior queries made to $f$ and only answers the query if no prior queries are recorded.

        They show that this notion is achievable for \emph{all} functionalities in the classical oracle model, which is both good and also unsatisfactory.
        It is good because it is a single notion that captures all functionalities.
        It is however unsatisfactory because it delegates the problem of identifying what level of security the definition guarantees to the problem of what can be learned through the CSEQ interface.
        For example, the CSEQ interface would allow an unlimited number of evaluations for a deterministic functionality thus keeping consistency with the impossibility result \cite{BGS13-otp}.
        However, for a general randomized functionality, it is not clear what can be learned through the CSEQ interface since the interface is somewhat complicated.
        Our SEQ notion somewhat mitigates this issue since our SEQ security is much simpler.

    \item Apart from CSEQ security, \cite{GLRRV25-otp} presents several security notions in the plain model, but these notions are restricted to specific settings or target specific applications.
    These notions can be viewed as extensions of 
    the security requirement for one-time signature tokens in \cite{BDS23-tokens} to more functionalities.
    \cite{BDS23-tokens} states that it is impossible to produce two valid signatures for two distinct messages with respect to the signature verification algorithm.
    This security can be generalized to any unforgeable functionalities such as one-time NIZKs \cite{GLRRV25-otp}.
    Similarly, the security of one-time PRFs (where part of the PRF input is sampled randomly) \cite{GLRRV25-otp} only states that it is impossible to simultaneously distinguish the outputs on two different inputs from random.
    \item The work by Gunn and Movassagh \cite{GM24-otp} showed that for their construction, a somewhat natural class of attackers (including the attack in \cite{BGS13-otp}) cannot even produce a second output as long as the functionality samples a high min entropy output on every input.
        However, it is possible that a more malicious attacker can still produce multiple outputs: an example of this can be seen from the plain-model impossibility in \cite[Theorem~7.7]{GLRRV25-otp}.
\end{itemize}

These prior results are covered by the two simulation arguments in \Cref{sec:intro-quantum-functionality,sec:intro-seq-best-possible}.
CSEQ is recovered as a special case of SEQ, while the other constructions discussed here admit testable implementations using their acceptance checks.
Best-possible testable security then guarantees that an SEQ-secure program reveals no more information than those implementations of the same channel, in their respective models.

\fi

\let\subsection\bpotpSubsection
\else

\section{Preliminaries}
\subsection{Quantum Computation}

We fix canonical description formats for all objects. For a classical randomized function $f$, we write $|f|$ for the bit-length of its canonical description (excluding any oracle access it might carry). For a quantum sampling program $P=(\rho, C)$, we write $|P|$ for the length of a canonical description consisting of a description of the state $\rho$ together with the description of the (possibly oracle-aided) circuit $C$; oracle access is not counted toward length. For a stateful quantum program $(U, \ket{\psi})$, we similarly write $|(U, \ket{\psi})|$ for the length of a canonical description of $U$ together with a description of $\ket{\psi}$.

When comparing programs by length (e.g., in indistinguishability obfuscation), we allow \emph{padding} to equalize lengths without changing behavior: padding may add ancilla qubits initialized to $\ket{0}$ and insert no-op gates that leave the program’s input–output behavior unchanged.

When we say a result relativizes to any unitary oracle, we mean that it holds in the model where all parties receive query access to a family of unitaries including inverse and control as is standard \cite{Zhandry25-unitary}.

We will need a few facts about the trace distance of two pure states whose amplitudes are defined by classical probability distributions.
\begin{definition}[Hellinger Distance] \label{def:hellinger-distance}
For two probability density functions $f, g$ over a finite domain $\mathcal{X}$, the squared Hellinger distance between $f$ and $g$ is defined as
\begin{align*}
    H^2(f, g) = \frac{1}{2} \sum_{x \in \mathcal{X}} \left(\sqrt{f(x)} - \sqrt{g(x)}\right)^2 = 1 - \sum_{x \in \mathcal{X}} \sqrt{f(x) g(x)}.
\end{align*}
\end{definition}

\begin{lemma}\label{lemma:hellinger-TV}
    Let $D_1, D_2$ be two probability density functions over a finite domain $\mathcal{X}$, then
    \begin{align*}
        H^2(D_1, D_2) \le \TV(D_1, D_2) \le \sqrt{2}\,H(D_1, D_2).
    \end{align*}
\end{lemma}

\begin{lemma}\label{lemma:hellinger-trace-distance}
    Let $\mathcal{X}$ be a finite set and let $D_1, D_2$ be probability densities on $\mathcal{X}$. Let
    \begin{align*}
        \ket{\psi_1} = \sum_{x\in \mathcal{X}} \sqrt{D_1(x)} \ket{x} \quad \text{ and } \quad \ket{\psi_2} = \sum_{x \in \mathcal{X}} \sqrt{D_2(x)}\ket{x}\;.
    \end{align*}
    Then, the trace distance between the pure states $\ket{\psi_1}, \ket{\psi_2}$ is
    \begin{align*}
        T(\ket{\psi_1} \bra{\psi_1}, \ket{\psi_2} \bra{\psi_2}) = \sqrt{1 -  (1 - H^2(D_1, D_2))^2}.
    \end{align*}
\end{lemma}

\subsection{Quantum Authentication Schemes}

\begin{definition}[Quantum Authentication Scheme]
    A \textbf{quantum authentication scheme} is a tuple of QPT algorithms $(\KeyGen, \Enc, \Dec, \Ver)$ with the following behavior.
    \begin{itemize}
        \item $k\gets \KeyGen(1^\secpar, n)$ takes as input the security parameter $\secpar$ and a length $n\in \bbN$, then outputs a key $k$. Here, $n$ is the length of the state to be authenticated.
        \item $\ket{\psi'} \gets \Enc_k(\ket{\psi})$ takes as input a key $k$ and an $n$-qubit state $\ket{\psi}$, then outputs another state $\ket{\psi'}$.
        \item $\ket{\psi''}/\bot \gets \Dec_k(\ket{\psi'})$ takes as input a key $k$ and a state $\ket{\psi'}$, then outputs another state $\ket{\psi''}$ or $\bot$.
        \item $\Accept/\Reject \gets \Ver_k(\ket{\psi'})$ takes as input a key $k$ and a state $\ket{\psi'}$, then outputs $\Accept$ (accept) or $\Reject$ (reject). $\Ver_k$'s behavior is identical to $\Dec_k$'s, except $\Ver_k$ outputs $\Accept$ whenever $\Dec_k$ would output $\ket{\psi}$ and outputs $\Reject$ whenever $\Dec_k$ would output $\bot$.
    \end{itemize}
    A quantum authentication scheme must satisfy correctness and security properties.
    \begin{itemize}
        \item \textbf{Correctness.} For every $k$ in the support of $\KeyGen$ and every state $\ket{\psi}$,
        \[
            \Dec_k(\Enc_k(\ket{\psi})) = \ket{\psi}
        \]

        \item \textbf{Security.} For every QPT adversary $\adv$, there exists an $\epsilon(\secpar)\in[0,1]$ such that for every $\ket{\psi}$,
        \[
            \left\{\Dec_k(\rho) \,\middle\vert \begin{array}{c}
                 k\gets \KeyGen(1^\secpar)  \\
                 \rho \gets \adv(\Enc_k(\ket{\psi})) 
            \end{array}\right\} 
            \approx 
            (1-\epsilon)\{\ket{\psi}\} + \epsilon\{\bot\}
        \]
        If this holds given oracle access to $\Ver_k$ (or if there is a public key $\mathsf{vk}$ that allows implementing $\Ver_k$ and it holds given access to $\mathsf{vk}$, we say that the scheme is publicly verifiable.
        
    \end{itemize}
\end{definition}

\cite{BBV24-qso} shows how to construct publicly verifiable quantum authentication using coset states.

\subsection{Ideal Obfuscation of Unitaries}

In this section, we will deal with quantum programs $(\ket{\psi}, Q)$ that approximately implement a unitary transform $\rho \mapsto U \rho U^\dagger$.

\begin{definition}
    A quantum program $(\ket{\psi}, Q)$ is an $\epsilon$-approximation of a unitary transformation $U$ if 
    \begin{align*}
        \| Q (\cdot \otimes \ket{\psi}) - U(\cdot)\|_\diamond \le \epsilon.
    \end{align*}
\end{definition}

\begin{definition}[Quantum State Obfuscation for Unitaries]\label{defn:unitary-obf}
A \emph{quantum state ideal obfuscation} for the class of approximately-unitary quantum programs in the classical oracle model is a pair of QPT algorithms $(\mathsf{QObf}, \mathsf{QEval})$ with the following syntax:
\begin{itemize}
    \item $\mathsf{QObf}(1^\lambda, (\ket{\psi}, Q)) \to (\ket{\widetilde{\psi}}, \mathsf{F})$: 
    The obfuscator takes as input the security parameter $1^\lambda$ and a quantum program $(\ket{\psi}, Q)$ in the plain model\footnote{To circumvent impossibility results of \cite{BGIRSV01-obfuscation}, we say that $Q$ does not make use oracles.}, and outputs an obfuscated program specified by a state $\ket{\widetilde{\psi}}$ and a classical function $\mathsf{F}$.
    \item $\mathsf{QEval}^{\mathsf{F}}(\rho_{\mathsf{in}}, \ket{\widetilde{\psi}}) \to \rho_{\mathsf{out}}$: 
    The evaluation algorithm executes the obfuscated program on quantum input $\rho_{\mathsf{in}}$ by making use of the state $\ket{\widetilde{\psi}}$ and making superposition queries to the classical oracle $\mathsf{F}$, and produces the quantum output $\rho_{\mathsf{out}}$.
\end{itemize}

These algorithms have to satisfy the following properties.

\begin{itemize}
    \item \textbf{Functionality-Preserving:}
For every quantum program $(\ket{\psi},Q)$ which is a negligible approximation of a unitary transformation $\mathcal{U}$, there exist negligible functions $\mu,\nu$ such that
\[
    \Pr_{(\ket{\widetilde{\psi}},\mathsf{F})\leftarrow\mathsf{QObf}(1^\lambda,(\ket{\psi},Q))}
    \left[
        \bigl\|\mathsf{QEval}^{\mathsf{F}}(\cdot\otimes\ket{\widetilde{\psi}})-\mathcal{U}(\cdot)\bigr\|_\diamond
        \leq\mu(\lambda)
    \right]\geq 1-\nu(\lambda).
\]
Thus, with overwhelming probability over compilation, the sampled program is correct uniformly over all inputs, including inputs entangled with a reference register.

\item \textbf{Ideal Obfuscation:}
There exists a QPT simulator $\mathsf{Sim}$ such that for every QPT adversary $\mathcal{A}$ and quantum program $(\ket{\psi}, Q)$ that $\negl$-approximates some unitary ${U}$, and every QPT distinguisher $D$,
\[
\begin{aligned}
\Bigl|&\Pr\bigl[1 \leftarrow D(\mathcal{A}(\mathsf{QObf}(1^\lambda, (\ket{\psi}, Q))))\bigr]\\
&-\Pr\bigl[1 \leftarrow D(\Sim^{U, U^\dagger}(1^\secp, \calA, n(\secp), m(\secp)))\bigr]\Bigr|
= \negl.
\end{aligned}
\]
Here, $n$ is the input length of the quantum program, $m$ is the size of the quantum program.

\end{itemize}
\end{definition}

\begin{theorem}[Adapted from Theorem 7.2 of~\cite{HT25-unitary}]\label{thm:HT25}
    There exists a quantum state ideal obfuscation for the class of approximately unitary programs with quantum inputs and outputs in the classical oracle model.\footnote{The oracle itself is efficient outside of the part of using a truly random function, which can be heuristically instantiated via quantum-secure PRFs.}
\end{theorem}
\begin{proof}[Strengthening correctness]
Security follows from the cited theorem. We verify the per-output correctness required here for its construction. Write $\mathcal{Q}_\psi=Q(\cdot\otimes\ket\psi)$ and suppose $\|\mathcal{Q}_\psi-\mathcal{U}\|_\diamond\leq\delta$, where $\delta$ is negligible. Fix a sampled obfuscation $\omega$, and let $\mathcal{E}_\omega$ be its evaluation channel. In the construction of~\cite{HT25-unitary}, the initial teleportation outcome $i$ is uniform independently of the input, and its signature $s$ is generated using an independent signing token. The oracle uses labels $H_\omega(j,b,i,s)$ for measurement index $j$ and outcome bit $b$. Let $t$ be the number of measurements and set
\[
    \beta_\omega=\Pr_{i,s}\left[\exists j<t:\ H_\omega(j,0,i,s)=H_\omega(j,1,i,s)\right].
\]
For every fixed $(i,s)$ outside this event, all labels decode unambiguously. Authentication correctness is an exact identity of measurement instruments for every authentication key, and token verification accepts honest signatures perfectly. Applying this identity successively through the evaluator therefore gives exactly $\mathcal{Q}_\psi$, including on entangled inputs. Since the distribution of $(i,s)$ is input-independent, for some channel $\mathcal{B}_\omega$ we have
\[
    \mathcal{E}_\omega=(1-\beta_\omega)\mathcal{Q}_\psi+\beta_\omega\mathcal{B}_\omega,
    \qquad
    \|\mathcal{E}_\omega-\mathcal{U}\|_\diamond\leq\delta+2\beta_\omega.
\]
If $H$ is a truly random function with $\kappa$-bit outputs, then $\mathbb{E}_\omega\beta_\omega\leq(t-1)2^{-\kappa}$, which is negligible for $\kappa\geq\lambda+\lceil\log_2t\rceil$. For a quantum-secure pseudorandom function, the same bound holds up to a negligible term: sampling $(i,s)$ and checking these $t-1$ equalities is an efficient test. Thus $\nu:=\mathbb{E}_\omega\beta_\omega+2^{-\lambda}$ is negligible. Markov's inequality gives $\beta_\omega\leq\sqrt\nu$ with probability at least $1-\sqrt\nu$, proving the required per-output diamond-norm bound $\delta+2\sqrt\nu$.
\end{proof}

\section{Impossibility of Best-Possible One-Time Compilers}\label{sec:impossibility-of-bp-otp}
\subsection{Definition of a Best-Possible One-Time Program Compiler}\label{sec:definitions}
\paragraph{Randomized Functions.} A randomized function $f: \calX \times \calR \to \calY$ takes an input $x \in \calX$, samples randomness $r \getsr \calR$ and outputs $y = f(x; r) \in \calY$. Let $f(x)$ be the distribution over $y$-values that results from this procedure.
By considering $f$ as a circuit description of itself, $f$ is computable in time $\poly[|f|]$, where $|f|$ is its description length.

Let $\calF$ be a set of randomized functions, and for each $\secp \in \bbN$, let $\calF_\secp$ be the set of all functions in $\calF$ with description length $\secp$.

\paragraph{Sampling Programs.} Given a randomized function $f: \calX \times \calR \to \calY$, a sampling program $P$ that implements $f$ is a quantum program $P = (\rho, \Eval)$ comprising a (possibly mixed) quantum state $\rho$ and a (possibly oracle-aided) quantum unitary circuit $\Eval$. We write $|P|$ for its description length.

$P$ can be evaluated on any $x \in \calX$ by computing 
    $$\Eval(\ketbra{x}_{\scrX} \otimes \ketbra{0}_{\scrY} \otimes \rho_{\scrP})$$ 
and measuring the $\scrY$ register to obtain output $y$. Let $P(x)$ refer to the distribution over $y$-values that results from this evaluation procedure.

\begin{definition}[Correctness of Sampling Programs]\label{def:quantum-samp-program}
Let $\mu(\secp)$ be a negligible function. Given a randomized function $f \in \calF_\secp$ and a sampling program $P$, we say that \textbf{$P$ implements $f$ with error $\mu$} if for every $x \in \calX$, the statistical distance between $P(x)$ and $f(x)$ is $\leq \mu(\secp)$. %
\end{definition}

Our notion of correctness is really a notion of one-time correctness. It says that the first time the program is evaluated, it will sample from approximately the desired distribution. However, measuring the output of the program may destroy the program state $\rho$, so there is no guarantee that the second evaluation of the program will be correct.

\paragraph{One-Time Program Compilers.} A one-time program compiler $\OTP$ takes a randomized function $f$ and outputs a sampling program $P$ that implements $f$. The compiler is best-possible if $\OTP(f)$ can be simulated by any program implementing $f$ for which $|f| = |P|$.

\begin{definition}[One-Time Program Compiler]\label{def:otp}
Let $\calF$ be a family of randomized functions. A \textbf{one-time program compiler for $\calF$} is a QPT machine $\OTP$ that takes as input the description
 of a function $f \in \calF$ and outputs a sampling program $P = \OTP(f)$. For correctness, we require that there is a negligible function $\mu(\secp)$ such that for any $\secp \in \bbN$ and any $f \in \calF_\secp$, $\OTP(f)$ implements $f$ with error $\mu(\secp)$.
\end{definition}

The following definition says that a one-time program compiler is best-possible if $\OTP(f)$ can be simulated given any program $P$ that implements $f$ with low error. We consider two notions of equivalence: perfect equivalence requires $P$ to implement $f$ with $0$ error, and statistical equivalence allows $P$ to implement $f$ with a non-zero, but still negligible, error.

\begin{definition}[Best-Possible One-Time Program Compiler]\label{def:bp-sim}
    Let $\calF$ be a family of randomized functions, and let $\OTP^*$ be a one-time program compiler for $\cal{F}$. 
    
    $\OTP^*$ is a \textbf{best-possible one-time program compiler for $\cal{F}$} with statistical/perfect equivalence if for every QPT adversary $\mathcal{A}$ and any function $\mu(\secp)$ that is (respectively) negligible/identically zero, there is a QPT simulator $\Sim$ such that for every QPT distinguisher $D$, there is a negligible function $\nu(\secp)$ such that for any $\secp \in \bbN$, any $f \in \calF_\secp$, and any sampling program $P$ that implements $f$ with error $\mu(\secp)$ and satisfies $|P| = |f|$,
    \begin{align*}
        \left|\Pr \left[ 1\leftarrow D(1^\secp, f, \mathcal{A}(1^\secp, \OTP^*(f)))\right] - \Pr \left[ 1\leftarrow D(1^\secp, f, \Sim(1^\secp, P))\right]\right| \leq \nu(\secp).
    \end{align*}
\end{definition}
The distinguisher receives $f$ as auxiliary input, while the simulator receives only $P$. The negligible bound may depend on $D$, but is uniform over the admissible choices of $f$ and $P$.

In the study of best-possible obfuscation, \cite{GR07-best} requires perfect equivalence. In other words, the simulator only needs to work correctly when $P$ implements $f$ with $0$ error. Statistical equivalence allows $\mu$ to be a non-zero, but negligible, function. Requiring statistical equivalence yields a stronger notion of security.

We rule out these two notions under the hardness of statistical zero knowledge and public-coin honest-verifier perfect zero knowledge, respectively (\Cref{sec:impossibility-from-szk,sec:impossibility-from-hvpzk}).
The resulting impossibilities also follow from the post-quantum hardness of LWE and weak pseudorandomness of group actions, respectively (\Cref{thm:impossibility-of-bp-otp}).

\subsection{Impossibility from Statistical Zero-Knowledge Hardness}
\label{sec:impossibility-from-szk}
\label{sec:impossibility-of-bp-otp-proof}

We construct the two kinds of sampling functionalities from \Cref{sec:intro-impossibility-overview} from Statistical Difference: in the lossy case, the output distribution is nearly independent of the input; in the injective case, the output almost determines the input.
A best-possible compiler would let an efficient quantum test distinguish the two cases, giving a quantum polynomial-time algorithm for this $\mathsf{SZK}$-complete problem.
We use the promise-problem versions of these complexity classes, so a $\mathsf{BQP}$ algorithm need only be correct on instances satisfying the promise.

\begin{theorem}\label{thm:impossibility-of-bp-otp-from-szk}
    If there exists a best-possible one-time program compiler with statistical equivalence for all classical randomized functions with a single input bit, then $\mathsf{SZK} \subseteq \mathsf{BQP}$.
\end{theorem}
\begin{proof}
    Suppose such a compiler $\OTP^*$ exists.
    By the completeness of Statistical Difference and the polarization lemma of Sahai and Vadhan~\cite[Theorem~3.1 and Lemma~3.3]{SV03-szk}, an instance of length $n$ can be transformed in polynomial time into two classical sampling circuits $C_0,C_1$ whose output distributions $D_0,D_1$ satisfy one of the promises
    \[
        \Delta(D_0,D_1) \leq \epsilon
        \qquad\text{or}\qquad
        \Delta(D_0,D_1) \geq 1-\epsilon,
        \qquad \epsilon = 2^{-n},
    \]
    where $\Delta$ denotes statistical distance.
    It suffices to distinguish these two cases in quantum polynomial time.
    Define the sampling functionality $f(b;r)=C_b(r)$ for $b\in\bit$, padding the randomness inputs if necessary.
    For each output string $y$, also consider the constant functionality $c_y(b)=y$.
    Pad these descriptions and the comparison programs below to a common length $\secp=q(n)$ for a fixed polynomial $q$, chosen large enough for all of them.
    Thus $\epsilon$ is negligible as a function of $\secp$, and the equal-length requirement of \Cref{def:bp-sim} is satisfied.

    \paragraph{The adversary.}
    Given a sampling program $P=(\rho,\Eval)$, the adversary $\calA$ prepares the two-qubit EPR state
    \[
        \ket{\Phi}_{\scrX\scrX'}
        = \frac{\ket{00}+\ket{11}}{\sqrt{2}}
    \]
    and initializes output registers $\scrY,\scrY'$ to zero.
    It applies $\Eval$ to the input $\scrX$, output $\scrY$, and program register $\scrP$, copies $\scrY$ to $\scrY'$ using CNOT gates, and applies $\Eval^\dagger$ to $\scrX\scrY\scrP$.
    Finally, it projects $\scrX\scrX'$ onto $\ket\Phi$, outputting $0$ if the test passes and $1$ otherwise.
    We call $\scrX'$ the reference bit and it is untouched throughout, and uncomputation also leaves the copied output $\scrY'$ untouched.
    All these operations are efficient because the evaluation circuit is public.

    \paragraph{Constant functions.}
    First consider $\calA$ on $\OTP^*(c_y)$.
    By correctness, evaluation on either classical input produces $y$ except with negligible probability, uniformly over $y$.
    The same holds on the EPR input, since the two input branches have orthogonal reference states.
    Projecting the output register onto $\ket y$ therefore changes the state by negligible trace distance.
    On this projected state, copying the output only adds the fixed value $y$ to $\scrY'$.
    Uncomputation consequently restores the initial EPR state up to negligible error, giving
    \begin{equation}
        \Pr[\calA(1^\secp,\OTP^*(c_y))\to1]\leq\negl
        \label{eq:sd-constant-test}
    \end{equation}
    uniformly over $y$.

    \paragraph{Close distributions: the lossy case.}
    Let $P_y=(\ketbra{y},\Eval_{\mathsf{copy}})$, where the common evaluator ignores the input bit and copies the program register into the output register.
    Each $P_y$ perfectly implements $c_y$, while
    \[
        P = \left(\sum_y D_0(y)\ketbra{y},\Eval_{\mathsf{copy}}\right)
    \]
    implements $f$ with error at most $\epsilon$.
    Fix a negligible error bound covering $\epsilon$ and use the same simulator $\Sim$ from \Cref{def:bp-sim} for $P$ and every $P_y$.
    By \eqref{eq:sd-constant-test} and best-possible security, applied to the distinguisher that reads the output bit, we obtain $\Pr[\Sim(1^\secp,P_y)\to 1]=\negl$ uniformly over $y$.
    By linearity in the program state,
    \[
        \Pr[\Sim(1^\secp,P)\to 1]
        = \sum_y D_0(y)\Pr[\Sim(1^\secp,P_y)\to 1]
        = \negl.
    \]
    Applying best-possible security once more, now to $f$ and $P$, shows that $\Pr[\calA(1^\secp,\OTP^*(f))\to 1]=\negl$.

    \paragraph{Far distributions: the injective case.}
    There are disjoint sets $\calY_0,\calY_1$ such that $D_b(\calY_b)\geq 1-\epsilon$ for each $b\in\bit$.
    These sets are only used in the analysis; the adversary does not need to compute them.
    By correctness, after evaluation and copying the output, the reference bit $b$ and copied output $y$ satisfy $y\in\calY_b$ except with negligible probability.
    Projecting onto this event changes the state by negligible trace distance.
    Uncomputation leaves the reference and copied output untouched, so the projected state remains supported on
    \[
        \Pi_{\mathsf{sep}}
        = \sum_{b\in\bit}\sum_{y\in\calY_b}
          \mathbb{I}_{\scrX}\otimes\ketbra{b}_{\scrX'}\otimes\ketbra{y}_{\scrY'}.
    \]
    In the EPR projector, the cross terms between $b=0$ and $b=1$ vanish when multiplied by $\Pi_{\mathsf{sep}}$ on both sides, because $\calY_0$ and $\calY_1$ are disjoint.
    The resulting operator is $1/2$ times a projector, so
    \[
        \left\|\Pi_{\mathsf{sep}}
        \bigl(\ketbra{\Phi}_{\scrX\scrX'}\otimes\mathbb{I}_{\scrY'}\bigr)
        \Pi_{\mathsf{sep}}\right\| \leq \frac12.
    \]
    Thus the EPR test accepts with probability at most $1/2+\negl$, and $\Pr[\calA(1^\secp,\OTP^*(f))\to 1]\geq 1/2-\negl$.

    Repeating the experiment a constant number of times with freshly compiled programs distinguishes the close and far cases with bounded error for all sufficiently large $n$.
    The finitely many smaller instances can be hardwired, giving a $\mathsf{BQP}$ algorithm for Statistical Difference and hence $\mathsf{SZK}\subseteq\mathsf{BQP}$.
\end{proof}

This proof uses statistical equivalence because $P$ need only approximate $f$ in the close case.
\Cref{sec:impossibility-from-hvpzk} obtains an impossibility with perfect equivalence from public-coin perfect zero-knowledge hardness.

\subsection{Impossibility from Public-Coin Perfect Zero-Knowledge Hardness}
\label{sec:impossibility-from-hvpzk}

Following the notation of Goldreich, Sahai, and Vadhan~\cite{GSV98-szk}, let $\left.\mathsf{HVPZK}\right|_{\mathrm{am}}$ denote the class of promise problems having public-coin honest-verifier perfect zero-knowledge proofs.
Here the verifier and simulator run in strict polynomial time, and on YES instances the simulator's output is identically distributed to the honest verifier's view.
Equivalently, one may allow the simulator to abort with probability at most $1/2$ and require exact simulation conditioned on not aborting: Malka's transformation~\cite[Lemma~5.2 of the full version]{Malka15-perfect} removes such aborts and preserves public coins for the honest verifier.
Perfect completeness is not required.

\begin{lemma}\label{lem:perfect-sd-hard}
    Let $\mathsf{SD}_{0,1/2}$ be the promise problem whose YES instances are pairs of classical sampling circuits with identical output distributions, and whose NO instances have output distributions at statistical distance at least $1/2$.
    Then $\mathsf{SD}_{0,1/2}$ is $\left.\mathsf{HVPZK}\right|_{\mathrm{am}}$-hard under bounded-error randomized polynomial-time many-one reductions.
\end{lemma}
\begin{proof}
    For any promise problem $\Pi=(\Pi_Y,\Pi_N)$ in $\left.\mathsf{HVPZK}\right|_{\mathrm{am}}$, Malka's reduction~\cite[Lemma~4.4 of the full version]{Malka15-perfect} maps an instance $x$ to classical sampling circuits $X,Y,Z$ satisfying
    \begin{align*}
        x\in\Pi_Y &\implies \Delta(X,Y)=0
            \quad\text{and}\quad \Pr[Z=1]\geq\frac23,\\
        x\in\Pi_N &\implies \Delta(X,Y)\geq\frac14
            \quad\text{or}\quad \Pr[Z=1]\leq\frac13.
    \end{align*}
    Taking a constant number of independent copies of $X$ and $Y$ preserves equality and increases distance at least $1/4$ to at least $1/2$~\cite[Lemma~3.4]{SV03-szk}.
    Estimate $\Pr[Z=1]$ to additive error less than $1/6$ with probability at least $9/10$.
    If the estimate is below $1/2$, output a fixed pair of circuits with disjoint output distributions; otherwise, output the amplified pair.
    On a YES instance, the reduction outputs an identical pair with probability at least $9/10$.
    On a NO instance, either both possible outputs are far pairs, or $\Pr[Z=1]\leq1/3$ and the fixed far pair is output with probability at least $9/10$.
    Thus the reduction outputs an instance on the correct side of the promise with probability at least $9/10$.
\end{proof}

\begin{theorem}\label{thm:impossibility-of-bp-otp-from-hvpzk}
    If there exists a best-possible one-time program compiler with perfect equivalence for all classical randomized functions with a single input bit, then $\left.\mathsf{HVPZK}\right|_{\mathrm{am}}\subseteq\mathsf{BQP}$.
\end{theorem}
\begin{proof}
    Suppose such a compiler $\OTP^*$ exists.
    By \Cref{lem:perfect-sd-hard} and closure of $\mathsf{BQP}$ under bounded-error randomized polynomial-time reductions, it suffices to give a $\mathsf{BQP}$ algorithm for $\mathsf{SD}_{0,1/2}$.
    Given an instance $(X,Y)$ of length $n$, take $t=O(n)$ independent copies to obtain sampling circuits $C_0,C_1$ with output distributions $D_0,D_1$.
    Equality is preserved exactly, while the NO case gives $\Delta(D_0,D_1)\geq1-2^{-n}$ by the direct-product bound~\cite[Lemma~3.4]{SV03-szk}.

    Define $f(b;r)=C_b(r)$, padding the randomness inputs if necessary, and $c_y(b)=y$ for each output string $y$.
    As in \Cref{thm:impossibility-of-bp-otp-from-szk}, pad these descriptions and the comparison programs to a common length $\secp=q(n)\geq n$ for a fixed polynomial $q$.
    Use the same padding of the state registers and the same evaluator for every comparison program.
    Let $\calA$ be the EPR adversary used in that proof.

    When $D_0=D_1$, the program
    \[
        P=\left(\sum_y D_0(y)\ketbra{y},\Eval_{\mathsf{copy}}\right)
    \]
    implements $f$ exactly.
    Its state has a polynomial-size exact preparation circuit: sample a uniform random tape, evaluate $C_0$, and discard the tape.
    Each $P_y=(\ketbra{y},\Eval_{\mathsf{copy}})$ likewise has an exact preparation circuit and perfectly implements $c_y$.
    The preparation descriptions can be padded to the same length; the simulator receives the prepared state and the common evaluator.
    Apply \Cref{def:bp-sim} with error bound $0$, using one simulator $\Sim$ for $P$ and every $P_y$ and the distinguisher that reads its output bit.
    Linearity in the program state and the uniform negligible bound $\nu$ give
    \begin{align*}
        \Pr[\calA(1^\secp,\OTP^*(f))\to1]
        &\leq \sum_y D_0(y)\Pr[\calA(1^\secp,\OTP^*(c_y))\to1]+2\nu(\secp)\\
        &=\negl,
    \end{align*}
    where the last equality uses the constant-function bound \eqref{eq:sd-constant-test}.
    When $\Delta(D_0,D_1)\geq1-2^{-n}$, the far-distribution analysis in \Cref{thm:impossibility-of-bp-otp-from-szk} instead gives
    \[
        \Pr[\calA(1^\secp,\OTP^*(f))\to1]\geq\frac12-\negl.
    \]
    Thus a constant number of repetitions with freshly compiled programs decides $\mathsf{SD}_{0,1/2}$ with bounded error for all sufficiently large $n$.
    Hardwiring the finitely many smaller instances gives $\mathsf{SD}_{0,1/2}\in\mathsf{BQP}$ and hence the claimed containment.
\end{proof}

\subsection{Consequences under Concrete Assumptions}
\label{sec:impossibility-concrete-assumptions}

The zero-knowledge hardness assumptions above follow from familiar cryptographic assumptions.
We record the resulting consequences here and give the constructions and reductions in \Cref{sec:lossy-pke}.

\begin{theorem}\label{thm:impossibility-of-bp-otp}
    Assuming the post-quantum hardness of LWE for the parameter choices given in \Cref{thm:lossy-tdf-from-LWE}, there exists a family of randomized functions $\calF$ for which there is no best-possible one-time program compiler with statistical equivalence (\Cref{def:bp-sim}).

    Assuming the post-quantum weak pseudorandomness of group actions (\Cref{asm:weak-pr}), there exists a family of randomized functions $\calF$ for which there is no best-possible one-time program compiler with perfect equivalence (\Cref{def:bp-sim}).
\end{theorem}
\begin{proof}
    By \Cref{thm:lossy-pke-hard-zk}, statistically/perfectly lossy encryption implies the corresponding zero-knowledge hardness assumption in \Cref{thm:impossibility-of-bp-otp-from-szk,thm:impossibility-of-bp-otp-from-hvpzk}.
    The result follows from the statistically lossy construction under LWE (\Cref{thm:lossy-pke-from-LWE}) and the perfectly lossy construction from group actions (\Cref{thm:lossy-PKE-from-group-actions}).
\end{proof}

The same arguments also rule out sampling quantum state iO and one-time VBB security when possible; see \Cref{sec:sampling-qsio,app:vbb-when-possible}.

\section{SEQ Implies Best-Possible Testable One-Time Security}\label{sec:CSEQ}

\begin{definition}[Testable quantum program]\label{def:testable-prog}
    A quantum program $P = (\ket{\psi}, \Eval)$, specified by a pure state $\ket{\psi}$ and a public unitary evaluation circuit $\Eval$, is \emph{testable} if it is accompanied by white-box access to a quantum program $\mathcal R_\psi=(\ket{\eta},V)$ implementing the reflection unitary\footnote{The usual phase reflection $R'=I-2\ketbra\psi$ is equivalent under coherent control. Compute--copy--uncompute makes $R$ reusable and controllable while restoring its auxiliary registers, which we suppress below.}

    \begin{align*}
        R = \mathsf{X} \otimes \ketbra{\psi} + I \otimes (I - \ketbra{\psi})
    \end{align*}
    Here the recipient receives the auxiliary state $\ket{\eta}$ and the classical description of the unitary circuit $V$. With this auxiliary state, tracing out $V$'s auxiliary output must realize the channel $R(\cdot)R^\dagger$ on every flag--program input, including inputs entangled with a reference. The auxiliary state may change during this evaluation. We include the reflection program in the size of the testable program and abbreviate the resulting object as $P'=(\ket{\psi},\Eval,R)$.
    
    A testable one-time program compiler is a one-time program compiler that outputs testable quantum programs.
\end{definition}
In a relativized model the public circuits may use the shared oracles. We also use the term \emph{testable interface} for an ideal oracle interface that provides the same coherent test of its designated initial hidden state; the black-box simulation arguments below apply to such interfaces as well.

Next we revise and generalize the SEQ definition of \cite{GLRRV25-otp} to all quantum functionalities and adopt the name SEQ for this generalized notion. To distinguish from the prior work, we refer to their original, classical version as CSEQ. Intuitively, the SEQ oracle embeds a single effective application of the purified channel into a globally unitary, repeatable interface.

\paragraph{Stinespring form and notation.} Let $\Phi: \mathsf{L}(\calH_{Q}) \to \mathsf{L}(\calH_{Q})$ be a quantum channel. Without loss of generality (by padding with dummy qubits if needed), assume the Stinespring dilation has matching input/output dimensions,\footnote{Padding the output can in some cases weaken SEQ security by making the interface testable; see \Cref{sec:seq-recognizable-outputs}.} so we fix a Stinespring unitary $U_\Phi: \calH_{Q} \otimes \calH_{P} \to \calH_{Q} \otimes \calH_{P}$ with the private register $P$ initialized to $\ket{0^{m}}_{P}$ as expected by the purification. Tracing out the private register yields $\Phi$ on the query register.
We represent $\Phi$ by a circuit for such a dilation and write $|\Phi|$ for the length of this description. This length is public, and descriptions may be padded as in the preliminaries.

\begin{definition}[The Single Effective Query Oracle (SEQ)]\label{def:CSEQ-oracle}
    For a channel $\Phi$ with fixed purification $U_\Phi$ as above, the single effective query oracle $O^{\SEQ}_\Phi$ acts on a query register $\calQ$ and maintains a private register $\calP$ and a one-qubit computed flag $\calC$. The registers are initialized to $\ket{0^m}_{\calP} \otimes \ket{0}_{\calC}$. On each query, the oracle applies the following unitary on $(\calQ, \calP, \calC)$:
    \begin{enumerate}
        \item Controlled on $\calC{=}1$, apply $U_\Phi^{\dagger}$ to $(\calQ,\calP)$.\label{CSEQ-oracle-step-1}
        \item Apply the swap on the subspace spanned by $\ket{0}_{\calC} \otimes \ket{0^m}_{\calP}$ and $\ket{1}_{\calC} \otimes \ket{0^m}_{\calP}$ (and act as the identity on the orthogonal subspace). Equivalently, flip $\calC$ controlled on $\calP{=}\ket{0^m}$.\label{CSEQ-oracle-step-2}
        \item Controlled on $\calC{=}1$, apply $U_\Phi$ to $(\calQ,\calP)$.\label{CSEQ-oracle-step-3}
    \end{enumerate}
    $O^{\SEQ}_\Phi$ may also be called $O^{\SEQ}_{U_\Phi}$ to make the particular Stinespring dilation $U_\Phi$ explicit.
\end{definition}
It is often unnecessary to write $O^{\SEQ}_{U_\Phi}$, which makes the particular dilation $U_\Phi$ explicit, and instead $O^{\SEQ}_\Phi$ suffices. \Cref{lem:cseq-sim-equals-cseq-canonical} says that any two Stinespring dilations $U, U'$ of $\Phi$ will produce indistinguishable oracles $O^{\SEQ}_{U}$ and $O^{\SEQ}_{U'}$.

The above makes the SEQ interface unitary-by-construction. Intuitively, the first time we query, Step 1 is inactive, Step 2 flips the computed flag, and Step 3 applies $U_\Phi$. On subsequent queries, Steps 1 and 3 cancel on $(\calQ,\calP)$, and Step 2 flips the flag only in the $\calP{=}\ket{0^m}$ subspace; the overall interface remains unitary and well-defined.

In fact, we can verify that this unitary exactly maps the well-initialized input to its output and vice versa, while acting as identity on everything else.

\begin{definition}[\SEQ-based simulation security for one-time programs]\label{def:CSEQ-security}
    A one-time program compiler $\OTP$ satisfies \SEQ-based simulation security for a class $\mathcal{C}$ of quantum channels if there exists a q.p.t.~simulator $\Sim$ such that for every $\Phi \in \mathcal{C}$ and for every q.p.t.~distinguisher $D$, there exists a negligible function $\negl$ with
    \begin{align*}
        \bigl| \Pr\bigl[ 1 \leftarrow D(1^\lambda, \Phi, \OTP(1^\lambda, \Phi)) \bigr]
        - \Pr\bigl[ 1 \leftarrow D(1^\lambda, \Phi, \Sim^{O_\Phi^{\SEQ}}(1^\lambda)) \bigr] \bigr| \leq \negl[\lambda].
    \end{align*}
    As before, giving $\Phi$ to $D$ means that $D$ may depend on $\Phi$ (e.g., via a classical description or black-box evaluation access); this redundancy matches the order of quantifiers.
\end{definition}

We now define the notion of a best-possible one-time program compiler, which produces one-time programs that are ``best-possible'' among all programs that implement the same sampling task and are testable.

\begin{definition}[Best-possible OTP among testable programs]\label{def:bp-testable-q}
    Let $\mathcal{C}$ be a family of quantum channels. A one-time program compiler $\OTP^*$ is \textbf{best-possible among testable programs} for $\mathcal{C}$ if there exists a q.p.t.~simulator $\Sim$ such that for every $\Phi \in \mathcal{C}$, for every testable quantum program $P$ that implements $\Phi$ and satisfies $|P|=|\Phi|$, and for every q.p.t.~distinguisher $D$,
    \begin{align*}
        \bigl| \Pr\bigl[ 1 \leftarrow D(1^\lambda, \Phi, \OTP^*(1^\lambda, \Phi)) \bigr]
        - \Pr\bigl[ 1 \leftarrow D(1^\lambda, \Phi, \Sim(1^\lambda, P)) \bigr] \bigr| \leq \negl[\lambda].
    \end{align*}
    Here, ``$P$ implements $\Phi$'' means that tracing out $P$'s private register after applying $\Eval$ to the program state $\ket{\psi}$ realizes the channel $\Phi$ on the external interface.
\end{definition}
Testability restricts the competing programs $P$ in this definition; $\OTP^*$ need not itself output testable programs. We use \emph{best-possible testable security} as shorthand for security against this comparison class.

\begin{theorem}\label{thm:cseq-testable-bpotp}
Any one-time program compiler $\OTP^*$ satisfying $\SEQ$ security for a class of quantum channels (\Cref{def:CSEQ-security}) is a best-possible one-time program compiler among testable programs for that class (\Cref{def:bp-testable-q}).
\end{theorem}
\begin{proof}
Since $\OTP^*$ satisfies \SEQ{} security, there exists a q.p.t.~simulator $\Sim$ such that for every channel $\Phi$ and every q.p.t.~distinguisher $D$,
\begin{align*}
    \bigl| \Pr[1 \leftarrow D(1^\lambda, \Phi, \OTP^*(1^\lambda, \Phi))]
    - \Pr[1 \leftarrow D(1^\lambda, \Phi, \Sim^{O^{\SEQ}_\Phi}(1^\lambda))] \bigr| \leq \negl[\lambda].
\end{align*}
It therefore suffices to show that given any testable implementation $P$ of $\Phi$, we can simulate oracle access to $O^{\SEQ}_\Phi$ using $P$. Composing $\Sim$ with this wrapper yields the simulator required by \Cref{def:bp-testable-q}.

Let $P = (\ket{\psi}, \Eval, R)$ be a testable quantum program that implements $\Phi$, where $R = \mathsf{X} \otimes \ketbra{\psi} + I \otimes (I - \ketbra{\psi})$ acts on a one-qubit flag and the program register.
Without loss of generality, we can think about $\ket\psi$ padded with zero qubits which are used for the auxiliary wires, and the reflection unitary is extended to also reflect around zeroes for those qubits. The following simulator $\Sim'$ takes any such testable program $P$ that implements $\Phi$ and simulates $O^{\SEQ}_\Phi$.

\begin{definition}[Simulator $\Sim'(P)$ for $O^{\SEQ}_\Phi$ from a testable program.]\label{def:sim-for-SEQ-oracle}
    The oracle $\Sim'(P)$ maintains internally the program register $\calP$ initialized to $\ket{\psi}$ and a one-qubit computed flag $\calC$ initialized to $\ket{0}$. On each query on register $\calQ$, $\Sim'(P)$ applies the following unitary on $(\calQ, \calP, \calC)$:
    \begin{enumerate}
        \item Controlled on $\calC{=}1$, apply $\Eval^{\dagger}$ to $(\calQ,\calP)$, mapping back to the input space of $\Eval$.
        \item Apply the reflection-controlled flip $R$ to $(\calC,\calP)$, i.e., flip $\calC$ iff the program register equals $\ket{\psi}$.
        \item Controlled on $\calC{=}1$, apply $\Eval$ to $(\calQ,\calP)$.
    \end{enumerate}
\end{definition}

We now argue that $\Sim'(P)$ is perfectly indistinguishable from $O^{\SEQ}_\Phi$. We write $\calA$ for the adversary’s private workspace register, which the adversary may initialize and act upon arbitrarily; the oracle’s hidden registers $(\calC,\calP)$ remain inaccessible.

\begin{lemma}\label{thm:sim-for-SEQ-oracle}
    Let $P$ be a testable quantum program or a testable oracle interface\footnote{The interface may maintain a hidden quantum state, provided its coherent test reflects around the designated initial state, including that hidden state.} that implements $\Phi$. Then oracles for $\Sim'(P)$ (\cref{def:sim-for-SEQ-oracle}) and $O^{\SEQ}_\Phi$ are perfectly indistinguishable after any number of quantum queries.
\end{lemma}
\begin{proof}
Let $U_\Phi$ be the canonical Stinespring dilation unitary for $\Phi$. $O^\SEQ_\Phi$ involves applying $U_\Phi$. We will refer to $O^\SEQ_\Phi$ as $O^\SEQ_{U_\Phi}$ to distinguish it from a similar oracle that applies a different unitary.

We will show that $\Sim'(P)$ is equivalent to $O^\SEQ_{U'_\Phi}$, for a particular Stinespring dilation $U'_\Phi$ of the channel $\Phi$. Note that the Stinespring dilation needs $\calP$ to be initialized to $\ket{0^m}$, whereas $\Sim'(P)$ initializes it to $\ket{\psi}$. We handle this discrepancy by conjugating each step of $\Sim'(P)$ with a swap operation that swaps the states $\ket{0^m}$ and $\ket{\psi}$.

Let $U'_{\Phi}$ be the following unitary:
\begin{enumerate}
    \item Swap $\ket{0^m}_\calP$ and $\ket\psi_\calP$.
    \item Apply $\Eval$ to $(\calQ, \calP)$.
    \item Swap $\ket{0^m}_\calP$ and $\ket\psi_\calP$.
\end{enumerate}

\begin{lemma}\label{thm:U-prime-valid-dilation}
    $U'_{\Phi}$ is a valid Stinespring dilation of $\Phi$.
\end{lemma}
\begin{proof}
    Given a state $\rho$ on the query register $\calQ$, consider the following three procedures for handling the query:
    \begin{enumerate}
        \item Apply channel $\Phi$ to $\rho_\calQ$.
        \item Apply $\Eval$ to $\rho_\calQ \otimes \ket{\psi}_\calP$ and then trace out $\calP$.
        \item Apply $U'_\Phi$ to $\rho_\calQ \otimes \ket{0^m}_\calP$ and then trace out $\calP$.
    \end{enumerate}
    Procedure 1 is equivalent to procedure 2 because $P$ implements $\Phi$. Next, procedure 2 is equivalent to procedure 3 because $U'_\Phi$ maps $\ket{0^m}_\calP \to \ket{\psi}_\calP$ before applying $\Eval$. This shows that $U'_{\Phi}$ is a Stinespring dilation of $\Phi$.
\end{proof}

Let $O^\SEQ_{U'_\Phi}$ be the SEQ oracle (\cref{def:CSEQ-oracle}) that uses unitary $U'_\Phi$ instead of $U_\Phi$. 

\begin{lemma}\label{thm:Sim-prime-equivalent-to-SEQ-oracle}
    $O^\SEQ_{U'_\Phi}$ and $\Sim'(P)$ are perfectly indistinguishable after any number of quantum queries.
\end{lemma}

\begin{proof}
Let's consider the following hybrids, which transform $O^\SEQ_{U'_\Phi}$ to $\Sim'(P)$:
\paragraph{Hybrid $1$ -- $O^\SEQ_{U'_\Phi}$ --} Start with state $\ket{0}_\calC \ket{0^m}_\calP$ and handle each query as follows:
    \begin{enumerate}
        \item Controlled on $\calC{=}1$, apply ${U'_{\Phi}}^\dag$ as follows:
        \begin{enumerate}
            \item Swap $\ket{0^m}_\calP$ and $\ket\psi_\calP$.
            \item Apply $\Eval^\dag$ to $(\calQ, \calP)$.
            \item Swap $\ket{0^m}_\calP$ and $\ket\psi_\calP$.
        \end{enumerate}
        \item Controlled on $\calP{=}\ket{0^m}$, flip $\calC$.
        \item Controlled on $\calC{=}1$, apply $U'_{\Phi}$ as follows:
        \begin{enumerate}
            \item Swap $\ket{0^m}_\calP$ and $\ket\psi_\calP$.
            \item Apply $\Eval$ to $(\calQ, \calP)$.
            \item Swap $\ket{0^m}_\calP$ and $\ket\psi_\calP$.
        \end{enumerate}
    \end{enumerate}

\paragraph{Hybrid $2$:} Start with state $\ket{0}_\calC \ket{0^m}_\calP$ and handle each query as follows:
\begin{enumerate}
        \item Controlled on $\calC{=}1$:\label{step:O-U-prime-eval-inverse}
        \begin{enumerate}
            \item Swap $\ket{0^m}_\calP$ and $\ket\psi_\calP$.
            \item Apply $\Eval^\dag$ to $(\calQ, \calP)$.
            \item Swap $\ket{0^m}_\calP$ and $\ket\psi_\calP$.
        \end{enumerate}
        \item \label{step:O-U-prime-reflection}
        \begin{enumerate}
            \item \textcolor{red}{Swap $\ket{0^m}_\calP$ and $\ket\psi_\calP$.}
            \item Controlled on $\calP{=}\textcolor{red}{\ket{\psi}}$, flip $\calC$.
            \item \textcolor{red}{Swap $\ket{0^m}_\calP$ and $\ket\psi_\calP$.}
        \end{enumerate}
        \item Controlled on $\calC{=}1$:\label{step:O-U-prime-eval}
        \begin{enumerate}
            \item Swap $\ket{0^m}_\calP$ and $\ket\psi_\calP$.
            \item Apply $\Eval$ to $(\calQ, \calP)$.
            \item Swap $\ket{0^m}_\calP$ and $\ket\psi_\calP$.
        \end{enumerate}
    \end{enumerate}

Hybrids 1 and 2 are equivalent. The only difference is step \ref{step:O-U-prime-reflection}. In hybrid 1, we control on $\calP = \ket{0^m}$. In hybrid 2, we swap $\ket{0^m}$ with $\ket{\psi}$ and control on $\calP = \ket{\psi}$. These procedures implement the same operation.

\paragraph{Hybrid $3$ -- $\Sim'(P)$ -- } Start with state $\ket{0}_\calC \textcolor{red}{\ket{\psi}_\calP}$ and handle each query as follows:
\begin{enumerate}
        \item Controlled on $\calC{=}1$, apply $\Eval^\dag$ to $(\calQ, \calP)$.
        \item Controlled on $\calP{=}\ket{\psi}$, flip $\calC$.
        \item Controlled on $\calC{=}1$, apply $\Eval$ to $(\calQ, \calP)$.
    \end{enumerate}

The difference between hybrids 2 and 3 is that in hybrid 3, we have omitted all the swap operations (the steps that swap $\ket{0^m}$ and $\ket{\psi}$), and the initial state of $\calP$ is $\ket{\psi}$, not $\ket{0^m}$. 

We will argue that hybrids 2 and 3 are perfectly indistinguishable. Given any user that submits a sequence of queries to the oracle in hybrids 2 or 3, we can view their sequence of queries as a sequence of invocations of steps \ref{step:O-U-prime-eval-inverse}, \ref{step:O-U-prime-reflection}, and \ref{step:O-U-prime-eval}. We will prove the following invariant: at the start or end of any step in the sequence, the current state of the system in hybrid 2 can be mapped to the current state of the system in hybrid 3 by applying the swap operation to $\calP$, which swaps $\ket{0^m}$ and $\ket{\psi}$.

First, before the first step of the first query, the state of the oracle's registers are $\ket{0}_\calC\ket{0^m}_\calP$ in hybrid 2 and $\ket{0}_\calC\ket{\psi}_\calP$ in hybrid 3, so the invariant is satisfied at this point.

Second, let us assume the invariant is satisfied at the start of some invocation of step \ref{step:O-U-prime-eval-inverse}, and let's step through the execution of step 1 to show that the invariant is satisfied at the end. If $\calC = 0$, then step 1 acts as the identity in both hybrids, so the invarinat will be satisfied at the end. Next, let's consider the case where $\calC=1$. In hybrid 2, step 1 applies the swap operation, which transforms the state to match the initial state of hybrid 3. Then it applies $\Eval^\dag$, as is done in hybrid 3. At this point, the state is the same in the two hybrids. Finally, in hybrid 2, we apply the swap operation again so the final state in hybrid 2 could be transformed into the final state in hybrid 3 by another swap operation. Therefore the invariant is satisfied at the end of step 1.

Third, let us assume the invariant is satisfied at the start of some invocation of step \ref{step:O-U-prime-reflection} or \ref{step:O-U-prime-eval}. We can show that the invariant is still satisfied at the end of this step using similar reasoning to our argument for step 1. 

In conclusion, before or after any step of any query, the current state of the system in hybrid 2 can be mapped to the current state of the system in hybrid 3 by applying a swap operation to $\calP$. The swap operation is applied to an internal register of the oracle, which is not part of the user's view. If a computationally unbounded user makes an unbounded number of queries to either the hybrid 2 or hybrid 3 oracle and outputs a final state on their registers, this state will be the same once we trace out $\calP$. Therefore, hybrids 2 and 3 are perfectly indistinguishable.\\

In total, we have shown that $O^\SEQ_{U'_\Phi}$ and $\Sim'(P)$ are perfectly indistinguishable after any number of queries.
\end{proof}
It just remains to show that $O^\SEQ_{U_\Phi}$ and $O^\SEQ_{U'_\Phi}$ are perfectly indistinguishable. This is implied by the following lemma.

\begin{lemma}[Self-adjoint implementations hide dilation differences] \label{lem:cseq-sim-equals-cseq-canonical}
    Let $U, U'$ be two Stinespring dilations of $\Phi$ with the same ancilla register $\calP$ initialized to $\ket{0^m}_{\calP}$. Then the oracles $O^\SEQ_U$ and $O^\SEQ_{U'}$ (\cref{def:CSEQ-oracle}) are perfectly indistinguishable after any number of queries.
\end{lemma}
\begin{proof}
Let $P_0 := I_{\calQ} \otimes \ketbra{0^m}_{\calP}$ and $P_\perp := I - P_0$. Writing $O^\SEQ_{U}$ in $2\times 2$ block form on $\calC$ using the identities
    \begin{align*}
        &\mathrm{C}_{\calC}\!U = \begin{pmatrix} I & 0 \\ 0 & U \end{pmatrix},\quad \mathrm{C}_{\calC}\!U^{\dagger} = \begin{pmatrix} I & 0 \\ 0 & U^{\dagger} \end{pmatrix},\quad \mathrm{C}_{\calP{=}\ket{0^m}}\!\mathrm{NOT}_{\calC} = \begin{pmatrix} P_\perp & P_0 \\ P_0 & P_\perp \end{pmatrix},
    \end{align*}
    where we use the shorthand $\mathrm{C}_{\calP{=}\ket{0^m}}\!\mathrm{NOT}_{\calC}$ to denote the controlled-NOT on target $\calC$ with control projector $P_0$ on $\calP$, namely $\mathrm{C}_{\calP{=}\ket{0^m}}\!\mathrm{NOT}_{\calC} := \mathsf{X}_{\calC} \otimes P_0 + I_{\calC} \otimes P_\perp$. Writing $\mathsf{X}_{\calC} = \ketbra{0}{1}_{\calC} + \ketbra{1}{0}_{\calC}$ and $I_{\calC} = \ketbra{0}{0}_{\calC} + \ketbra{1}{1}_{\calC}$, the resulting $2\times 2$ block form on $\calC$ is precisely $\begin{psmallmatrix} P_\perp & P_0 \\ P_0 & P_\perp \end{psmallmatrix}$. We obtain
\begin{align}
    O^\SEQ_{U} = \begin{pmatrix} P_\perp & P_0 U^{\dagger} \\ U P_0 & U P_\perp U^{\dagger} \end{pmatrix}_{\calC;(\calP,\calQ)}.\label{eq:OU-block}
\end{align}

By Stinespring uniqueness (for equal-dimension ancillae), there exists a unitary $V_{\calP}$ on $\calP$ such that $U'P_0 = (I_{\calQ} \otimes V_{\calP})UP_0$. This equality concerns the initialized subspace; the full unitary extensions need not agree elsewhere. Let $W := I_{\calQ} \otimes V_{\calP}$. Since $P_\perp=I-P_0$, we also have
\[
    WUP_\perp U^\dagger W^\dagger
    = I-WUP_0U^\dagger W^\dagger
    = I-U'P_0(U')^\dagger
    = U'P_\perp(U')^\dagger.
\]
Define a unitary on $(\calC,\calP,\calQ)$ that is controlled by $\calC$:
\begin{align*}
    S := \ketbra{0}_{\calC} \otimes I_{\calP\calQ} + \ketbra{1}_{\calC} \otimes W\,.
\end{align*}
Then, using the block form from \eqref{eq:OU-block},
\begin{align*}
    S\,O^\SEQ_{U}\,S^{\dagger} &= \begin{pmatrix} I & 0 \\ 0 & W \end{pmatrix} \begin{pmatrix} P_\perp & P_0 U^{\dagger} \\ U P_0 & U P_\perp U^{\dagger} \end{pmatrix} \begin{pmatrix} I & 0 \\ 0 & W^{\dagger} \end{pmatrix} \\
    &= \begin{pmatrix} P_\perp & P_0 U^{\dagger} W^{\dagger} \\ W U P_0 & W U P_\perp U^{\dagger} W^{\dagger} \end{pmatrix} \\
    &= \begin{pmatrix} P_\perp & P_0 (U')^{\dagger} \\ U' P_0 & U' P_\perp (U')^{\dagger} \end{pmatrix} = O^\SEQ_{U'}\,.
\end{align*}
Thus $O^\SEQ_{U'} = S O^\SEQ_{U} S^{\dagger}$ with $S$ acting only on the hidden registers $(\calC,\calP)$ and trivially on $(\calA,\calQ)$.

Consider any $t$-query adversary that interleaves local operations $V_i$ on $(\calA,\calQ)$ with oracle calls. Writing $O_U^{\mathrm{op}} := O^\SEQ_{U} \otimes I_{\calA}$ and similarly for $O^\SEQ_{U'}$, we have
\begin{align*}
    \mathcal{U}_{U'} &= V_t O_{U'}^{\mathrm{op}} \cdots V_1 O_{U'}^{\mathrm{op}} = V_t (S O_U^{\mathrm{op}} S^{\dagger}) \cdots V_1 (S O_U^{\mathrm{op}} S^{\dagger}) \\
    &= S\,\bigl( V_t O_U^{\mathrm{op}} \cdots V_1 O_U^{\mathrm{op}} \bigr)\,S^{\dagger} = S\,\mathcal{U}_U\,S^{\dagger},
\end{align*}
since $S$ acts only on $(\calC,\calP)$ and commutes with each $V_i$ (which acts on $(\calA,\calQ)$). The joint initial state is $\ket{\psi}_{\calA} \ket{0}_{\calC} \ket{0^m}_{\calP} \ket{\phi}_{\calQ}$, and $S$ fixes it because $S$ acts as identity on the $\calC{=}0$ subspace. Therefore $\mathcal{U}_{U'} \ket{\mathrm{init}} = S\,\mathcal{U}_U \ket{\mathrm{init}}$. Tracing out the inaccessible registers $(\calC,\calP)$ yields identical reduced states on $(\calA,\calQ)$ in the two worlds (partial trace is invariant under local conjugation on the traced-out subsystem). Hence no test on $(\calA,\calQ)$ can distinguish $O^\SEQ_{U}$ from $O^\SEQ_{U'}$, even across multiple queries.
\end{proof}

\Cref{thm:U-prime-valid-dilation,thm:Sim-prime-equivalent-to-SEQ-oracle,lem:cseq-sim-equals-cseq-canonical} imply that oracle access to $\Sim'(P)$ and $O^{\SEQ}_\Phi$ are perfectly indistinguishable. 
\end{proof}

Giving the \SEQ-security simulator $\Sim$ query access to $\Sim'(P)$, instead of $O^{\SEQ}_\Phi$, completes the proof of the theorem.

\end{proof}

\section{Stateful Quantum State Obfuscation: Towards a Plain Model Construction}
\label{sec:stateful-obfuscation}

Now that we have a strong definition for OTPs which makes sense in the plain model, we would like to actually construct it there. For a moment, let us suppose that we had a candidate construction. How would we show that any functionally equivalent program reveals no more information than the candidate? In the classical setting, we can prove that iO is best-possible by using iO to obfuscate the other, allegedly ``better'' program. However, in the quantum setting, the other program may use a quantum state - so at a minimum we need quantum state iO~\cite{BBV24-qso,CG24-qio}. Further than that, the other program \emph{might change its behavior as it is queried} -- for example by refusing to answer a second query. So it seems that any security proof would need the ability to obfuscate such programs.

In this section, we define and investigate the feasibility of quantum state obfuscation for programs which may modify their behavior as they are queried.

\begin{definition}[Stateful Quantum Programs]
    A stateful quantum program is specified by a tuple $(U, \ket{\psi})$ consisting of a unitary $U$ and an initial state $\ket{\psi}$ contained in register $\calR$.
    To evaluate a stateful quantum program on a state contained in register $\calQ$, apply $U$ to register $(\calQ, \calR)$ and output register $\calQ$.

    We define oracle access to $P$ as follows. We initialize an internal register $\calR$ with the state $\ket{\psi}$. A query can be either a forward or inverse query. On a forward query in register $\calQ$, the operation $U$ is performed on registers $(\calQ, \calR)$, and the query register $\calQ$ is returned. On an inverse query in register $\calQ$, the operation $U^{-1}$ is performed on registers $(\calQ, \calR)$ and $\calQ$ is returned.
    Following our unitary-oracle convention, either query may also be coherently controlled by a qubit supplied by the caller; the inactive branch leaves both $\calQ$ and $\calR$ unchanged.
    For a quantum oracle algorithm $A$, we denote oracle access to program $P$ as $A^P$.
\end{definition}

\begin{definition}[Functional Equivalence Between Stateful Programs]
    Two stateful quantum programs $P_1 = (U_1, \ket{\psi_1})$ and $P_2 = (U_2, \ket{\psi_2})$ are $(s, \epsilon)$-functionally equivalent if for every quantum oracle algorithm $A_s$ making at most $s$ oracle queries, the trace distance between the state of $A_s$ given $P_1$ and its state when its given $P_2$ is at most $\epsilon$, that is, \footnote{Note that the order of quantifiers permits $A_s$ to depend on the programs. For example, $A_s$ could evaluate $P_1$ on a description of $U_1$.}
    \[
        \tracedist{A_s^{P_1}}{A_s^{P_2}} \leq \epsilon
    \]
    If this holds for every $s = \poly$, we say they are $(\poly, \epsilon)$-functionally equivalent.

\end{definition}

A stateful quantum program essentially implements a quantum channel \emph{with memory}. As the program is queried, the state on register $\calR$ may evolve, changing the behavior on future queries. As an example, imagine the ``counting program'' given by $\ket{\psi} = \ket{0^n}$ and $U$ which increments the contents of $\calR$ by $1$ modulo $2^n$, then classical-copies the result to register $\calQ$. If one were to query this program twice on $\ket{0}$, the program would return $\ket{1}$ the first time and $\ket{2}$ the second time.

\subsection{Stateful Quantum Indistinguishability Obfuscation}

Using the notion of stateful functional equivalence, the definition of stateful iO is quite natural: if two programs have stateful functional equivalence, they can be indistinguishably obfuscated.

\begin{definition}[Stateful Quantum Indistinguishability Obfuscation]\label{def:stateful-io}
    A stateful quantum state indistinguishability obfuscator is a QPT compiler $\Obfs$ which takes in a stateful quantum program $(U, \ket{\psi})$, then outputs another stateful quantum program $(\widetilde{U}, \ket{\widetilde{\psi}})$:
    \[
        (\widetilde{U}, \ket{\widetilde{\psi}}) 
        \gets
        \Obfs(1^\secpar, (U, \ket{\psi}))
    \]
    It must satisfy the following two properties.
    \begin{itemize}
        \item \textbf{Correctness:} There exists some negligible function $\mu$ such that for every stateful program $(U, \ket{\psi})$, the obfuscation $(\widetilde{U}, \ket{\widetilde{\psi}}) \gets\Obfs(1^\secpar, (U, \ket{\psi}))$ is $(\poly, \mu(\secpar))$-functionally equivalent to $(U, \ket{\psi})$ with probability $1-\negl$.

        \item \textbf{Security:} For every pair of stateful programs $(U_1, \ket{\psi_1})$ and $(U_2, \ket{\psi_2})$ which are $(\poly, \negl)$-functionally equivalent and satisfy $|(U_1, \ket{\psi_1})| = |(U_2, \ket{\psi_2})|$, 
        \[
            \{\Obfs(1^\secpar, (U_1, \ket{\psi_1}))\}
            \approx
            \{\Obfs(1^\secpar, (U_2, \ket{\psi_2}))\}
        \]
    \end{itemize}
\end{definition}

It is useful to compare this to the classical notion of indistinguishability obfuscation. In the classical setting, two programs $P_1$ and $P_2$ are considered to be functionally equivalent if $P_1(x) = P_2(x)$ for \emph{every} $x$. 
In the quantum setting, the functional equivalence requirement is usually relaxed to allow negligible error in implementing the program~\cite{BK21-qio,HT25-unitary}, because it is generally quite difficult to exactly implement a quantum program.

This behavior can be seen as requiring that once $P_1$ and $P_2$ are fixed, \emph{no} possible query $x$ exists that would allow distinguishing black-box access to $P_1$ from access to $P_2$. We emphasize that in this view, $x$ is still allowed to depend on the code of the two program even though the adversary otherwise can only make black-box queries to the selected program. 
As a sanity check, allowing the adversary to depend on the code ensures that this view still avoids impossibility results like \cite{BGIRSV01-obfuscation,ABDS21-qvbb} which are based on ``feeding the program to itself''.

In the stateful setting, the definition needs to account for the ability of the program to change over the course of several queries. We allow this by considering an adversary which may submit a polynomial number of queries in sequence. 
If the program did not update its state, then this would be equivalent to quantifying over all single queries, since sequential queries would have the same effect as single queries (up to negligible error due to allowing imperfect implementations).
On the other hand, it might be the case that two stateful programs have identical outputs for, say, 10 queries, before suddenly changing behavior on the 11'th. 
Since even an adversary with black-box access could perform 11 queries and notice the difference on the last query, it does not make sense to claim that the two programs can be indistinguishably obfuscated.

The restriction to a polynomial number of queries is also quite natural. 
Although any polynomial-time adversary could certainly reach a polynomially-late program state by simply performing that many queries, to reach a superpolynomially-late program state they would need to interact with the program in a non-black-box manner. Thus, we argue that the ability to skip to a superpolynomially-late program state should be considered as a major break in the security of the scheme. At that point, we should simply consider obfuscating a state\emph{less} program which allows anyone to skip to any point in the future.

Nonetheless, we do note that strengthening the functional equivalence requirement to a superpolynomial number of queries also makes sense as a (weaker) definition, if obfuscation for polynomial functional equivalence turns out to be impossible. We provide evidence in \Cref{sec:stateful-io-evidence} that obfuscation is possible for programs which are only polynomially-equivalent by constructing such a scheme in an idealized model.

\subsection{Stateful Quantum iO Implies Best-Possible Testable One-Time Programs}

We show that such stateful quantum iO is actually sufficient to achieve quantum one-time programs that are best possible among testable programs. Thus, if one were to construct stateful quantum indistinguishability obfuscation in the plain model, they would also obtain good one-time programs.

\begin{theorem}\label{thm:siO-implies-bptotp}
    Assuming stateful quantum indistinguishability obfuscation, there exists a quantum one-time program compiler for quantum channels that is best possible among testable programs (\Cref{def:bp-testable-q}).
\end{theorem}
\begin{proof}
    Let $\Phi:L(\calH_\calQ) \rightarrow L(\calH_\calQ)$ be a quantum channel and let $P = (\ket{0^m}, \Eval)$ be a canonical Stinespring dilation of $\Phi$. Equip $P$ with the circuit that reflects around $\ket{0^m}$. Fix a public polynomial $L$ bounding the size of this canonical $\SEQ$ wrapper for channel descriptions of length $s$, as well as the size of $\SEQ(P')$ for testable programs $P'$ of length $s$, including their reflection circuits and auxiliary states. Throughout this proof, pad each wrapper to length $L(s)$. The construction is $\statefuliO(\SEQ(P))$, with $s=|\Phi|$, where $\SEQ(P)$ is a program implementing the $\SEQ$ oracle on $P$ (see \Cref{sec:CSEQ}).

    Given a testable program $P'$ that implements $\Phi$ and satisfies $|P'|=|\Phi|$, the simulator constructs $\SEQ(P')$ from the supplied evaluation and reflection circuit descriptions, including the reflection auxiliary state in its private memory. This is a plain-model stateful program. The simulator outputs $\statefuliO(\SEQ(P'))$, padded using the public size $s=|\Phi|$. The two wrappers therefore have the same length. By \Cref{thm:sim-for-SEQ-oracle}, they are $(t,0)$-functionally equivalent for every $t\in\bbN$. Stateful iO security now gives the required indistinguishability.

\end{proof}

\subsection{Stateful Oracles in the Classical Oracle Model}\label{sec:stateful-io-evidence}

Next, we provide evidence that stateful obfuscation is possible by constructing it in the classical oracle model. Unfortunately, we do not yet know how to construct stateful quantum state indistinguishability obfuscation in the plain model. Even giving a plausible candidate would resolve a major open question since it implies weaker notions of obfuscation which have not yet been constructed, such as iO for pseudodeterministic circuits or for unitaries.

Another interpretation of our construction is that 
\emph{oracle models with memory are just as good as stateless oracle models.}
For example, our Single Effective Query (SEQ) security is a query interface that maintains internal state across queries. Although it might seem somewhat odd to allow an oracle to maintain state, our result shows that this behavior can be achieved in the more standard classical oracle model.

\begin{definition}[Ideal Stateful Quantum State Obfuscation]
    An ideal stateful quantum state obfuscator is a QPT compiler $\Obfs$ which takes as input a stateful quantum program $(U, \ket{\psi})$ and outputs another stateful quantum program $(\widetilde{U}, \ket{\widetilde{\psi}})$.
    \[
        (\widetilde{U}, \ket{\widetilde{\psi}}) 
        \gets
        \Obfs(1^\secpar, (U, \ket{\psi}))
    \]
    In an oracle model, $\Obfs$ may also output a description of an oracle algorithm matching the model (e.g. a unitary for the unitary oracle model), and all parties are given oracle access to the resulting operation. In the classical oracle model, this algorithm may use a truly random function internally; it is efficient relative to that function.
    
    It must satisfy the following two properties.
    \begin{itemize}
        \item \textbf{Correctness:} There exists some $\epsilon = \negl$ such that for every stateful program $(U, \ket{\psi})$, the obfuscation $(\widetilde{U}, \ket{\widetilde{\psi}}) \gets\Obfs(1^\secpar, (U, \ket{\psi}))$ is $(\poly, \epsilon)$-functionally equivalent to $(U, \ket{\psi})$ with probability $1-\negl$.

        \item \textbf{Security:} There exists a simulator $\Sim$ such that for every stateful program $P = (U, \ket{\psi})$, every QPT adversary $\adv$, and every QPT distinguisher $D$, 
        \[
            \begin{aligned}
            \Bigl|&\Pr[1\gets D(\adv(\Obfs(1^\secpar, P)))]\\
                &-\Pr[1\gets D(\Sim^{P}(1^\secpar, \adv, n(\secpar), m(\secpar)))]\Bigr| = \negl.
            \end{aligned}
        \]
        where $U$ has at most $n(\secpar)$ gates, $\ket{\psi}$ has at most $n(\secpar)$ qubits, and $P$ takes as input an $m(\secpar)$-sized register.
    \end{itemize}
\end{definition}

A straightforward hybrid argument shows that ideal stateful obfuscator is also a stateful indistinguishability obfuscator. First, ideal obfuscation allows replacing an obfuscation of $P_1$ by a simulator who only has oracle access to $P_1$. Then, since the simulator makes only a polynomial number of queries, it cannot distinguish between oracle access to $P_1$ or $P_2$. Finally, ideal obfuscation allows switching back from the simulator with oracle access to $P_2$ to an obfuscation of $P_2$.

\begin{claim}
    Any ideal stateful quantum state obfuscator is also a stateful quantum state indistinguishability obfuscator.
\end{claim}
\renewcommand{\adv}{\mathsf{Adv}}

We now show how to construct ideal stateful obfuscation in the unitary oracle model. As a corollary of \cite{HT25-unitary}, it can also be constructed in the classical oracle model.

\begin{construction}
    The construction uses a publicly verifiable quantum authentication scheme $\QAS$ (e.g. the coset authentication scheme from \cite{BBV24-qso}) and an ideal unitary obfuscation $\Obfs'$ (e.g. \Cref{thm:HT25}~\cite{HT25-unitary}).
    
    Without loss of generality, the QAS scheme supports Pauli key updates since we can always attach a uniformly random Pauli correction to the key. Such an update preserves the key distribution and changes only the decoded logical Pauli frame, hence preserving the accepting and rejecting subspaces of the authentication scheme. \cite{BBV24-qso}'s authentication scheme also supports Pauli key updates natively.
    We implement $\QAS.\Dec_k$ and $\QAS.\Enc_k$ as reversible circuits that act unitarily on an enlarged work space and return all ancillas to $\ket{0}$ (i.e., they are coherently implemented isometries), so the following overall procedure is unitary. Assume that $U$ does not make any use of oracles.

    On input $P = (U, \ket{\psi})$ and security parameter $1^\secpar$:
    \begin{enumerate}
        \item Sample an authentication key $k\gets \QAS.\KeyGen(1^\secpar)$.
        \item Compute $\ket{\widetilde{\psi}} \coloneqq \QAS.\Enc_k(\ket{\psi}) $ on register $\calR$.
        \item Let $\widetilde{U}$ be the following unitary:

        \begin{enumerate}
            \item Take as input two registers $(\calQ,\ \calR)$
            \item Apply $\QAS.\Dec_k$ to register $\calR$. If the result is $\bot$, uncompute $\QAS.\Dec_k$ and return the input registers immediately.
            \item Otherwise, apply $U$ to $(\calQ, \calR)$.
            \item Apply $\QAS.\Enc_k$ to register $\calR$.
            \item Return $(\calQ, \calR)$.
        \end{enumerate}
        \item Compute the unitary obfuscation $(\ket{\widetilde{\phi}}, C) \gets \Obfs'(\widetilde{U})$.\footnote{Here and below, $\Obfs'(V)$ includes the following adaptation: obfuscate $W=I\oplus V$, and use $Q^\dagger X_b QX_b$, where $Q$ is its reversible evaluator and the added qubit $b$ starts in $\ket0$. By \cite[Lemma~6.4]{HT25-unitary}, this restores the auxiliary state and supports reusable controlled forward and inverse evaluation with negligible error. The resulting controlled oracle queries are handled throughout the simulation and hybrids below by conjoining their controls with authentication acceptance.}
        \item Output $\ket{\widetilde{\psi}}, \Obfs'(\widetilde{U})$.
    \end{enumerate}
\end{construction}

\begin{theorem}\label{thm:ideal-sqO}
    Ideal stateful quantum obfuscation exists in the classical oracle model.
\end{theorem}
\begin{proof}
    Correctness follows from the perfect correctness of authentication and the per-output, worst-case correctness of \Cref{thm:HT25}. With overwhelming probability over the unitary obfuscation, its reusable evaluator implements the authenticated update and its controlled forward and inverse applications with negligible error over any polynomial number of queries. On honestly authenticated memory these updates implement the original stateful program, giving the required functional equivalence.

    Let $n$ be the size of the internal state register $\calR$, in qubits. Let $\adv$ be the adversary. Let $\Sim'$ be the simulator of the ideal unitary obfuscation scheme $\Obfs'$. Let $\adv'$ be the adversary obtained as $\adv' = \Sim'(\adv, \cdot)$.
    
    The simulator $\Sim^P(\adv')$ first samples $k\gets \Auth.\KeyGen(1^\secpar)$. Prepare $n$ EPR pairs in registers $(\calA_1, \calB_1),\ \dots,\ (\calA_n, \calB_n)$ and apply $\Auth.\Enc_k$ to register $\calB = (\calB_1, \dots, \calB_n)$, expanding the register as necessary, to obtain
    \[
        \propto \sum_{x\in \{0,1\}^n} \ket{x}_\calA \otimes \Auth.\Enc_k(\ket{x})_\calB
    \]
    It runs ${\adv'}^{(\cdot)}(\calB)$, answering its oracle queries as follows. Verification and its control on the oracle call are performed coherently, and the verification workspace is uncomputed before returning.
    \begin{enumerate}
        \item Parse the query register as $(\calQ, \calB')$.
        \item Run $\Auth.\Ver_k(\calB')$. If the result is $\bot$, uncompute $\Auth.\Ver_k$ and return immediately. Otherwise continue.
        \item Query $P$ on $\calQ$. Return the result.
    \end{enumerate}

    We show that $\Sim^P(\adv)$ is indistinguishable from $\adv^{\widetilde{U}}(\ket{\widetilde{\psi}})$ for any $q = \poly$ queries via a series of hybrid experiments.
    \begin{itemize}
        \item $\Hyb_{-1} = \adv(\ket{\widetilde{\psi}}, \Obfs'(\widetilde{U}))$.
        \item $\Hyb_0 = {\Sim'}^{\widetilde{U}, \widetilde{U}^\dagger}(\adv,  \ket{\widetilde{\psi}}) = {\adv'}^{\widetilde{U}, \widetilde{U}^\dagger}(\ket{\widetilde{\psi}})$.
        \item $\Hyb_1$: Prepare the half-authenticated EPR pairs
        \[
            \propto \sum_{x\in \{0,1\}^n} \ket{x}_{\calA} \otimes \Auth.\Enc_k(\ket{x})_{\calB}
        \]
        as in the simulator. Teleport $\ket{\psi}$ into the authenticated register $\calB$ using the EPR halves in register $\calA$. 
        Run ${\Sim'}^{\widetilde{U}_{k'}, \widetilde{U}_{k'}^\dagger}(\adv,  \calB) = {\adv'}^{\widetilde{U}_{k'}, \widetilde{U}_{k'}^\dagger}(\calB)$. Here, $\widetilde{U}_{k'}$ denotes that the oracle uses key $k'$.
        
        \item $\Hyb_2$: Instead of performing the teleportation and key update $k \mapsto k'$ at initialization, perform them when $\Sim'$ submits its first query (before answering the query). Until that point, store the internal state, which is initialized to $\ket{\psi}$, in register $\calR$.

        \item $\Hyb_{2+i}$ for $i = 1$ to $q$: 
        Initialize $\Sim'$ as in $\Hyb_2$. For the first $i$ queries, answer the query as in $\Sim$, performing $P$ internally using register $\calR$ and controlled on authentication acceptance. Then, upon receiving query $i+1$, teleport the state of register $\calR$ using $\calA$ and update the key $k \mapsto k'$. Then continue answering queries as in the real scheme. If no further query is made, omit the teleportation and key update.

        \item $\Sim = \Hyb_{2+q}$.
    \end{itemize}

    $\Hyb_{-1} \approx \Hyb_0$ by the security of the ideal obfuscation of unitary $\widetilde{U}$ (see \Cref{defn:unitary-obf}).

    $\Hyb_0 = \Hyb_1$ by the correctness of teleportation and the ability to perform Pauli key updates on the authentication scheme. $\Hyb_1 = \Hyb_2$ since the teleportation and key updates are performed on disjoint registers from those hold by $\adv$. The main step is to show that $\Hyb_{2+i} \approx \Hyb_{2+i+1}$.

    \begin{claim}
        For every $0\leq i<q$,
        \[
            \Hyb_{2+i} \approx \Hyb_{2+i+1}
        \]
    \end{claim}
    \begin{proof}
        Set $j=i+1$. We further divide the transition into a sequence of hybrid experiments. For an inverse-oracle query, replace $U$ by $U^\dagger$ throughout.
        \begin{itemize}
            \item $\Hyb_{2+i,1}$: For the first $i$ queries, answer the query as in $\Sim$. The difference from $\Hyb_{2+i}$ happens upon receiving query $j$.
            
            Upon receiving query $j$, coherently compute $\QAS.\Dec_k$ on the queried register $\calB'$. Let $\Pi_{\mathrm{acc}}$ project onto successful decoding, and let $\Pi_{\mathrm{EPR}}$ project the decoded message register and $\calA$ onto
            \[
                \ket{\Omega}_{\calA\calB'}=2^{-n/2}\sum_{x\in\{0,1\}^n}\ket{x}_{\calA}\otimes\ket{x}_{\calB'}.
            \]
            Measure the two-outcome projection with bad projector
            \[
                \Pi_{\mathrm{bad}}=\Pi_{\mathrm{acc}}(I-\Pi_{\mathrm{EPR}}).
            \]
            Abort and output $\bot$ only on the bad outcome. The good outcome includes both rejection and acceptance with intact EPR pairs, preserving coherence between these two subspaces. Uncompute the decoder and continue as in $\Hyb_{2+i}$, starting with teleportation using $\calA$ and the key update.

            \item $\Hyb_{2+i,2}$: This is the same as $\Hyb_{2+i,1}$ except that after the good outcome on query $j$, the order of computation and teleportation is swapped. While $\calB'$ is decoded, apply $U$ to $(\calQ,\calR)$ controlled on authentication acceptance. Uncompute the decoder, then teleport the contents of $\calR$ using $\calA$ and update $k\mapsto k'$. Return the query registers and answer subsequent queries as in the real scheme.

            \item $\Hyb_{2+i,3}$: This is the same as $\Hyb_{2+i,2}$, except that the early abort measurement on query $j$ is \emph{not} performed. The computation of $U$ remains controlled on authentication acceptance.

            Note that the decoded message in $\calB'$ is not used for query $j$; only the coherent acceptance flag controls $U$ on $(\calQ,\calR)$.

            \item $\Hyb_{2+i+1}$: Replace coherent decoding, control on acceptance, and uncomputation by coherent verification, the same control, and uncomputation. Defer teleportation and the key update until immediately before query $j+1$, omitting them if there is no further query.
        \end{itemize}

        Publicly verifiable authentication bounds the probability $\eta$ of $\Pi_{\mathrm{bad}}$ by a negligible function. To see the entangled-state guarantee explicitly, let $\sigma_{\mathrm{acc}}$ be the subnormalized decoded state on $(\calA,\calB')$ on acceptance, averaging over the authentication key. The $n$ EPR pairs are the unique joint $+1$ eigenstate of the $2n$ commuting Bell stabilizers $X_{\calA_j}X_{\calB'_j}$ and $Z_{\calA_j}Z_{\calB'_j}$. Their projectors therefore give the union bound
        \[
            \eta=\operatorname{Tr}\!\left[(I-\ketbra\Omega)\sigma_{\mathrm{acc}}\right]
            \leq\sum_{\ell=1}^n\sum_{S\in\{X,Z\}}
            \operatorname{Tr}\!\left[\frac{I-S_{\calA_\ell}S_{\calB'_\ell}}{2}\,\sigma_{\mathrm{acc}}\right].
        \]
        For either $S$-basis, measuring $\calA$ at initialization prepares a uniformly random product $S$-basis state in the authenticated message. Since $\calA$ is untouched before teleportation, each summand equals the average probability that authentication accepts but the decoded message's $\ell$th $S$-basis bit differs from that input bit. Choosing $(\ell,S)$ uniformly gives one efficient test with failure probability $\epsilon$, so $\eta\leq 2n\epsilon$. This test uses only product $X/Z$-basis preparations and single-bit tests, and publicly verifiable mapping security of \cite{BBV24-qso} makes $\epsilon$ negligible. Thus $\eta$ is negligible, and inserting the bad-outcome measurement changes the state by at most $O(\sqrt\eta)$ by gentle measurement, proving $\Hyb_{2+i}\approx\Hyb_{2+i,1}$.
        The inserted measurement detects accepted corruption; it does not measure the authentication acceptance flag itself. Its abort probability is therefore $\eta$, even when ordinary authentication rejection has large probability, and its good outcome preserves coherence between rejection and intact acceptance.
        Public verifiability is necessary because oracle access to $\Ver_k$ is used to answer the first $i$ queries.

        $\Hyb_{2+i,1}=\Hyb_{2+i,2}$ holds on the entire good subspace. On its accepting component, $(\calA,\calB')$ contains $n$ EPR pairs, so teleporting $\calR$ and then applying $U$ is equivalent to applying $U$ to $(\calQ,\calR)$ and then teleporting. On its rejecting component, both procedures skip $U$ and perform the same teleportation. The Pauli key update preserves the authentication subspaces, so rejection before teleportation remains rejection afterward. These identities hold for each teleportation outcome and hence also preserve coherence between the accepting and rejecting components.

        $\Hyb_{2+i,2}\approx\Hyb_{2+i,3}$ follows by removing the same gentle measurement, again at cost $O(\sqrt\eta)$.

        Finally, $\Hyb_{2+i,3}=\Hyb_{2+i+1}$ because coherent verification implements the same acceptance projector as decoding, and the decoded message is otherwise unused. Teleportation acts only on the simulator's registers $(\calA,\calR)$, so it can be deferred past the adversary's computation until the next query, or omitted at the end. Summing over the $q$ transitions gives negligible total error $O(q\sqrt\eta)$, taking $\eta$ to bound the bad-outcome probability in every transition.
    \end{proof}
\end{proof}
By using \Cref{thm:ideal-sqO} and obfuscating the SEQ oracle, which is a stateful quantum program, we achieve one-time programs with SEQ security for all quantum channels in the classical oracle model.

\begin{corollary}\label{cor:seq-otp-classical}
    There exist one-time programs with SEQ security for all quantum channels in the classical oracle model.
\end{corollary}

\fi

\section*{Use of Generative AI}

Generative AI tools, mainly Codex with various versions of GPT, were used in an iterative workflow involving evolving manuscript drafts and author feedback.
Assistance included organizing and expressing author-supplied ideas and arguments, drafting and revising passages, and checking and clarifying technical arguments.
Intellectual contributions included identifying a reference for the Stinespring uniqueness fact used in the proof of \Cref{lem:cseq-sim-equals-cseq-canonical}, pointing out that the assumption for the impossibility can be weakened to hardness of Statistical Difference (originally lossy encryptions), and developing the separation proof of \Cref{prop:seq-testability-gap} in \Cref{sec:seq-testability-separation}.
The authors reviewed and revised the resulting material and take responsibility for the final work.

\ifblind\else
\section*{Acknowledgements}
AG was supported in part by DARPA under Agreement No. HR00112020023, NSF CNS-2154149 and a Simons Investigator Award. This work was done in part while AG was at the Simons Institute and participating in the Challenge Institute for Quantum Computation at UC Berkeley.
\fi

\label{page:main-end}
\clearpage
\printbibliography

\clearpage
\appendix
\section*{Supplementary Material}
\crefname{appendix}{appendix}{appendices}
\Crefname{appendix}{Appendix}{Appendices}
\crefalias{section}{appendix}
\crefalias{subsection}{appendix}
\crefalias{subsubsection}{appendix}
\ifllncs
\makeatletter
\renewcommand*{\theHsection}{app.\Alph{section}}
\renewcommand*{\theHsubsection}{app.\Alph{section}.\arabic{subsection}}
\renewcommand*{\theHsubsubsection}{app.\Alph{section}.\arabic{subsection}.\arabic{subsubsection}}
\makeatother

\fi
\section{Lossy Public-Key Encryption}\label{sec:lossy-pke}
This appendix records how lossy public-key encryption (lossy PKE) gives hard problems in the zero-knowledge classes used in our impossibility results. We recall the standard construction from LWE by citation and give the perfectly lossy construction from group actions.

In injective mode, lossy PKE functions like normal PKE, in which $\Enc$ hides the message computationally, but not statistically. In lossy mode, $\Enc$ hides the message statistically or perfectly, so the output distributions of $\Enc(0)$ and $\Enc(1)$ are (respectively) statistically close or identical. Furthermore, the injective and lossy modes are indistinguishable to an adversary who is given the public encryption key.

The following definition is based on \cite[Definition 2]{HLOV11-lossy}. We allow a negligible fraction of bad keys, while requiring perfect correctness for every message on the remaining injective keys. We also distinguish statistical from perfect lossiness and require indistinguishability against quantum adversaries.

\begin{definition}[Statistically/Perfectly Lossy Public-Key Encryption]\label{def:lossy-pke}
A \textbf{statistically/perfectly lossy public-key encryption} (lossy PKE) scheme is a tuple of PPT algorithms $(\Gen, \Enc, \Dec)$ with the following syntax. 

Let $\calM$ be the message space, and let $\calR_\Gen, \calR_\Enc, \calR_\Dec$ be the sample space of the randomness for $\Gen, \Enc, \Dec$, respectively. The sizes of these sets may depend on $\secp$. Also let $r_\Gen \getsr \calR_\Gen, r_\Enc \getsr \calR_\Enc, r_\Dec \getsr \calR_\Dec$. Next,
\begin{itemize}
    \item $\Gen(1^\secp, \mathsf{mode}; r_\Gen)$: Takes a security parameter $\secp \in \bbN$ and a $\mathsf{mode} \in \{\mathsf{inj}, \mathsf{lossy}\}$ and outputs keys $(\pk, \sk)$. 
    Additionally, let $\GenPK(1^\secp, \mathsf{mode};r_\Gen)$ compute $\Gen(1^\secp, \mathsf{mode}; r_\Gen)$ but only output $\pk$.
    \item $\Enc(\pk, m; r_\Enc)$: Takes a public key $\pk$ and a message $m \in \calM$ and outputs a ciphertext $c$.
    \item $\Dec(\sk, c; r_\Dec)$: Takes a secret key $\sk$ and a ciphertext $c$, decrypts $c$, and outputs the message $m$.
\end{itemize}

We will often omit the random inputs in our notation, for instance, writing $\Enc(\pk, m)$ instead of $\Enc(\pk, m; r_\Enc)$.

Next, the encryption scheme satisfies the following properties:
\begin{itemize}
    \item \textbf{Correctness:} With overwhelming probability over the sampling of $(\pk, \sk) \gets \Gen(1^\secp, \mathsf{inj})$, the following is true for every $m \in \calM$: 
    \[\Pr_{r_\Enc, r_\Dec}\left[m = \Dec\left[\sk,\Enc(\pk, m)\right]\right] = 1\]

    \item \textbf{Statistical/Perfect Lossiness:} There is a function $\mu(\secp)$ that is (respectively) negligible/identically zero such that with overwhelming probability over the sampling of $(\pk, \sk) \gets \Gen(1^\secp, \mathsf{lossy})$, the following is true for every $m_0, m_1 \in \calM$: the statistical distance between $\Enc(\pk, m_0)$ and $\Enc(\pk, m_1)$ is $\leq \mu(\secp)$.

    \item \textbf{Indistinguishability of Modes:} The output distributions of $\GenPK(1^\secp, \mathsf{inj})$ and $\GenPK(1^\secp, \mathsf{lossy})$ are indistinguishable to any QPT adversary. Formally, we require that for any QPT adversary $\calA$, there is a negligible function $\nu(\secp)$ such that:
    \[\left|\Pr\left[\calA(1^\secp, \GenPK(1^\secp, \mathsf{inj})) \to 1 \right] - \Pr\left[\calA(1^\secp, \GenPK(1^\secp, \mathsf{lossy})) \to 1 \right]\right| \leq \nu(\secp)\]
\end{itemize}
\end{definition}

\subsection{Hard Zero-Knowledge Problems from Lossy Encryption}\label{sec:lossy-pke-hard-zk}
The standard implication from statistically lossy encryption to hard problems in $\mathsf{SZK}$ is recalled in \cite[Appendix~A.1 of the full version]{BDRV18-laconic}. Perfect lossiness gives the corresponding implication for public-coin honest-verifier perfect zero knowledge.

\begin{lemma}\label{thm:lossy-pke-hard-zk}
If statistically lossy PKE exists, then $\mathsf{SZK}\not\subseteq\mathsf{BQP}$. If perfectly lossy PKE exists, then $\left.\mathsf{HVPZK}\right|_{\mathrm{am}}\not\subseteq\mathsf{BQP}$. In both cases, the message space is assumed to contain two distinct efficiently specifiable messages.
\end{lemma}
\begin{proof}
Given a public key $\pk$, construct polynomial-size sampling circuits $Q_b(r)=\Enc(\pk,m_b;r)$ for two distinct messages $m_0,m_1$. For all but a negligible fraction of injective keys, correctness implies that $Q_0$ and $Q_1$ have disjoint supports. For all but a negligible fraction of lossy keys, their distributions are statistically close, or identical in the perfectly lossy case. Distinguishing statistically close distributions from distributions with disjoint supports is a promise problem in $\mathsf{SZK}$. The promise problem with identical distributions as YES instances and disjoint supports as NO instances has a public-coin honest-verifier perfect zero-knowledge proof \cite[Proposition 6.2.6]{Vadhan99-thesis}; its simulator samples the public challenge and the selected circuit's random input, so runs in strict polynomial time. A $\mathsf{BQP}$ algorithm for the corresponding promise problem would distinguish the two key-generation modes, contradicting \Cref{def:lossy-pke}.
\end{proof}

\subsection{Statistically Lossy Encryption from LWE}
Statistically lossy PKE with $\ell$-bit messages follows from $(n,k)$-lossy trapdoor functions whenever $k-\ell=\omega(\log\secp)$ \cite[Section 4.2]{BHY09-selective}. The construction preserves perfect correctness on good injective keys and applies to almost-always lossy trapdoor functions. Its key-indistinguishability reduction also applies to QPT adversaries. Combining it with \cite[Theorem 6.4 of the full version]{PW08-lossy} gives the following standard consequence.

\begin{theorem}[{\cite{BHY09-selective,PW08-lossy}}]\label{thm:lossy-pke-from-LWE}\label{thm:lossy-tdf-from-LWE}
Let $q=4\secp^{10}$, $\alpha=1/(32\secp^{10})$, and $\chi=\overline{\Psi}_{\alpha}$. Assuming the post-quantum hardness of $\mathsf{LWE}_{q,\chi}$, there exists a statistically lossy PKE scheme (\Cref{def:lossy-pke}) with message space $\calM=\bit^\secp$.
\end{theorem}
Indeed, the parameters $d=\secp$, $c_1=4$, and $c_2=c_3=2$ in \cite[Theorem 6.4 of the full version]{PW08-lossy} give input length $n=\secp^2$ and lossiness $k\geq n/2-o(n)$, sufficient for the standard encryption construction.

\subsection{Construction from Group Actions}
In this section, we describe a perfectly lossy PKE assuming the post-quantum weak pseudorandomness of effective group actions.

\paragraph{Notation.} For a regular and abelian group action $\star: \mathbb{G} \times \mathbb{X} \rightarrow \mathbb{X}$, we use additive notation to denote the group operation in $\mathbb{G}$. %

\begin{definition}[Effective Group Action] A regular and abelian group action $(\mathbb{G}, \mathbb{X}, \star)$ is effective if it satisfies the following properties.
\begin{enumerate}
    \item The group $\mathbb{G}$ is finite and there exist efficient p.p.t.~algorithms for membership testing (deciding whether a binary string represents a group element), equality testing, group operation and computing the inverse of any element. We require exact uniform sampling in $\mathbb{G}$ in strict polynomial time.
    \item The set $\mathbb{X}$ is finite and there exist efficient algorithms for membership testing (to check if a binary string represents a valid set element), and unique representation.
    \item There exists a distinguished element $x_0 \in \mathbb{X}$ with known representation.
    \item There exists an efficient algorithm that given any $g \in \mathbb{G}$ and any $x \in \mathbb{X}$, outputs $g \star x$.
\end{enumerate}
\end{definition}

\begin{assumption}[Post-Quantum Weak Pseudorandomness of an Effective Group Action, strengthening {\cite{ADMP20-ega}}]\label{asm:weak-pr}
Suppose $(\mathbb{G}, \mathbb{X}, \star)$ is a family of effective group actions indexed by $\secp$, with $|\mathbb{G}|^{-1} = \negl$. For every polynomial $t=t(\secp)$, the distributions $\{(x_i, g \star x_i)\}_{i=1}^t$ and $\{(x_i, y_i)\}_{i=1}^t$ are computationally indistinguishable to QPT adversaries, where $g \gets \mathbb{G}$ and all $x_i, y_i \gets \mathbb{X}$ are sampled independently and uniformly at random.
\end{assumption}

\begin{theorem}\label{thm:lossy-PKE-from-group-actions}
Suppose there exists an effective group action for which \Cref{asm:weak-pr} holds. Then, there exists a perfectly lossy public key encryption scheme (\Cref{def:lossy-pke}) with message space $\calM = \bit^\secp$.
\end{theorem}

\begin{construction}\label{cons:lossy-pke-group-actions}
    Let $(\mathbb{G}, \mathbb{X}, \star)$ be a group action for which \Cref{asm:weak-pr} holds. Let $x_* \in \mathbb{X}$ be a distinguished point that is specified as a public parameter. We first describe encryption of one bit.
    \begin{itemize}
        \item $\Gen(1^\secp, \mode)$: Sample $v \gets \mathbb{G}$ and set $x_0 = v \star x_*$, which is uniform in $\mathbb{X}$ by regularity.
        \begin{itemize}
            \item If $\mode = \lossy$, then sample random $g, h \gets \mathbb{G}$ and set $\pk := (x_0, g \star x_0, h \star x_0, (g + h) \star x_0)$ and $\sk = \bot$.
            \item If $\mode = \inj$, then sample independent random $g, h, k \gets \mathbb{G}$. Set $\pk = (x_0, g \star x_0, h \star x_0, k \star x_0)$ and $\sk = (g, k - h)$.
        \end{itemize}
        \item $\Enc(\pk, m)$:
        \begin{itemize}
            \item Sample a random $u \gets \mathbb{G}$. Parse $\pk = (x_0, x_1, x_2, x_3)$.
            \item If $m = 0$, output $\ct = (u \star x_0, u \star x_1)$.
            \item If $m = 1$, output $\ct = (u \star x_2, u \star x_3)$.
        \end{itemize}
        \item $\Dec(\sk, \ct)$: 
        \begin{itemize}
            \item Parse injective secret key as $\sk = (g, g')$ and ciphertext $\ct = (y_1, y_2)$.
            \item If $g \star y_1 = y_2$ then output $m = 0$.
            \item Else if $g' \star y_1 = y_2$ then output $m = 1$.
            \item Otherwise, output $\bot$.
        \end{itemize}
    \end{itemize}
    To encrypt a message in $\bit^\secp$, encrypt each bit independently under the same public key and decrypt componentwise.
\end{construction}

\begin{proof}
We will prove that the scheme described in \Cref{cons:lossy-pke-group-actions} is a perfectly lossy PKE scheme, assuming the weak pseudorandomness assumption of the group action.
\paragraph{Correctness.}
Except when $k=g+h$, an event of probability $1/|\mathbb{G}|=\negl$, this encryption scheme satisfies perfect decryption correctness on injective keys.
In the injective mode, an encryption of $0$ looks like
\[
  (u \star x_0,  u \star x_1) = (u \star x_0, (u + g) \star x_0),
\]
for a random $u \gets \mathbb{G}$, while an encryption of $1$ looks like
\[
  (u\star x_2, u \star x_3) = ((u + h) \star x_0, (u + k) \star x_0).
\]
In order for there to be a collision between encryptions of $0$ and $1$, we would need
some $w,v$ satisfying
\[
  w \star x_0 = (v + h) \star x_0 \quad \text{and} \quad (w + g) \star x_0 = (v + k) \star x_0.
\]
Because we are working with a regular group action, this implies $w = v + h$ and $w + g = v + k$, hence $g + h = k$. Therefore, on every key with $k\neq g+h$, decryption is perfectly correct, as required by \Cref{def:lossy-pke}.

\paragraph{Perfect Lossiness.}
In the lossy mode, note that an encryption of $1$ takes the form
\[
\begin{aligned}
  (u \star x_2, u \star x_3)
  &= ((u + h) \star x_0,\, (u + g + h) \star x_0)\\
  &= (u' \star x_0, (u' + g) \star x_0)\\
  &= (u' \star x_0, u' \star x_1).
\end{aligned}
\]
Here $u' = u + h$. Since $u$ is sampled uniformly at random, so is $u'$, and thus the distribution of an encryption of $1$ is identical to that of an encryption of $0$. Hence, encryptions of $0$ and $1$ are identically distributed, and the scheme is perfectly lossy.

\paragraph{Indistinguishability of modes.} Apply \Cref{asm:weak-pr} to two samples and order the resulting tuple as $(x_0,g\star x_0,x_2,g\star x_2)$. Since $x_0,x_2$ are independent and uniform, regularity lets us write $x_2=h\star x_0$ for a uniform $h\in\mathbb{G}$ independent of $x_0,g$. Thus this is exactly a lossy public key. In the random case, all four entries are independent and uniform, which is exactly the injective public-key distribution, proving indistinguishability against QPT adversaries.

Finally, encrypting each bit independently preserves perfect correctness on good injective keys and perfect lossiness for every message in $\bit^\secp$, while the public-key distributions remain unchanged.
\end{proof}

\section{Classical Single Effective Query (CSEQ)}
\subsection{Classical Single Effective Query (CSEQ) Model}
\label{sec:prelim_seq_model}
In this section, we recall the \emph{single effective query} security definition from \cite{GLRRV25-otp} and, to distinguish it from our generalized notion, refer to this classical variant as CSEQ.
\begin{definition}[The Single Effective Query Oracle]
\label{def:SEQ-security}
    For a randomized function $f : \calX \times \calR \rightarrow \calY$, we define the classical single effective query oracle $O^{\CSEQ}_f$ as implementing the following algorithm:
    \begin{itemize}
        \item We assume it has oracle access to $O_f$, which acts as
        \[\ket{x,r, u} \mapsto \ket{x, r, u \oplus f(x;r)}.\]
        \item The oracle maintains two internal registers: a workspace register $\mathcal{W}$ to perform intermediate computations, and a database register $\mathcal{D}$ to implement the database for the compressed random oracle. These registers are initialized to $\ket{0}_{\mathcal{W}} \otimes \ket{\emptyset}_{\mathcal{D}}$.
        \item To describe its behavior on queries, all we need to do is describe its behavior on basis states. We will maintain the invariant that the workspace register is always $\ket{0}_{\calW}$. On query $\ket{x,u,b}_\calQ$ on query register $\calQ$, the oracle implements an isometry specified by the following steps on each basis state $\ket{x,u,b}_{\calQ} \otimes \ket{D}_\calD \otimes \ket{0}_\calW$:
        \begin{enumerate}
            \item If there exists some $x'\neq x$ such that $(x',r) \in D$ for some $r \in \calR$, skip the following steps.\label{step:SEQ-check}
            \item Copy $x$ into the workspace register $$\ket{x,u,b}_{\calQ} \otimes \ket{D}_\calD \otimes \ket{0}_\calW \mapsto \ket{x,u,b}_{\calQ} \otimes \ket{D}_\calD \otimes \ket{x, 0}_\calW.$$
            \item Apply a compressed random oracle query isometry with query register $\calW$ and database register $\calD$.
            \item On the basis states apply the following unitary map that acts on registers $\calQ, \calW$:
            \begin{align*}
                \ket{x, u, b}_{\calQ} \otimes \ket{D}_{\calD} \otimes \ket{x, r}_{\calW} \mapsto \ket{x, u \oplus f(x; r), b \oplus 1}_{\calQ} \otimes \ket{D}_{\calD} \otimes \ket{x, r}_{\calW}
            \end{align*}
            \item Uncompute the workspace register: first, query the compressed oracle again with the query register being $\calW$ and the database register being $\calD$, and then copy the input $x$.
        \end{enumerate}
    \end{itemize}
\end{definition}

\begin{definition}[$\CSEQ$-based simulation security for one-time programs]\label{def:classical-SEQ-security}
    A one-time program compiler $\OTP$ satisfies $\CSEQ$-based simulation security for a class $\calF$ of randomized functions if there exists a q.p.t.~simulator $\Sim$ such that for every $f \in \calF$ and for every q.p.t.~distinguisher $D$, there exists a negligible function $\negl$ such that
    \begin{align*}
        \left| \Pr\left[ 1 \leftarrow D(1^\lambda, f, \OTP(1^\lambda, f)) \right] - \Pr\left[1 \leftarrow D(1^\lambda, f, \Sim^{O_f^{\CSEQ}}(1^\lambda))\right] \right| \le \negl[\lambda].
    \end{align*}
\end{definition}
By giving $f$ as input to the distinguishing algorithm, we mean that $D$ gets access to $f$ in a way that it can evaluate on any $(x,r)$ of its choice. Note that this is actually redundant, because the order of quantifiers is such that $D$ is allowed to depend on the choice of $f$.

\subsection{When is \CSEQ{} a special case of \SEQ}\label{rem:seq-in-cseq}
    We will show that for any classical functionality $f$ with a domain size $|\calX|>1$, there is a channel $\Phi_f$ such that the classical SEQ oracle $O^\CSEQ_f$ reveals the same information as the (channel) SEQ oracle $O^\SEQ_{\Phi_f}$.\\

    The channel $\Phi_f$ computes $f$ on any input $x$ and uses a (compressed) random oracle $H$ to sample $r$. Given a function $f: \calX \times \calR \to \calY$ and a random oracle $H: \calX \to \calR$, $\Phi_f$ maps $x \to f(x; H(x))$ and also flips a bit $b$ to indicate that it has executed. $\Phi_f$ is essentially the same as the classical SEQ oracle $O^\CSEQ_f$, except without step \ref{step:SEQ-check} of $O^\CSEQ_f$.
    
    \begin{definition}[Channel $\Phi_f$]\label{def:SEQ-channel}   
        Given a randomized classical function $f: \calX \times \calR \to \calY$, let $\Phi_f$ be the following quantum channel. 
        \begin{itemize}
        \item The channel's internal register has two parts $\calP = \calD \times \calW$. $\calD$ is a database register for the compressed random oracle, and $\calW$ is a workspace register to perform intermediate computations. $\calD$ is initialized to $\ket{\emptyset}_{\mathcal{D}}$, and we maintain the invariant that $\calW$ is in the state $\ket{0,0}_{\mathcal{W}}$ at the start of each query.
        \item The channel's Stinespring unitary $U_{\Phi_f}$ acts as follows on each basis state $\ket{x,u,b}_{\calQ} \otimes \ket{D}_\calD \otimes \ket{0,0}_\calW$:
        \begin{enumerate}
            \item Copy $x$ into the workspace register:\label{step:copy-x} 
            \[\ket{x,u,b}_{\calQ} \otimes \ket{D}_\calD \otimes \ket{0,0}_\calW \mapsto \ket{x,u,b}_{\calQ} \otimes \ket{D}_\calD \otimes \ket{x, 0}_\calW\]
            \item Apply the query operation of the compressed oracle with $\calW$ as the query register and $\calD$ as the database register.\label{step:query-CO}
            \item Apply the unitary that acts as follows on the basis states:
            \begin{align*}
                \ket{x, u,b}_{\calQ} \otimes \ket{D}_{\calD} \otimes \ket{x, r}_{\calW} \mapsto \ket{x, u \oplus f(x; r), b \oplus 1}_{\calQ} \otimes \ket{D}_{\calD} \otimes \ket{x, r}_{\calW}
            \end{align*}\label{step:answer-query}
            \item Uncompute step \ref{step:query-CO} and then step \ref{step:copy-x}.
        \end{enumerate}
    \end{itemize}
    \end{definition}

Note that $f$ and $\Phi_f$ provide no one-time guarantee. They never reject queries, so the user may make many effective queries. There are two ways to limit the user's queries: we can implement $f$ with the classical SEQ oracle $O^\CSEQ_f$ and implement $\Phi_f$ with the SEQ oracle $O^\SEQ_{\Phi_f}$. The purpose of both oracles is to restrict the user to just one effective query.

Which oracle reveals more information about $f$? In fact, they reveal the same information. We formalize this by showing that $O^\CSEQ_f$ can be used to simulate $O^\SEQ_{\Phi_f}$, and vice versa as long as $|\calX| > 1$ (\cref{thm:cseq-seq-equivalence}).

To prove the first simulation indistinguishability in \cref{thm:cseq-seq-equivalence}, we first define the CSEQ-wrapper simulator $\Sim'$.

The simulator $\Sim'$ below has the following internal registers. Let $\calP$ store the program state of $P$, which is initialized to $\ket{\psi}$. Let $\calD$ be an internal database register that stores a query $x \in \calX \cup \{\emptyset\}$ and is initialized to $\ket{\emptyset}$. Let $\calW_x \times \calW_y \times \calW_b$ be work registers that store values $(x, y, b) \in (\cal{X} \cup \{\emptyset\}) \times (\calY \cup \{\emptyset\}) \times \text{$\{0,1\}$}$. $\calW_x \times \calW_y \times \calW_b$ are initialized to $\ket{\emptyset, \emptyset, 0}$.

We represent elements of $\calX \cup \{\emptyset\}$ as the following binary strings. Each $x \in \calX$ becomes $x \|1$ and $\emptyset$ becomes $0^n\|0$. In particular, $x=0^n$ is represented as $0^n\|1$, which is different from the representation $0^n\|0$ of $\emptyset$. Therefore, we have that $\emptyset + x = x$, and $x + x = \emptyset$. The same kind of representation is used for $\calY \cup \{\emptyset\}$.

$\Sim'(P)$ works as follows. First, it uses the $\EvalCopy$ operation to evaluate $P(x)$. Rather than evaluating the program on registers $\calQ_x \times \calQ_u$ directly, $\EvalCopy$ evaluates it on work registers $\calW_x \times \calW_y$ instead, and copies the inputs and outputs between $\calQ_x \times \calQ_u$ and $\calW_x \times \calW_y$. This provides insulation from any pathological behavior of $\Eval$. Even if $\Eval$ modifies or reads from its query register arbitrarily, we are guaranteed that $\EvalCopy$ only reads from $\calQ_x$ and only writes to $\calQ_y$.

Next, $\Sim'(P)$ stores a query $x'$ in the database $\calD$ and rejects any new query $x$ if $x' \notin  \{x, \emptyset\}$. This is analogous $O^\CSEQ_f$'s database. The $\ReflectCopy$ operation updates $\Sim'(P)$'s database by copying $x$ from $\calQ_x$ to $\calD$ if and only if the remaining internal registers $\calP \times \calW_x \times \calW_y \times \calW_b$ are in their initial state (\cref{thm:reflect-and-copy}).

\begin{definition}[Simulator $\Sim'$ for $O_f^{\CSEQ}$]\label{def:sim-for-CSEQ}
The simulator also applies to testable oracle interfaces, using their evaluation and coherent-test operations directly. All subsequent simulation lemmas in this appendix cover both forms of access.
\begin{itemize}
    \item $\Sim'$ takes as input a testable quantum program $P = (\ket{\psi}, \Eval, R)$.
    \item Initialize the internal state to $\ket{\psi}_\calP \otimes \ket{\emptyset}_\calD \otimes \ket{\emptyset}_{\calW_x} \otimes  \ket{\emptyset}_{\calW_y} \otimes \ket{0}_{\calW_b}$.
    \item Given a query $\ket{x, u, b}_\calQ$ and database $\ket{x'}_\calD$, $\Sim'$ acts as follows.
    \begin{enumerate}
        \item If $x' \notin \{x, \emptyset\}$, do nothing and ignore the remaining steps. 
        \item $\ReflectCopy$: If $\calW_x \times \calW_y$ contains $(\emptyset, \emptyset)$, then do the following:\label{step:reflect-and-copy}
        \begin{enumerate}
            \item Apply the reflection oracle $R$ to registers $(\calW_b,\calP)$.
            \item If $\calW_b$ contains $1$, then $\CNOT$ the value $x$ from $\calQ_x$ onto $\calD$.
            \item Apply the reflection oracle $R$ to registers $(\calW_b,\calP)$.
        \end{enumerate}
        \item $\EvalCopy$:\label{step:eval-and-copy}
        \begin{enumerate}
            \item $\CNOT$ the value $x$ from $\calQ_x$ to $\calW_x$.
            \item Apply $\Eval$ to registers $(\calW_x, \calW_y, \calP)$.
            \item $\CopyOutput$: $\CNOT$ the contents of $\calW_y$ onto $\calQ_u$.
            \item Apply $\Eval^\dag$ to registers $(\calW_x, \calW_y, \calP)$.
            \item $\CNOT$ the value $x$ from $\calQ_x$ to $\calW_x$.
        \end{enumerate}
        \item Apply $\ReflectCopy$ (step \ref{step:reflect-and-copy}) again.
        \item $\Bitflip$: Apply $\mathsf{X}$ to the $\calQ_b$ register.
    \end{enumerate}
\end{itemize}
\end{definition}
In summary, the simulator acts as follows: it acts as the identity on states $\ket{x, u, b} \ket{D} \ket{\cdots}$ if $x' \in D$ for some $x' \neq x$, and on the rest of the state space it acts as the unitary
\begin{align*}
    \Bitflip \circ \ReflectCopy \circ \EvalCopy \circ \ReflectCopy.
\end{align*}

The next lemma says that $\ReflectCopy$ acts similarly to a reflection oracle on registers $\calP \times \calW_x \times \calW_y \times \calW_b$. It CNOTS $x$ from $\calQ_x$ to $\calD$ if and only if $\calP \times \calW_x \times \calW_y \times \calW_b$ is in its initial state.
\begin{lemma}\label{thm:reflect-and-copy}
    If $P$ is a testable quantum program, then $\ReflectCopy$ (step \ref{step:reflect-and-copy} of $\Sim'(P)$) is equivalent to the operation that applies $\CNOT$ from $\calQ_x$ to $\calD$ controlled on $\calP \times \calW_x \times \calW_y \times \calW_b$ being in their initial state $\ket{\psi, \emptyset, \emptyset, 0}$.

    Additionally, after any number of queries to $\Sim'(P)$, $\calW_b$ is in state $\ket{0}$.
\end{lemma}
\begin{proof}    
    First, $\calW_b$ is initialized to $\ket{0}$. Next, note that $\calW_b$ is only modified during the $\ReflectCopy$ operation. Let us assume that at the start of a given $\ReflectCopy$ operation, $\calW_b$ contains $\ket{0}$. We will show that at the end of this operation, $\calW_b$ still contains $\ket{0}$. This establishes that $\calW_b = \ket{0}$ after every query to $\Sim'(P)$.

    Let us consider how $\ReflectCopy$ acts on the following basis states.

    Case 1: $\calW_x \times \calW_y$ contains a state orthogonal to $\ket{\emptyset, \emptyset}$. Then $\ReflectCopy$ acts as the identity, so at the end of $\ReflectCopy$, $\calW_b = \ket{0}$, and $\calQ_x \times \calD \times \calP \times \calW_x \times \calW_y$ are unchanged.
    
    Case 2: $\calW_x \times \calW_y = \ket{\emptyset, \emptyset}$, but $\calP$ contains a state that is orthogonal to $\ket{\psi}$. $\ReflectCopy$ first applies $R$ to $(\calW_b, \calP)$, resulting in $\calW_b = \ket{0}$. Then the $\CNOT$ operation from $\calQ_x$ to $\calD$ is controlled on $0$, so we do not copy $x$ to $\calD$, and we can replace this operation with the identity. Finally, the second application of $R$ also leaves the $\calW_b = \ket{0}$. At the end of $\ReflectCopy$, $\calW_b = \ket{0}$, and $\calQ_x \times \calD \times \calP \times \calW_x \times \calW_y$ are unchanged.
    
    Case 3: $\calW_x \times \calW_y = \ket{\emptyset, \emptyset}$, and $\calP = \ket{\psi}$. $\ReflectCopy$ first applies $R$ to $(\calW_b, \calP)$, resulting in $\calW_b = \ket{1}$. Then, the $\CNOT$ operation copies $x$ from $\calQ_x$ to $\calD$. Since the $\CNOT$ step acts only on $\calW_b \times \calQ_x \times \calD$, this does not affect the $\calP$ register, which will still be in the state $\ket{\psi}$. Finally, the second application of the reflection oracle $R$ will XOR the $\calW_b$ register with $1$, returning it to $\ket{0}$. At the end of $\ReflectCopy$, $\calW_b = \ket{0}$, and $\calP \times \calW_x \times \calW_y$ are unchanged.

    In summary, we've shown that for each type of basis state listed above, $\ReflectCopy$ keeps $\calW_b = \ket{0}$, and it performs a $\CNOT$ from $\calQ_x \to \calD$ controlled on $\calP \times \calW_x \times \calW_y = \ket{\psi, \emptyset, \emptyset}$. Any state on registers $\calQ_x \times \calD \times \calP \times \calW_x \times \calW_y$ is in the span of the basis states above, so we have completely characterized the behvaior of $\ReflectCopy$.
\end{proof}

Next, let us define a testable quantum program $P_f = (\ket{+}, \Eval_f, R_+)$ as follows: 
    \begin{itemize}
        \item $\ket{+} = \frac{1}{\sqrt{\abs{\calR}}} \cdot \sum_{r \in \calR} \ket{r}$
        \item $\Eval_f$ is a unitary that maps $\ket{x, u} \otimes \ket{r} \mapsto \ket{x, u \oplus f(x;r)} \otimes \ket{r}$
        \item $R_+$ is the reflection about $\ket{+}$
    \end{itemize}
Note that $P_f$ is a testable quantum program that implements $f$ with $0$ error.

\begin{lemma}\label{lem:bpqotp-seq-sim-prime}
    Oracles for $\Sim'(P_f)$ and $O^\CSEQ_f$ are perfectly indistinguishable to any adversary making an unbounded number of quantum queries.
\end{lemma}

\begin{proof}
The following hybrids transform $O^\CSEQ_f$ into $\Sim'(P_f)$.

\paragraph{Hybrid 0.} This is the CSEQ oracle, except that the size of the compressed oracle database is fixed to record at most 1 query. This is equivalent to the CSEQ oracle because the CSEQ oracle holds at most one entry in the database at any point in time, by Claim 4.7 of \cite{GLRRV25-otp}.

Recall that $\StdDecomp$ applied to $\calW \times \calD$ swaps the states $\ket{x,r}_{\calW} \otimes \ket{\emptyset,\emptyset}_{\calD}$ and $\ket{x,r}_{\calW} \otimes \ket{x,+}_{\calD}$, and acts as the identity on the rest of the space. Additionally, $\mathsf{CStO}'$ applied to $\calW \times \calD$ maps $\ket{x, r}_\calW \otimes \ket{D}_\calD \mapsto \ket{x, r \oplus D(x)} \otimes \ket{D}$.

The oracle initializes its internal registers $\calW \times \calD = \ket{0} \times \ket{\emptyset}$. On query $\ket{x,u,b}_\calQ$, the oracle acts as follows:
        \begin{enumerate}
            \item If there exists some $x'\neq x$ such that $(x',r) \in D$ for some $r \in \calR$, skip the following steps.
            \item Copy $x$ into the workspace register $$\ket{x}_{\calQ_x} \otimes \ket{0}_\calW \mapsto \ket{x}_{\calQ_x} \otimes \ket{x, 0}_\calW.$$\label{step:hyb-0-copy-input-1}
            \item Apply the compressed random oracle to registers $\calW \times \calD$ by executing the following unitaries:\label{step:hyb-0-query-CO-1}
            \[\StdDecomp \circ \mathsf{CStO}' \circ \StdDecomp\]
            \item Apply the following unitary map:\label{step:hyb-0-copy-output}
            \begin{align*}
                \ket{u, b}_{\calQ_u \times \calQ_b} \otimes \ket{x, r}_{\calW} \mapsto \ket{u \oplus f(x; r), b \oplus 1}_{\calQ_u \times \calQ_b} \otimes \ket{x, r}_{\calW}
            \end{align*}
            \item Query the compressed oracle again:\label{step:hyb-0-query-CO-2}
            \[\StdDecomp \circ \mathsf{CStO}' \circ \StdDecomp\]
            \item Copy $x$ into the workspace register $$\ket{x}_{\calQ_x} \otimes \ket{x,0}_\calW \mapsto \ket{x}_{\calQ_x} \otimes \ket{0}_\calW$$\label{step:hyb-0-copy-input-2}
        \end{enumerate}

\paragraph{Hybrid 1.} We've canceled out operations with their inverse whereever possible. We have also moved the bitflip on $\calQ_b$ to the final step.

The oracle initializes its internal registers $\calW \times \calD = \ket{0} \times \ket{\emptyset}$. On query $\ket{x,u,b}_\calQ$, the oracle acts as follows:
        \begin{enumerate}
            \item If there exists some $x'\neq x$ such that $(x',r) \in D$ for some $r \in \calR$, skip the following steps.
            \item Apply $\StdDecomp$ to registers $\calQ_x \times \calD$.
            \item Apply the following unitary map:
            \begin{align*}
                \ket{x, u, b}_{\calQ} \otimes \ket{x,r}_\calD \mapsto \ket{x, u \oplus f(x; r), b}_{\calQ} \otimes \ket{x,r}_\calD
            \end{align*}
            \item Apply $\StdDecomp$ to registers $\calQ_x \times \calD$.
            \item Apply $\mathsf{X}$ (bitflip) to $\calQ_b$.
        \end{enumerate}

\begin{claim}
    The oracles in hybrids 0 and 1 are perfectly indistinguishable after any number of queries.
\end{claim}
\begin{proof}
    First, in hybrid 0, $\StdDecomp$ applies $\StdDecomp_x$ to $\calD$ controlled on the value of $x$ written on $\calW$. In hybrid 1, $\StdDecomp$ applies $\StdDecomp_x$ to $\calD$ controlled on the value of $x$ written on $\calQ_x$. These operations are equivalent because during any query in hybrid 0, the value of $x$ written on $\calQ_x$ matches the value written on $\calW$ whenever $\StdDecomp$ is called.

    Let us modify hybrid 0 so that $\StdDecomp$ is controlled on $\calQ_x$ instead of $\calW$.

    Second, $\StdDecomp$ now commutes with \cref{step:hyb-0-copy-input-1,step:hyb-0-copy-output,step:hyb-0-copy-input-2} of hybrid 0. This is because $\StdDecomp$ applies $\StdDecomp_x$ to $\calD$ controlled on the computational-basis value of $\calQ_x$. Step \ref{step:hyb-0-copy-input-1} applies a unitary to $\calW$ controlled on the computational-basis value of $\calQ_x$. $\StdDecomp$ and step \ref{step:hyb-0-copy-input-1} commute because they are controlled by the computational-basis value of the same register $\calQ_x$, and they apply a controlled unitary to disjoint registers, $\calD$ and $\calW$ respectively. The same reasoning shows that $\StdDecomp$ commutes with step \ref{step:hyb-0-copy-input-2}. Step \ref{step:hyb-0-copy-output} does not act on $\calQ_x$ or $D$, so it commutes with $\StdDecomp$ because they act on disjoint registers.

    Let us switch the order of step \ref{step:hyb-0-copy-input-1} and the first $\StdDecomp$ in step \ref{step:hyb-0-query-CO-1}. Let us also switch the order of step \ref{step:hyb-0-copy-input-2} and the last $\StdDecomp$ in step \ref{step:hyb-0-query-CO-2}.

    Third, the final application of $\StdDecomp$ in step \ref{step:hyb-0-query-CO-1} cancels out with the first application of $\StdDecomp$ in step \ref{step:hyb-0-query-CO-2}. This is because $\StdDecomp^{-1} = \StdDecomp$, and the only operation that occurs in between these two applications of $\StdDecomp$ is step \ref{step:hyb-0-copy-output}, which commutes with $\StdDecomp$.

    Fourth, after making the changes above, we have that hybrid 0 is equivalent to the following operations:
    \begin{enumerate}
        \item If there exists some $x'\neq x$ such that $(x',r) \in D$ for some $r \in \calR$, skip the following steps.
        \item Apply $\StdDecomp$ to $\calQ_x \times \calD$.
        \item Step \ref{step:hyb-0-copy-input-1}: $\ket{x}_{\calQ_x} \otimes \ket{0}_\calW \to \ket{x}_{\calQ_x} \otimes \ket{x,0}_\calW$
        \item $\mathsf{CStO}'$: $\ket{x,0}_\calW \otimes \ket{x,r}_{\calD} \to \ket{x,r}_\calW \otimes \ket{x,r}_{\calD}$
        \item Step \ref{step:hyb-0-copy-output}: $\ket{u, b}_{\calQ_u \times \calQ_b} \otimes \ket{x, r}_{\calW} \mapsto \ket{u \oplus f(x; r), b \oplus 1}_{\calQ_u \times \calQ_b} \otimes \ket{x, r}_{\calW}$
        \item $\mathsf{CStO}'$: $\ket{x,r}_\calW \otimes \ket{x,r}_{\calD} \to \ket{x,0}_\calW \otimes \ket{x,r}_{\calD}$
        \item Step \ref{step:hyb-0-copy-input-2}: $\ket{x}_{\calQ_x} \otimes \ket{x,0}_\calW \to \ket{x}_{\calQ_x} \otimes \ket{0}_\calW$
        \item Apply $\StdDecomp$ to $\calQ_x \times \calD$.
    \end{enumerate}
    Note that the middle steps, $\text{Step \ref{step:hyb-0-copy-input-2}} \circ \mathsf{CStO}' \circ \text{Step \ref{step:hyb-0-copy-output}} \circ \mathsf{CStO}' \circ \text{Step \ref{step:hyb-0-copy-input-1}}$, are all classical operations, and they are equivalent to the following operation:
    \[\ket{x,u,b}_{\calQ} \otimes \ket{0}_\calW \otimes \ket{x,r}_{\calD} \to \ket{x, u \oplus f(x;r), b \oplus 1}_{\calQ} \otimes \ket{0}_\calW \otimes \ket{x,r}_{\calD}\]
    Note that $\calW$ is returned to the $\ket{0}$ state by the end of this sequence of operations, so we can ignore it.
    
    Fifth, let us delay applying the bitflip to $\calQ_b$ until after all the other steps have finished. This is equivalent to the steps above because the bitflip commutes with all other steps. 
    
    Now we have that hybrid 0 is equivalent to the following steps:
    \begin{enumerate}
        \item If there exists some $x'\neq x$ such that $(x',r) \in D$ for some $r \in \calR$, skip the following steps.
        \item Apply $\StdDecomp$ to $\calQ_x \times \calD$.
        \item Apply the following unitary map:
        \[\ket{x,u,b}_{\calQ} \otimes \ket{x,r}_{\calD} \to \ket{x, u \oplus f(x;r),b}_{\calQ} \otimes \ket{x,r}_{\calD}\]
        \item Apply $\StdDecomp$ to $\calQ_x \times \calD$.
        \item Apply $\mathsf{X}$ (bitflip) to $\calQ_b$.
    \end{enumerate}
    This is exactly hybrid 1, so we've shown that hybrids 0 and 1 are perfectly indistinguishable.
\end{proof}

\paragraph{Hybrid 2.} Let us split the database register $\calD$ into two registers, called $\calD \times \calP$ to match the eventual notation of $\Sim'(P_f)$. $\calD$ will store the database's $x$-value, and $\calP$ will store its $r$-value.

Let $\Record$ be the unitary that swaps the states $\ket{x}_{\calQ_x} \otimes \ket{\emptyset}_{\calD} \ket{+}_{\calP}$ and $\ket{x}_{\calQ_x} \otimes \ket{x}_{\calD} \ket{+}_{\calP}$, and acts as the identity on the rest of the space. Additionally, let $\Decomp$ be the unitary that swaps $\ket{\emptyset}_{\calP}$ and $\ket{+}_{\calP}$.

In this hybrid, the oracle initializes database registers $\calD \times \calP = \ket{\emptyset, \emptyset}$ and work registers $\calW_x \times \calW_y \times \calW_b = \ket{\emptyset, \emptyset, 0}$, although the oracle never operates on $\calW_x \times \calW_y \times \calW_b$.

Next, the oracle answers queries on $\ket{x, u, b}_\calQ \otimes \ket{x'}_\calD$ as follows:
\begin{enumerate}
    \item If $x' \notin \{x, \emptyset\}$, then skip the following steps.
    \item Apply $\textcolor{red}{\Record \circ \Decomp}$ to the $(\calQ_x,\calD, \calP)$ registers.
    \item Apply $\Eval_f$ on the $(\calQ_x, \calQ_u, \calP)$ registers:
    \[\ket{x,u}_{\calQ_x \times \calQ_u} \otimes \ket{r}_{\calP} \to \ket{x, u \oplus f(x;r)}_{\calQ_x \times \calQ_u} \otimes \ket{r}_{\calP}\]
    \item Apply $\textcolor{red}{\Decomp \circ \Record}$ on the $(\calQ_x,\calD, \calP)$ registers.
    \item Apply $\mathsf{X}$ (bitflip) to $\calQ_b$.
\end{enumerate}

\begin{claim}
    The oracles in hybrids 1 and 2 are perfectly indistinguishable after any number of quantum queries.
\end{claim}
\begin{proof}
We claim that the unitaries defined in hybrids 1 and 2 are identical. We argue this by choosing a suitable basis for states on registers $\calQ_x \times \calD$ (hybrid 1) or $\calQ \times \calP \times \calD$ (hybrid 2) and then tracing how the unitaries $\StdDecomp$, $\Eval_f$, $\Record \circ \Decomp$ and $\Decomp \circ \Record$ transform these basis states.

The unitaries in hybrids 1 and 2 both act as the identity on basis states of the form $\ket{x}_{\calQ_x} \ket{x'}_{\calD} \ket{\cdots}$ for $x' \notin \{x,\emptyset\}$, so we will restrict our attention to states supported on $\ket{x}_{\calQ_x} \ket{x}_{\calD}$ and $\ket{x}_{\calQ_x} \ket{\emptyset}_{\calD}$. For the $\calP$ register, our orthogonal ``basis'' will be $\ket{\emptyset}$, $\ket{+}$ and $\ket{-}$, where, in an abuse of notation, we write $\ket{-}$ to denote any state in the span of $\{\frac{1}{\sqrt{\abs{\calR}}} \cdot \sum_{r \in \calR} (-1)^{r \cdot u}\ket{r} \}_{u \neq 0}$.

\begin{itemize}
    \item The starting state of the first query is $\ket{x}_{\calQ_x} \ket{\emptyset}_\calD \ket{\emptyset}_\calP$, and both $\StdDecomp$ and $\Record \circ \Decomp$ map it to $\ket{x} \ket{x} \ket{+}$.
    \item $\Eval_f$ maps $\ket{x}\ket{x} \ket{+}$ and $\ket{x}\ket{x} \ket{-}$ to superpositions of $\ket{x} \ket{x} \ket{+}$ and $\ket{x} \ket{x} \ket{-}$ (where $\ket{-}$ does not necessarily refer to the same state as before the application of $\Eval_f$).
    \item Both $\StdDecomp$ and $\Decomp \circ \Record$ map these states as follows:
        \begin{align*}
        \ket{x} \ket{x} \ket{+} &\mapsto \ket{x} \ket{\emptyset} \ket{\emptyset}, \text{ and}\\
        \ket{x} \ket{x} \ket{-} &\mapsto \ket{x} \ket{x} \ket{-}.
        \end{align*}
    \item Note that $\StdDecomp$, $\Record \circ \Decomp$, and $\Decomp \circ \Record$ each leave $\ket{x} \ket{x} \ket{-}$ invariant.
\end{itemize}

Let us define the set of \textit{reachable} states to be states whose database is supported on eigenstates of the form: $\ket{x, -}, \ket{\emptyset, \emptyset}$. We claim that the database is in a reachable state at the start of any query. First, the state of the database at the start of the first query is reachable. Second, the observations above imply that if the database is in a reachable state at the start of a query, then it will end the query in a reachable state.

The only difference in the operations applied by hybrids 1 and 2 is that hybrid 1 uses $\StdDecomp$ twice, and hybrid 2 uses $\Record \circ \Decomp$ and $\Decomp \circ \Record$ instead.

Let us step through a query in hybrids 1 or 2 when the starting state is reachable. The first time that $\StdDecomp$ or $\Record \circ \Decomp$ is called during the query, the database is supported on states of the form $\ket{\emptyset, \emptyset}, \ket{x,-}$. We've already argued that $\StdDecomp$, $\Record \circ \Decomp$ act equivalently on such states. After this operation, the database is supported on states of the form $\ket{x, +}, \ket{x,-}$. Then $\Eval_f$ maps these states to superpositions of $\ket{x, +}, \ket{x,-}$. Next, $\StdDecomp$ or $\Decomp \circ \Record$ is called. We've argued that $\StdDecomp$ and $\Decomp \circ \Record$ act equivalently on database states of the form $\ket{x, +}, \ket{x,-}$.

This shows that the unitaries in hybrids 1 and 2 act equivalently on any reachable state, so hybrids 1 and 2 are perfectly indistinguishable.
\end{proof}

\paragraph{Hybrid 3.} This is $\Sim'(P_f)$. 

In this hybrid, the oracle initializes internal registers $\calD \times \calP \times \calW_x \times \calW_y \times \calW_b = \ket{\emptyset} \otimes \textcolor{red}{\ket{+}} \otimes \ket{\emptyset, \emptyset, 0}$, and answers queries on $\ket{x, u, b} \otimes \ket{x'}_\calD$ as follows:
\begin{enumerate}
    \item If $x' \notin \{x, \emptyset\}$, do nothing and ignore the remaining steps.
    \item Apply $\textcolor{red}{\ReflectCopy}$ to the $(\calQ_x,\calD, \calP, \textcolor{red}{\calW_x, \calW_y, \calW_b})$ registers.
    \item Apply $\textcolor{red}{\EvalCopy}$ on the $(\calQ_x, \calQ_u, \calP, \textcolor{red}{\calW_x, \calW_y})$ registers.
    \item Apply $\textcolor{red}{\ReflectCopy}$ to the $(\calQ_x,\calD, \calP, \textcolor{red}{\calW_x, \calW_y, \calW_b})$ registers.
    \item Apply $\Bitflip$ to $\calQ_b$.
\end{enumerate}

\begin{claim}
    The oracles in hybrids 2 and 3 are perfectly indistinguishable after any number of quantum queries.
\end{claim}
\begin{proof}
First, consider program $P_f$. For this program, $\Record$ is equivalent to $\ReflectCopy$. The latter copies $x$ from $\calQ_x$ to $\calD$ controlled on $\calP \times \calW_x \times \calW_y \times \calW_b$ being in the initial state (\cref{thm:reflect-and-copy}). $\Record$ copies $x$ from $\calQ_x$ to $\calD$ controlled on $\calP$ being in its initial state $\ket{+}$. Furthermore, every time that $\Record$ is called in hybrid 2, the $\calW_x \times \calW_y \times \calW_b$ registers are in their initial state because hybrid 2 does not operate on these registers after initializing them. Therefore, we can replace $\Record$ with $\ReflectCopy$.

Second, $\Eval_f$ is equivalent to $\EvalCopy$ for program $P_f$. If $\calW_x \times \calW_y$ are in the state $\ket{\emptyset,\emptyset}$ at the start of $\EvalCopy$, then $\EvalCopy$ acts as follows:
\begin{align*}
    &\ket{x,u}_{\calQ_x \times \calQ_u} \otimes \ket{\emptyset,\emptyset}_{\calW_x \times \calW_y} \otimes \ket{r}_\calP\\
    &\to \ket{x,u}_{\calQ_x \times \calQ_u} \otimes \ket{x,\emptyset}_{\calW_x \times \calW_y} \otimes \ket{r}_\calP\\
    &\to \ket{x,u}_{\calQ_x \times \calQ_u} \otimes \ket{x,f(x;r)}_{\calW_x \times \calW_y} \otimes \ket{r}_\calP\\
    &\to \ket{x,u \oplus f(x;r)}_{\calQ_x \times \calQ_u} \otimes \ket{x,f(x;r)}_{\calW_x \times \calW_y} \otimes \ket{r}_\calP\\
    &\to \ket{x,u \oplus f(x;r)}_{\calQ_x \times \calQ_u} \otimes \ket{x,\emptyset}_{\calW_x \times \calW_y} \otimes \ket{r}_\calP\\
    &\to \ket{x,u \oplus f(x;r)}_{\calQ_x \times \calQ_u} \otimes \ket{\emptyset,\emptyset}_{\calW_x \times \calW_y} \otimes \ket{r}_\calP
\end{align*}
We claim that $\calW_x \times \calW_y = \ket{\emptyset,\emptyset}$ at the start of every $\EvalCopy$ query. Before any queries have been made, $\calW_x \times \calW_y = \ket{\emptyset,\emptyset}$. Next, the computational-basis values of $\calW_x \times \calW_y$ are only modified during $\EvalCopy$. We showed above that $\EvalCopy$ returns $\calW_x \times \calW_y$ to $\ket{\emptyset,\emptyset}$, so $\calW_x \times \calW_y = \ket{\emptyset,\emptyset}$ at the start of every $\EvalCopy$ query.

We also showed that $\EvalCopy$ is equivalent to 
\[\ket{x,u}_{\calQ_x \times \calQ_u} \otimes \ket{r}_\calP \to \ket{x,u \oplus f(x;r)}_{\calQ_x \times \calQ_u} \otimes \ket{r}_\calP\]
which is the application of $\Eval_f$ to $\calQ_x \times \calQ_u \times \calP$.

Third, let us remove all the $\Decomp$ operations from hybrid 2 as follows. The $\Decomp$ operation at the end of one query cancels with the $\Decomp$ operation at the start of the next one. This is because $\Decomp \circ \Decomp = I$, and $\Decomp$ operates only on the internal registers of the oracle. 

Next, we can remove the first $\Decomp$ from the first query. This $\Decomp$ operation acts on $\ket{\emptyset}_{\calP}$ to produce the state $\ket{+}_{\calP}$. In hybrid 3, we remove this $\Decomp$ operation and instead have the oracle start with state $\ket{+}_{\calP}$.

Finally, we can remove the final $\Decomp$ from the final query. $\Decomp$ acts only on the internal registers of the oracle, so the view of a distinguisher that has query access to the oracle is unchanged.

In summary, we have transformed hybrid 2 into hybrid 3 and shown that the two hybrids are equivalent.
\end{proof}
\end{proof}

\begin{lemma}\label{lem:testable-bpotp-sim-p1-p2-perfect}
    If $P_1$ and $P_2$ are testable quantum programs that implement the same functionality with $0$ error, then oracles for $\Sim'(P_1)$ and $\Sim'(P_2)$ are perfectly indistinguishable to any adversary making an unbounded number of quantum queries.
\end{lemma}

The intuition for this proof is that given two program states $\ket{\psi^1}, \ket{\psi^2}$ that implement the same sampling functionality, for each $x$, we can decompose each program state into a superposition over an orthonormal basis of programs that give deterministic output values. Since the two programs implement the same functionality, these bases are equivalent up to a rotation/labeling.

This nice picture breaks down when you consider querying the program on multiple $x$ values, but the database that $\Sim'$ uses to record queries allows us to get around this issue and still make the argument go through.

\begin{proof}

First, for any program $P$ and any $(x, y) \in \calX \times \calY$, let $\alpha_{x,y} = \sqrt{\Pr[P(x) \to y]}$.

In $\Sim'(P)$, the initial state of $\calP \times \calW_x \times \calW_y$ is $\ket{\psi, \emptyset, \emptyset}$. The following claim says that we can decompose $\ket{\psi, \emptyset, \emptyset}$ into a basis defined by the $y$-values that the program outputs. For each $y$-value, $\ket{\psi_{x,y}}$ is the component of $\ket{\psi, \emptyset, \emptyset}$ that produces output $y$ when $\EvalCopy$ is called.

\begin{claim}\label{thm:EvalCopy}
    For any $x \in \calX$, there exists a set of orthonormal states $\{\ket{\psi_{x, y}}\}_{y \in \calY}$ on registers $\calP \times \calW_x \times \calW_y$ such that for any $u \in \calY \cup \{\emptyset\}$ and any $y$ in the support of $P(x)$,
    \begin{align*}
        \EvalCopy \ket{x,u}_{\calQ_x \times \calQ_u} \otimes \ket{\psi_{x,y}} &= \ket{x, u + y}_{\calQ_x \times \calQ_u} \otimes \ket{\psi_{x,y}}
    \end{align*}
    Additionally, $\ket{\psi, \emptyset, \emptyset}$ is in the span of $\{\ket{\psi_{x, y}}\}_{y}$:
    \[\ket{\psi, \emptyset, \emptyset} = \sum_{y \in \calY} \alpha_{x, y} \ket{\psi_{x,y}}\]
\end{claim}
\begin{proof}

First, note that $\EvalCopy$ does not modify the computational-basis value of $\calQ_x$. It just applies an operation that is controlled by $\calQ_x$. When $\calQ_x$ contains value $x$, the $\EvalCopy$ step acts as follows on registers $\calQ_u \times \calW_x \times \calW_y \times \calP$:
\[\Eval_x^\dag \circ \CopyOutput \circ \Eval_x\]
where we define $\Eval_x$ to be the operation that $\CNOT$s $x$ onto register $\calW_x$ and then applies $\Eval$ to $\calW_x \times \calW_y \times \calP$.

Second, for each $x \in \calX$ and each $y$ in the support of $P(x)$, let $\ket{\psi_{x,y}}$ be the component of $\ket{\psi, \emptyset, \emptyset}$ that produces output $y$ when the program is queried on $x$. Formally, for each $x \in \calX$ and each $y$ in the support of $P(x)$,
\[\text{let } \ket{\psi_{x,y}} = \frac{1}{\alpha_{x,y}} \cdot \Eval_x^\dag \cdot (\mathbb{I}_{\calP \times \calW_x} \otimes \ketbra{y}_{\calW_y}) \cdot \Eval_x \cdot (\ket{\psi, \emptyset, \emptyset}_{\calP \times \calW_x \times \calW_y})\]
This state is well-defined because for all $y$ in the support of $P(x)$, $\alpha_{x,y} \neq 0$.

The state $\ket{\psi_{x,y}}$ defined above has unit norm. We know that when $P$ is evaluated on $x$, the value written to $\calW_y$ will be $y$ with probability $\alpha_{x,y}^2 = \Pr[P(x) \to y]$. That means 
\begin{align*}
    \left\|(\mathbb{I}_{\calP \times \calW_x} \otimes \ketbra{y}_{\calW_y}) \cdot \Eval_x \cdot \ket{\psi, \emptyset, \emptyset}\right\| &= \alpha_{x,y}
\end{align*}
Next, since $\Eval_x$ is
a unitary,
\begin{align*}
    \left\|\ket{\psi_{x,y}}\right\| &= \left\|\frac{1}{\alpha_{x,y}} \cdot \Eval_x^\dag \cdot (\mathbb{I}_{\calP \times \calW_x} \otimes \ketbra{y}_{\calW_y}) \cdot \Eval_x \cdot \ket{\psi, \emptyset, \emptyset}\right\|\\
    &= \frac{1}{\alpha_{x,y}} \cdot \left\|(\mathbb{I}_{\calP \times \calW_x} \otimes \ketbra{y}_{\calW_y}) \cdot \Eval_x \cdot \ket{\psi, \emptyset, \emptyset}\right\|\\
    &= \frac{1}{\alpha_{x,y}} \cdot \alpha_{x,y} = 1
\end{align*}

Third, let us define $\ket{\psi_{x,y}}$ for any $y$ not in the support of $P(x)$. The definition above will not suffice here because $\alpha_{x,y}=0$.
\[\text{Let } \ket{\psi_{x,y}} = \Eval_x^\dag \cdot (\ket{\psi}_\calP \ket{\emptyset}_{\calW_x}\ket{y}_{\calW_y})\]
This state is well-defined and has unit norm because $\Eval_x^\dag$ is a unitary.

Now we will prove some properties of $\{\ket{\psi_{x,y}}\}_{y \in \calY}$.

Fourth, for any $x \in \calX$, the states $\{\ket{\psi_{x,y}}\}_{y \in \calY}$ are orthogonal. This is because for any $y \in \calY$, $\Eval_x \ket{\psi_{x,y}}$ is in the span of $\mathbb{I}_{\calP \times \calW_x} \otimes \ketbra{y}_{\calW_y}$. For any two distinct values $y, y' \in \calY$, $\mathbb{I}_{\calP \times \calW_x} \otimes \ketbra{y}_{\calW_y}$ and $\mathbb{I}_{\calP \times \calW_x} \otimes \ketbra{y'}_{\calW_y}$ project onto orthogonal subspaces, so $\Eval_x \ket{\psi_{x,y}}$ and $\Eval_x\ket{\psi_{x,y'}}$ are orthogonal. Finally, $\Eval_x$ is a unitary, so it preserves inner products, and $\ket{\psi_{x,y}}$ and $\ket{\psi_{x,y'}}$ are orthogonal as well.

Fifth, $\ket{\psi, \emptyset, \emptyset}$ is in the span of $\{\ket{\psi_{x, y}}\}_{y \in \calY}$.
\begin{align*}
    \ket{\psi, \emptyset, \emptyset} &= \Eval_x^\dag \cdot \Eval_x \cdot \ket{\psi, \emptyset, \emptyset}\\
    &= \Eval_x^\dag \cdot (\mathbb{I}_{\calP \times \calW_x} \otimes \mathbb{I}_{\calW_y}) \cdot \Eval_x \cdot \ket{\psi, \emptyset, \emptyset}\\
    &= \Eval_x^\dag \cdot \left(\mathbb{I}_{\calP \times \calW_x} \otimes \sum_{y \in \calY \cup \{\emptyset\}} \ketbra{y}\right) \cdot \Eval_x \cdot \ket{\psi, \emptyset, \emptyset}\\
    &= \sum_{y \in \calY \cup \{\emptyset\}} \Eval_x^\dag \cdot (\mathbb{I}_{\calP \times \calW_x} \otimes \ketbra{y}) \cdot \Eval_x \cdot \ket{\psi, \emptyset, \emptyset}\\
    &= \sum_{y \in \calY} \alpha_{x, y} \ket{\psi_{x,y}}
\end{align*}
We used the fact that 
\[\mathbb{I}_{\calP \times \calW_x} \otimes \left(\sum_{y \in \calY \cup \{\emptyset\}} \ketbra{y}\right) = \sum_{y \in \calY \cup \{\emptyset\}} \mathbb{I}_{\calP \times \calW_x} \otimes \ketbra{y}\]

We also used the fact that $(\mathbb{I}_{\calP \times \calW_x} \otimes \ketbra{\emptyset}) \cdot \Eval_x \cdot \ket{\psi, \emptyset, \emptyset} = \mathbf{0}$ because $P(x)$ never outputs $\emptyset$. Likewise, for any $y$ not in the support of $P(x)$, $(\mathbb{I}_{\calP \times \calW_x} \otimes \ketbra{y}) \cdot \Eval_x \cdot \ket{\psi, \emptyset, \emptyset} = \mathbf{0}$.

Sixth, let us show that for any $y$ in the support of $P(x)$,
\[\EvalCopy \ket{x,u}_{\calQ_x \times \calQ_u} \otimes \ket{\psi_{x,y}} = \ket{x, u + y}_{\calQ_x \times \calQ_u} \otimes \ket{\psi_{x,y}}\]

$\EvalCopy$ does not modify the computational basis value of $\calQ_x$ and acts as $\Eval_x^\dag \circ \CopyOutput \circ \Eval_x$ on the remaining registers. Let us apply $\Eval_x$:
\ifllncs
\begin{align*}
    \Eval_x (\ket{u}_{\calQ_u} \otimes \ket{\psi_{x,y}})
    &= \ket{u}_{\calQ_u} \otimes \Eval_x \cdot \Biggl(
    \frac{1}{\alpha_{x,y}} \cdot \Eval_x^\dag \cdot (\mathbb{I}_{\calP \times \calW_x} \otimes \ketbra{y}_{\calW_y})\\
    &\quad\cdot \Eval_x \cdot \ket{\psi, \emptyset, \emptyset}
    \Biggr)\\
    &= \frac{1}{\alpha_{x,y}} \cdot \ket{u}_{\calQ_u} \otimes \Bigl(
    (\mathbb{I}_{\calP \times \calW_x} \otimes \ketbra{y}_{\calW_y}) \cdot \Eval_x \cdot \ket{\psi, \emptyset, \emptyset}
    \Bigr)
\end{align*}
\else
\begin{align*}
    \Eval_x (\ket{u}_{\calQ_u} \otimes \ket{\psi_{x,y}}) &= \ket{u}_{\calQ_u} \otimes \Eval_x \cdot \left(\frac{1}{\alpha_{x,y}} \cdot \Eval_x^\dag \cdot (\mathbb{I}_{\calP \times \calW_x} \otimes \ketbra{y}_{\calW_y}) \cdot \Eval_x \cdot \ket{\psi, \emptyset, \emptyset}\right)\\
    &= \frac{1}{\alpha_{x,y}} \cdot \ket{u}_{\calQ_u} \otimes \left((\mathbb{I}_{\calP \times \calW_x} \otimes \ketbra{y}_{\calW_y}) \cdot \Eval_x \cdot \ket{\psi, \emptyset, \emptyset}\right)
\end{align*}
\fi
Next, if we apply $\CopyOutput$, this copies $y$ from $\calW_y$ to $\calQ_u$:
\ifllncs
\begin{align*}
    \CopyOutput \cdot \Eval_x \cdot \left(\ket{u} \otimes \ket{\psi_{x,y}}\right)
    &= \frac{1}{\alpha_{x,y}} \cdot \ket{u + y} \otimes \Bigl(
    (\mathbb{I}_{\calP \times \calW_x} \otimes \ketbra{y}_{\calW_y})\\
    &\quad\cdot \Eval_x \cdot \ket{\psi, \emptyset, \emptyset}
    \Bigr)
\end{align*}
\else
\begin{align*}
    \CopyOutput \cdot \Eval_x \cdot \left(\ket{u} \otimes \ket{\psi_{x,y}}\right) &= \frac{1}{\alpha_{x,y}} \cdot \ket{u + y} \otimes \left((\mathbb{I}_{\calP \times \calW_x} \otimes \ketbra{y}_{\calW_y}) \cdot \Eval_x \cdot \ket{\psi, \emptyset, \emptyset}\right)
\end{align*}
\fi
Finally, let us apply $\Eval_x^\dag$:
\ifllncs
\begin{align*}
    &\Eval_x^\dag \cdot \CopyOutput \cdot \Eval_x \cdot \left(\ket{u} \otimes \ket{\psi_{x,y}}\right)\\
    &= \frac{1}{\alpha_{x,y}} \cdot \ket{u + y} \otimes \Bigl(
    \Eval_x^\dag \cdot (\mathbb{I}_{\calP \times \calW_x} \otimes \ketbra{y}_{\calW_y})\\
    &\quad\cdot \Eval_x \cdot \ket{\psi, \emptyset, \emptyset}
    \Bigr)\\
    &= \ket{u + y} \otimes \ket{\psi_{x,y}}
\end{align*}
\else
\begin{align*}
    \Eval_x^\dag \cdot \CopyOutput \cdot \Eval_x \cdot \left(\ket{u} \otimes \ket{\psi_{x,y}}\right) &= \frac{1}{\alpha_{x,y}} \cdot \ket{u + y} \otimes \left(\Eval_x^\dag \cdot (\mathbb{I}_{\calP \times \calW_x} \otimes \ketbra{y}_{\calW_y}) \cdot \Eval_x \cdot \ket{\psi, \emptyset, \emptyset}\right)\\
    &= \ket{u + y} \otimes \ket{\psi_{x,y}}
\end{align*}
\fi
In summary,
\[\EvalCopy \ket{x,u}_{\calQ_x \times \calQ_u} \otimes \ket{\psi_{x,y}} = \ket{x, u + y}_{\calQ_x \times \calQ_u} \otimes \ket{\psi_{x,y}}\]
\end{proof}

Let $S_x$ be the span of  $\{\ket{\psi_{x,y}}\}_{y \in \calY}$, and let $S_{\perp,x}$ be the subspace of $S_x$ that is orthogonal to $\ket{\psi,\emptyset,\emptyset}$. Since $\ket{\psi_{x,y}} \in S_{x}$, $\ket{\psi_{x,y}}$ can be written as a superposition of $\ket{\psi, \emptyset, \emptyset}$ and a state $\ket{\psi^{\perp}_{x,y}} \in S_{\perp,x}$.\\

For programs $P_1 = (\ket{\psi^1}, \Eval_1, R_1)$ and $P_2 = (\ket{\psi^2}, \Eval_2, R_2)$, let us use superscripts or subscripts for the variables defined above. For example, let $\ket{\psi^1_{x,y}}$ and $\ket{\psi^2_{x,y}}$ be the state $\ket{\psi_{x,y}}$ defined for programs $P_1$ and $P_2$, respectively. Let $\ReflectCopy_1, \EvalCopy_1, \ReflectCopy_2, \EvalCopy_2$ be defined analogously.

For each $x$, let us define a unitary $U_x$ that maps
    \begin{align*}
        U_x : \ket{\psi^1_{x, y}} \mapsto \ket{\psi^2_{x, y}}
    \end{align*}
    for all $y \in \calY$. Such a unitary exists because $\{\ket{\psi^1_{x, y}}\}_{y}$ and $\{\ket{\psi^2_{x, y}}\}_{y}$ are each orthonormal. Next, let us define $\CU$ to be a controlled version of $U_x$ that reads the value of $x$ from $\calQ_x$ and applies $U_x$ to $\calP \times \calW_x \times \calW_y$.

    \begin{claim}\label{thm:CU}
        For any $x \in \calX$, $U_x \cdot \ket{\psi^1,\emptyset,\emptyset} = \ket{\psi^2,\emptyset, \emptyset}$.
    \end{claim}
    \begin{proof}
    First, for any $(x,y)$, $\alpha^1_{x,y} = \sqrt{\Pr[P_1(x)\to y]} = \sqrt{\Pr[P_2(x)\to y]} = \alpha^2_{x,y}$.

    Second,
    \begin{align*}
        U_x \cdot \ket{\psi^1,\emptyset, \emptyset} &= U_x \cdot \sum_{y} \alpha^1_{x,y} \ket{\psi^1_{x,y}}\\
        &= \sum_{y} \alpha^1_{x,y} \cdot U_x \cdot \ket{\psi^1_{x,y}}\\
        &= \sum_{y} \alpha^1_{x,y} \cdot \ket{\psi^2_{x,y}}\\
        &= \sum_{y} \alpha^2_{x,y} \ket{\psi^2_{x,y}}\\
        &= \ket{\psi^2,\emptyset, \emptyset} 
\end{align*}
\end{proof}

Now let us define some hybrids to transform $\Sim'(P_1)$ into $\Sim'(P_2)$.

\paragraph{Hybrid 1.} $\Sim'(P_1)$

\paragraph{Hybrid 2.} At a high level, in this hybrid we operate on initial state $\ket{\psi^1}$ but on each query it transforms the state into $\ket{\psi^2}$ and acts with the $P_2$ evaluation unitary. An important detail is that transforming from $\ket{\psi^1}$ to $\ket{\psi^2}$ can only be done with respect to some $x$ -- the $x$-value given as input to the query.
\begin{itemize}
    \item Initialize the program register with $\ket{\psi^1}$.
    \item On queries of the form $\ket{x, u, b} \ket{D}$ with $x' \in D$ such that $x' \neq x$, act as the identity and skip the following steps.
    \item \textcolor{red}{Apply $\CU$, which applies $U_x$ to $\calP \times \calW_x \times \calW_y$ controlled on the value $x$ written on $\calQ_x$.}
    \item Apply the unitary
    \begin{align*}
        \Bitflip \circ \ReflectCopy_{\textcolor{red}{2}} \circ \EvalCopy_{\textcolor{red}{2}} \circ \ReflectCopy_{\textcolor{red}{2}}
    \end{align*}
    \item \textcolor{red}{Apply $\CU^\dag$.}
\end{itemize}

\paragraph{Hybrid 3.} $\Sim'(P_2)$\\

Before showing the indistinguishability of these hybrids, we make the following claim. 

\begin{claim}\label{clm:invariant}
 At the end of any query in hybrids 1 and 2, the database and program registers maintain the following invariant: $D = \emptyset$ if and only if $\calP \times \calW_x \times \calW_y$ contain the initial program state $\ket{\psi^1, \emptyset, \emptyset}$, and $D = x$ if and only if $\calP \times \calW_x \times \calW_y$ are in some state $\ket{\phi} \in S^1_{\perp,x}$.
\end{claim}
 
\begin{proof}
    We first prove this invariant for hybrid 1 using induction. For the base case: the oracle is initialized in the state $\ket{\emptyset}_{\calD} \ket{\psi^1, \emptyset, \emptyset}_{\calP \times \calW_x \times \calW_y}$, so the invariant is satisfied to begin with. For the inductive case, let us assume that at the start of a query, the invariant is satisfied. We will show that the invariant is still satisfied at the end of the query. 

    First, if $\calQ_x = \ket{x}$, $\calD = \ket{x'}$, and $x' \notin \{\emptyset, x\}$, then this query aborts, and the state on $\calD \times \calP \times \calW_x \times \calW_y$ still satisfy the invariant.
    
    Next, let us consider the case where $\calQ_x = \ket{x}$, and $\calD$ contains either $x$ or $\emptyset$. Let us also assume that the state satisfies the invariant at the start of the query. First, by \cref{thm:reflect-and-copy}, $\ReflectCopy$ acts as follows.
\begin{align*}
    \ReflectCopy :& \ket{\emptyset} \ket{\psi^1, \emptyset, \emptyset} \mapsto \ket{x} \ket{\psi^1, \emptyset, \emptyset}\\
    &\ket{x} \ket{\phi} \mapsto \ket{x} \ket{\phi}
\end{align*}
where $\ket{\phi} \in S^1_{\perp,x}$. Therefore, at the end of $\ReflectCopy$, $\calD = \ket{x}$, and $\calP \times \calW_x \times \calW_y$ are in $S^1_x$.

Second, apply $\EvalCopy$ to the state
\[\ket{x}_\calD \ket{\psi^1_{x,y}}_{\calP \times \calW_x \times \calW_y}.\]
By \cref{thm:EvalCopy}, this leaves the state on the $\calD \times \calP \times \calW_x \times \calW_y$ registers unchanged. This implies that $\EvalCopy$ maps states in $S^1_x$ to states in $S^1_x$. Additionally, $\EvalCopy$ does not touch $\calD$. Therefore, at the end of $\EvalCopy$, $\calD$ still contains $x$, and $\calP \times \calW_x \times \calW_y$ still lies in $S^1_x$.

Third, we apply $\ReflectCopy$ again. At the start of this operation, $\calD \times \calP \times \calW_x \times \calW_y$ is in the span of $\ket{x}\ket{\psi^1, \emptyset, \emptyset}$ and $\ket{x}\ket{\phi}$ for states $\ket{\phi} \in S^1_{\perp, x}$. $\ReflectCopy$ acts on these basis states as follows:

\begin{align*}
    \ReflectCopy :& \ket{x} \ket{\psi^1, \emptyset, \emptyset} \mapsto \ket{\emptyset} \ket{\psi^1, \emptyset, \emptyset}\\
    &\ket{x} \ket{\phi} \mapsto \ket{x} \ket{\phi}
\end{align*}
Therefore, the state at the end of this operation is in the span of states that satisfy the invariant.

Fourth, we apply $\Bitflip$. This operation acts only on $\calQ_b$, so the state at the end of this operation will still satisfy the invariant.

 A very similar argument also shows that this invariant holds for hybrid 2.
\end{proof}

\begin{claim}\label{thm:input-to-eval-copy}
    At the start of any invocation of $\EvalCopy_1$ in hybrid 1, the state of $\calP \times \calW_x \times \calW_y$ is in $S^1_x$, where $x$ is the value written on $\calQ_x$.
\end{claim}
\begin{proof}
    \Cref{clm:invariant} says that at the start of any query, the state of $\calD \times \calP \times \calW_x \times \calW_y$ is in $\ket{\emptyset}\ket{\psi^1, \emptyset, \emptyset}$ or $\ket{x'}\ket{\phi}$ for some $\ket{\phi} \in S^1_{\perp,x'}$. In the second case, we can assume that $x'=x$ because otherwise, the query will immediately abort, and $\EvalCopy_1$ will not be executed.

    Next, the query to $\Sim'(P)$ applies $\ReflectCopy_1$, which acts as follows (\cref{thm:reflect-and-copy}):
\begin{align*}
    \ReflectCopy_1 :& \ket{\emptyset} \ket{\psi^1, \emptyset, \emptyset} \mapsto \ket{x} \ket{\psi^1, \emptyset, \emptyset}\\
    &\ket{x} \ket{\phi} \mapsto \ket{x} \ket{\phi}
\end{align*}
for any $\ket{\phi} \in S^1_{\perp,x}$. In either case, $\ReflectCopy_1$ maps states in $S^1_x$ to states in $S^1_x$.

Next, the only time that $\EvalCopy_1$ is invoked is right after the first invocation of $\ReflectCopy$. We've shown that at the start of $\EvalCopy_1$, the state of $\calP \times \calW_x \times \calW_y$ is in $S^1_x$.
\end{proof}

\begin{claim}
    Hybrids 1 and 2 are perfectly indistinguishable.
\end{claim}
\begin{proof}
First, we make the following observations:
\begin{enumerate}
    \item $\CU^\dag \circ \Bitflip \circ \CU = \Bitflip$. This is because $\Bitflip$ acts only on $\calQ_b$, and $\CU$ acts only on $\calQ_x \times \calP \times \calW_x \times \calW_y$. Since they act on disjoint registers, they commute, so $\CU^\dag \circ \Bitflip \circ \CU = \CU^\dag \circ \CU \circ \Bitflip = \Bitflip$.
    \item $\CU^\dagger \circ \ReflectCopy_2 \circ \CU = \ReflectCopy_1$. 
    
    First, $\ReflectCopy_2$ is equivalent to the operation that applies $\CNOT$ from $\calQ_x$ to $\calD$ controlled on $\calP \times \calW_x \times \calW_y \times \calW_b$ being in the state $\ket{\psi^2, \emptyset, \emptyset, 0}$ (\cref{thm:reflect-and-copy}). Second, $\CU$ acts as follows (\cref{thm:CU}):
    \[\ket{x}_{\calQ_x}\ket{\psi^1, \emptyset, \emptyset, 0} \to \ket{x}_{\calQ_x}\ket{\psi^2, \emptyset, \emptyset, 0}\]
    That implies that $\CU^\dag \circ \ReflectCopy_2 \circ \CU$ applies $\CNOT$ from $\calQ_x$ to $\calD$ controlled on the state $\ket{\psi^1, \emptyset, \emptyset, 0}$. This is the same as $\ReflectCopy_1$.
    \item \label{item:obs2}For every $x, y$,
    \begin{align*}
        \EvalCopy_1 (\ket{x,u} \ket{\psi^1_{x,y}}) &= \ket{x, u + y} \ket{\psi^1_{x,y}}\\
        &= \CU^\dagger \circ \EvalCopy_2 (\ket{x, u} \ket{\psi^2_{x,y}})\\
        &= \CU^\dagger \circ \EvalCopy_2 \circ \CU (\ket{x, u} \ket{\psi^1_{x,y}})
    \end{align*}
    \item For every $x, y$, and every state $\ket{\phi} \in S^1_x$, we have that
    \begin{align*}
        \EvalCopy_1 (\ket{x, u} \ket{\phi}) = \CU^\dagger \circ \EvalCopy_2 \circ \CU (\ket{x, u} \ket{\phi})
    \end{align*}
    This follows from the previous observation (\Cref{item:obs2}).
\end{enumerate}

By \cref{thm:input-to-eval-copy}, in hybrid 1, $\EvalCopy_1$ is only invoked on states in $S^1_x$, and on these states, $\EvalCopy_1$ acts the same as $\CU^\dagger \circ \EvalCopy_2 \circ \CU$. Therefore, in hybrid 1, we can replace $\EvalCopy_1$ with $\CU^\dagger \circ \EvalCopy_2 \circ \CU$, and the change will be perfectly indistinguishable.

Let us put everything together. We can replace
    \[\Bitflip \circ \ReflectCopy_1 \circ \EvalCopy_1 \circ \ReflectCopy_1\]
from hybrid 1 with
\ifllncs
\begin{align*}
    &\textcolor{red}{\CU^\dag} \circ \Bitflip \circ \textcolor{red}{\CU} \circ \textcolor{red}{\CU^\dag} \circ \ReflectCopy_{\textcolor{red}{2}} \circ \textcolor{red}{\CU}\\
    &\quad\circ \textcolor{red}{\CU^\dag} \circ \EvalCopy_{\textcolor{red}{2}} \circ \textcolor{red}{\CU} \circ \textcolor{red}{\CU^\dag} \circ \ReflectCopy_{\textcolor{red}{2}} \circ \textcolor{red}{\CU}
\end{align*}
\else
\[\textcolor{red}{\CU^\dag} \circ \Bitflip \circ \textcolor{red}{\CU} \circ \textcolor{red}{\CU^\dag} \circ \ReflectCopy_{\textcolor{red}{2}} \circ \textcolor{red}{\CU} \circ \textcolor{red}{\CU^\dag} \circ \EvalCopy_{\textcolor{red}{2}} \circ \textcolor{red}{\CU} \circ \textcolor{red}{\CU^\dag} \circ \ReflectCopy_{\textcolor{red}{2}} \circ \textcolor{red}{\CU}\]
\fi
Then when we cancel out adjacent applications of $\CU$ and $\CU^\dag$, the operation becomes:
\ifllncs
\begin{align*}
    &\textcolor{red}{\CU^\dag} \circ \Bitflip \circ \ReflectCopy_{\textcolor{red}{2}} \circ \EvalCopy_{\textcolor{red}{2}}\\
    &\quad\circ \ReflectCopy_{\textcolor{red}{2}} \circ \textcolor{red}{\CU}
\end{align*}
\else
\[\textcolor{red}{\CU^\dag} \circ \Bitflip \circ \ReflectCopy_{\textcolor{red}{2}} \circ \EvalCopy_{\textcolor{red}{2}} \circ \ReflectCopy_{\textcolor{red}{2}} \circ \textcolor{red}{\CU}\]
\fi
This is the sequence of operations found in hybrid 2. Therefore, hybrids 1 and 2 are perfectly indistinguishable.

\end{proof}

\begin{claim}
    Hybrids 2 and 3 are perfectly indistinguishable. 
\end{claim}
\begin{proof}
    In hybrid 2, the first application of $\CU$ during the first query to the simulator converts $\ket{\psi^1, \emptyset, \emptyset}$ into $\ket{\psi^2, \emptyset, \emptyset}$, which is consistent with the initial state of $\ket{\psi^2, \emptyset, \emptyset}$ in hybrid 3. We would now like to argue that the applications of $\CU$ and $\CU^\dagger$ cancel each other, but this is not true in general: in between queries, the adversary can act on the input query register, so that we are effectively applying a $U_{x'} \circ U^\dagger_x$ operation on the internal state of the oracle, which is not the identity operation in general.

    We make use of the invariant from \Cref{clm:invariant}: In every branch of the superposition of hybrid 2, after the end of a query, either $D = \emptyset$ and $\calP \times \calW_x \times \calW_y = \ket{\psi^1, \emptyset, \emptyset}$, or $D$ contains $x$, and $\calP \times \calW_x \times \calW_y$ contains $\ket{\phi} \in S^1_{\bot,x}$.

    In the branch with $\ket{\psi^1, \emptyset, \emptyset}$, $\CU$ acts the same no matter what the control register is, and always transforms $\ket{\psi^1, \emptyset, \emptyset}$ to $\ket{\psi^2, \emptyset, \emptyset}$, In this branch, the applications of $\CU$ and $\CU^\dagger$ cancel each other. In the branch with $\ket{\phi} \ket{x}$, where $\ket{\phi} \in S^1_{\bot,x}$, if $x'\neq x$, the next query acts as the identity on this branch, so we can ignore it.
\end{proof}
This shows that $\Sim'(P_1)$ and $\Sim'(P_2)$ are perfectly indistinguishable and completes the proof of \Cref{lem:testable-bpotp-sim-p1-p2-perfect}.
\end{proof}

\begin{lemma}[Simulation of the CSEQ oracle]\label{thm:simulating-CSEQ}
    Given any (possibly oracle-aided) testable quantum program $P$ that implements $f$ with $0$ error, an oracle for $\Sim'(P)$ (\cref{def:sim-for-CSEQ}) is perfectly indistinguishable from $O_f^\CSEQ$ to any adversary making an unbounded number of quantum queries.
\end{lemma}
\begin{proof}
    First, $\Sim'$ only needs black-box access to the $\Eval$ and $R$ operations of $P$. It does not access $\ket{\psi}$ directly. $P$ may even be oracle-aided, and the oracle may maintain an internal pure state, which is considered part of $\ket{\psi}$.

    Second, $P_f$ (defined in \cref{lem:bpqotp-seq-sim-prime}) and $P$ are testable quantum programs. $P_f$ implements $f$ with $0$ error. \Cref{lem:bpqotp-seq-sim-prime} says that $\Sim'(P_f)$ and $O_f^\CSEQ$ are perfectly indistinguishable after an unbounded number of quantum queries.

    Third, since $P$ also implements $f$ with $0$ error, \cref{lem:testable-bpotp-sim-p1-p2-perfect} implies that oracles for $\Sim'(P_f)$ and $\Sim'(P)$ are perfectly indistinguishable after an unbounded number of quantum queries. Therefore, $\Sim'(P)$ and $O_f^\CSEQ$ are also perfectly indistinguishable after an unbounded number of quantum queries.
\end{proof}

\begin{remark}
    For classical functionalities, any OTP compiler that satisfies CSEQ security is best-possible among testable programs, in the sense of \cref{def:bp-testable-q}. This follows from \Cref{thm:simulating-CSEQ}; the compiler's own outputs need not be testable.
    \end{remark}

\begin{theorem}\label{thm:cseq-seq-equivalence}
$ $
\begin{enumerate}
    \item For any sets of bitstrings $\calX, \calR, \calY$, there exists a Q.P.T. simulator $\Sim_1$ such that for every randomized function $f: \calX \times \calR \to \calY$, $O^\CSEQ_f$ and $\Sim_1^{O^\SEQ_{\Phi_f}}$ are perfectly indistinguishable.

    \item For any sets of bitstrings $\calX, \calR, \calY$, \textbf{for which $|\calX| > 1$}, there exists a Q.P.T. simulator $\Sim_2$ such that for every randomized function $f: \calX \times \calR \to \calY$, $\Sim_2^{O^\CSEQ_f}$ and $O^\SEQ_{\Phi_f}$ are perfectly indistinguishable.
\end{enumerate}
\end{theorem}
\begin{proof}
$ $

\paragraph{Simulating $O_f^\CSEQ$ with $\Sim_1$:} Let $\Sim'$ be the simulator from \cref{def:sim-for-CSEQ}. By \cref{thm:simulating-CSEQ}, if $\Sim'$ is given black-box access to any testable quantum program $P$ that implements $f$ with $0$ error, then $\Sim'(P)$ is perfectly indistinguishable from $O^\CSEQ_f$. The program $P$ may be oracle-aided, and the oracle may maintain an internal pure state.

Next, $O^\SEQ_{\Phi_f}$ can be used to construct such a testable oracle interface that implements $f$ with $0$ error (\cref{thm:O-SEQ-can-implement-f}). Composing these two procedures yields $\Sim_1$.

\paragraph{Simulating $O_{\Phi_f}^\SEQ$ with $\Sim_2$:} $\Sim_2$ mainly uses the simulator $\Sim'$ from \cref{def:sim-for-SEQ-oracle} (this is a different $\Sim'$ than the one used to construct $\Sim_1$). This $\Sim'$ takes any testable quantum program that implements $\Phi_f$ and perfectly simulates $O^\SEQ_{\Phi_f}$ (\cref{thm:sim-for-SEQ-oracle}). The program can be oracle-aided, and the oracle can even maintain an internal state. $\Sim'$ only needs black-box access to the program's $\Eval$ and $R$ operations. 

Next, we can use $O^\CSEQ_f$ to construct a testable oracle interface that implements $\Phi_f$ as long as $|\calX|>1$ (\cref{thm:O-CSEQ-can-implement-Phi-f}). Then if we run $\Sim'$ on this interface, it will perfectly simulate $O^\SEQ_{\Phi_f}$.
\end{proof}

\begin{lemma}\label{thm:O-SEQ-can-implement-f}
    $O^\SEQ_{\Phi_f}$ can be used to construct a testable oracle interface that implements $f$ with $0$ error.
\end{lemma}
The program is oracle-aided since it makes queries to $O^\SEQ_{\Phi_f}$, and the program state includes the internal state of the oracle.

\begin{proof}
Let us construct a testable oracle interface $P = (\ket{\psi}, \Eval, \mathsf{R})$ that implements $f$.
\begin{itemize}
    \item Program State $\ket{\psi}$:
    Let the program state be the internal state of $O^\SEQ_{\Phi_f}$ concatenated with a cache register $\calB$ that is initialized to $\ket{0}$. 
    Note that the internal state of $O^\SEQ_{\Phi_f}$ is defined as a pure state, so the initial program state $\ket{\psi}$ is indeed pure.

    \item $\Eval$: The program's $\Eval$ operation simply queries $O^\SEQ_{\Phi_f}$. $\Eval$ takes an external query register $\calQ$ with basis states of the form $\ket{x, u}$ where $(x,u) \in \calX \times \calY$. Then $\Eval$ queries $O^\SEQ_{\Phi_f}$ on $\calQ \times \calB$:
    \[\ket{x,u}_\calQ \otimes \ket{b}_\calB \to \ket{x,u \oplus y}_\calQ \otimes \ket{b \oplus b'}_\calB\]
    Finally, $\Eval$ outputs $\calQ$.

    \item $R$: The reflection operation $R$ takes an external register $\calE$. If $\calB$ stores $0$, then $R$ applies $\mathsf{X}$ to $\calE$. 
\end{itemize}

    \begin{claim}
        The program $P$ defined above implements $f$ with $0$ error (\cref{def:quantum-samp-program}).
    \end{claim}
    \begin{proof}
        Given a query of the form $\ket{x, 0}$, the program's $\Eval$ operation queries $O^\SEQ_{\Phi_f}$ on $\ket{x, 0, b}$. The first time $O^\SEQ_{\Phi_f}$ is queried, it applies the unitary $U_{\Phi_f}$, which uses a compressed oracle to compute $y = f(x;r)$ for a uniformly random $r \getsr \calR$. This is exactly the output distribution of $f(x)$.
    \end{proof}

    The following claim shows that $P$ is indeed testable.

    \begin{claim}
        After any number of $\Eval$ and $R$ operations, applying $R$ is equivalent to applying the function $\mathsf{X} \otimes \ketbra{\psi} + I \otimes (I - \ketbra{\psi})$ to $\calE$ and the current program state.
    \end{claim}
    \begin{proof}
        Here is some intuition for why $R$ works as the reflection oracle. First, $O^\SEQ_{\Phi_f}$'s internal state includes a register $\calC$ that is $\ket{0}$ if and only if $O^\SEQ_{\Phi_f}$ is in its initial state. Second, any time that $\calC$ is changed, $O^\SEQ_{\Phi_f}$ flips the output bit $b$. This implies that $\calC$ and $\calB$ always store the same value after any query, and that $\calB = \ket{0}$ if and only if the program state is in its initial state. Next, we will make this argument formal.\\
    
        The program state comprises three registers: $\calC$ and $\calP$, which contain $O^\SEQ_{\Phi_f}$'s state, and $\calB$. The initial state of $\calB \times \calC \times \calP$ is $\ket{\psi} := \ket{0} \otimes \ket{0} \otimes \ket{0^m}$.
        
        We will show that after any number of $\Eval$ and $R$ operations, $\calB = \ket{0}$ if and only if $\calB \times \calC \times \calP = \ket{\psi}$. 

        After any number of $\Eval$ and $R$ operations, $\calB$ and $\calC$ store the same computational-basis value. First, $\calB \times \calC$ are initialized to $\ket{0} \otimes \ket{0}$. Next, the only time that $\calC$'s computational-basis value changes is potentially during an $\Eval$ operation, specifically during step \ref{CSEQ-oracle-step-2} of the query to $O^\SEQ_{\Phi_f}$. If $\calC$ is not flipped in step \ref{CSEQ-oracle-step-2}, then this call to $O^\SEQ_{\Phi_f}$ acts as the identity, and $\calB$ is not flipped either. On the other hand, if $\calC$ is flipped in step \ref{CSEQ-oracle-step-2}, then this call to $O^\SEQ_{\Phi_f}$ applies $U_{\Phi_f}$ or $U_{\Phi_f}^\dag$, which are equal according to \cref{thm:U-Phi-f-involution}. $U_{\Phi_f}$ flips the value of $b$. Since the program records the value of $b$ on register $\calB$, that means that $\calB$ is flipped on any query that flips $\calC$. Therefore, after any number of $\Eval$ or $R$ operations, $\calB \times \calC$ store $(0,0)$ or $(1,1)$.

        After any number of $\Eval$ and $R$ operations, if $\calC = \ket{0}$ then $\calP = \ket{0^m}$. First, $\calC \times \calP$ are initialized to $\ket{0} \otimes \ket{0^m}$, so the invariant is initially true. Second, $R$ does not change the state of $\calC \times \calP$. Third, the first $\Eval$ query flips $\calC$ to $\ket{1}$ with certainty. If any future $\Eval$ query flips $\calC$ back to $\ket{0}$, then the check in step \ref{CSEQ-oracle-step-2} of $O^\SEQ_{\Phi_f}$ must have found that $\calP = \ket{0^m}$ and flipped $\calC$ to $\ket{0}$. Then step \ref{CSEQ-oracle-step-3} would have acted as the identity because $\calC = \ket{0}$, so at the end of the query, it is still true that $\calP = \ket{0^m}$. This argument extends to show that after any number of queries, either $\calC = \ket{1}$ or $\calC \times \calP = \ket{0} \otimes \ket{0^m}$.

        The previous discussion implies that after any number of $\Eval$ and $R$ operations, if $\calB = \ket{0}$, then $\calC = \ket{0}$, and $\calP = \ket{0^m}$, so $\calB \times \calC \times \calP = \ket{\psi}$. Next, if $\calB = \ket{1}$, then $\calB \times \calC \times \calP \neq \ket{\psi}$. Therefore, $\calB = \ket{0}$ if and only if $\calB \times \calC \times \calP = \ket{\psi}$. 
        
        Finally, $R$ applies $\mathsf{X}$ to $\calE$ if $\calB = \ket{0}$. This is equivalent to applying $\mathsf{X}$ to $\calE$ if $\calB \times \calC \times \calP$ is in state $\ket{\psi}$.

    \end{proof}

\begin{claim}\label{thm:U-Phi-f-involution}
        For the unitary $U_{\Phi_f}$ defined in \cref{def:SEQ-channel}, $U^\dag_{\Phi_f} = U_{\Phi_f}$.
    \end{claim}
    \begin{proof}
        First, steps \ref{step:copy-x} - \ref{step:answer-query} of $U_{\Phi_f}$ are each their own inverses. Step \ref{step:copy-x} CNOTs $x$ from $\calQ$ to $\calW$, which can be uncomputed by applying step \ref{step:copy-x} a second time. 
        
        Step \ref{step:query-CO} queries the compressed oracle. The compressed oracle is perfectly indistinguishable from a random oracle when queried as a black box, and any query to a random oracle can be uncomputed by querying the random oracle a second time. Therefore, step \ref{step:query-CO} can be uncomputed by computing step \ref{step:query-CO} a second time.

        Step \ref{step:answer-query} CNOTs some values that are computed from $\calW$ onto $\calQ$. This can be uncomputed by applying step \ref{step:answer-query} a second time.

        Second, $U_{\Phi_f}$ and $U_{\Phi_f}^\dag$ apply the following sequences of operations:
        \begin{align*}
            U_{\Phi_f} &= \text{step \ref{step:copy-x}}^\dag \cdot \text{step \ref{step:query-CO}}^\dag \cdot \text{step \ref{step:answer-query}} \cdot \text{step \ref{step:query-CO}} \cdot \text{step \ref{step:copy-x}}\\
            U_{\Phi_f}^\dag &= \left(\text{step \ref{step:copy-x}}^\dag \cdot \text{step \ref{step:query-CO}}^\dag \cdot \text{step \ref{step:answer-query}} \cdot \text{step \ref{step:query-CO}} \cdot \text{step \ref{step:copy-x}}\right)^\dag\\
            &= \text{step \ref{step:copy-x}}^\dag \cdot \text{step \ref{step:query-CO}}^\dag \cdot \text{step \ref{step:answer-query}}^\dag \cdot \left(\text{step \ref{step:query-CO}}^\dag\right)^\dag \cdot \left(\text{step \ref{step:copy-x}}^\dag\right)^\dag\\
            &= \text{step \ref{step:copy-x}}^\dag \cdot \text{step \ref{step:query-CO}}^\dag \cdot \text{step \ref{step:answer-query}} \cdot \text{step \ref{step:query-CO}} \cdot \text{step \ref{step:copy-x}}\\
            &= U_{\Phi_f}
        \end{align*}
    \end{proof}
\end{proof}

\begin{lemma}\label{thm:O-CSEQ-can-implement-Phi-f}
    If $|\calX| > 1$, then $O^\CSEQ_f$ can be used to construct a testable oracle interface that implements $\Phi_f$.
\end{lemma}
\begin{proof}
    Let us construct a testable oracle interface $P = (\ket{\psi}, \Eval, \mathsf{R})$ that implements $\Phi_f$.
    \begin{itemize}
        \item Program State $\ket{\psi}$: Let the program state $\ket{\psi}$ be the internal state of $O^\CSEQ_f$ concatenated with several work registers. Let $x_0, x_1 \in \calX$ be two distinct values, and then initialize the work registers as follows:
        \[\calQ_0 = \ket{x_0,0, 0}, \calQ_1 = \ket{x_1, 0, 0}, \calB_0 = \ket{0}, \calB_1 = \ket{0}\]
        Note that the internal state of $O^\CSEQ_f$ is defined as a pure state, so $\ket{\psi}$ is indeed pure.

        \item $\Eval$: The program's $\Eval$ operation simply queries $O^\CSEQ_f$.

        \item $R$: The reflection operation $R$ should test whether the program state is in its initial state. Part of the program state is contained in the oracle $O^\CSEQ_f$, so our construction cannot access it directly. Instead, we can make queries to $O^\CSEQ_f$ to test if its internal state is in its initial state. Our strategy is to query $O^\CSEQ_f$ on two different inputs, while uncomputing between queries. If the oracle answers both queries, then its state began in the initial state. If it rejects at least one query, then its database must already record a query.

        Formally, $R$ takes as input a single-qubit register $\calE$ and acts as follows.
    \begin{enumerate}
        \item Query $O^\CSEQ_f$ on $\calQ_0$. The resulting state is $\ket{x_0, y_0, b_0}_{\calQ_0}$ for some $y_0 \in \calY$ and $b_0 \in \bit$.\label{step:reflection-x-0-query}
        \item Copy $b_0$ onto $\calB_0$.\label{step:reflection-copy-b-0}
        \item Uncompute \cref{step:reflection-x-0-query}.\label{step:reflection-uncompute-x-0-query}
        \item Query $O^\CSEQ_f$ on $\calQ_1$. The resulting state is $\ket{x_1, y_1, b_1}_{\calQ_1}$ for some $y_1 \in \calY$ and $b_1 \in \bit$.\label{step:reflection-x-1-query}
        \item Copy $b_1$ onto $\calB_1$.\label{step:reflection-copy-b-1}
        \item Uncompute \cref{step:reflection-x-1-query}.\label{step:reflection-uncompute-x-1-query}
        \item If $\ket{b_0, b_1}_{\calB_0 \times \calB_1} = \ket{1, 1}$, then apply $\mathsf{X}$ (bitflip) to register $\calE$.\label{step:reflection-bitflip}
        \item Uncompute steps \ref{step:reflection-x-0-query} - \ref{step:reflection-uncompute-x-1-query}.\label{step:reflection-uncompute-all-previous-steps}
    \end{enumerate}
\end{itemize}

    \begin{claim}
        $P$ implements $\Phi_f$ (with $0$ error).
    \end{claim}
    \begin{proof}
        When we apply $\Eval$ to the user's query register $\calQ$ and the initial program state $\ket{\psi}$, it queries $O^\CSEQ_f$. Since the database of $O^\CSEQ_f$ is initialized to $\ket{\emptyset}$, step \ref{step:SEQ-check} of $O^\CSEQ_f$ does nothing during this query.

        If we omit step \ref{step:SEQ-check} of $O^\CSEQ_f$, then the remaining steps are exactly the same as $U_{\Phi_f}$. Therefore, applying $\Eval$ with program state $\ket{\psi}$ computes the same operation as $\Phi_f$.
    \end{proof}

    The following claim shows that $P$ is indeed testable.

    \begin{claim}
        After any number of $\Eval$ and $R$ operations, applying $R$ is equivalent to applying the function $\mathsf{X} \otimes \ketbra{\psi} + I \otimes (I - \ketbra{\psi})$ to $\calE$ and the current program state.
    \end{claim}
    \begin{proof}
        Let us consider the case where $P$'s program state is in its initial state $\ket{\psi}$. We will show that in this case, $R$ applies $\mathsf{X}$ to $\calE$ and restores the program state to its initial state. 
        
        Step \ref{step:reflection-x-0-query} of $R$ queries $O^\CSEQ_f$ on input $\ket{x_0, 0, 0}$. Since we are in the initial state, $O^\CSEQ_f$ responds to the query and flips $b_0$ to $1$ with certainty. Step \ref{step:reflection-copy-b-0} copies $b_0$ to another register $\calB_0$, but this step does not change the joint state on $\calQ_0$ and the internal registers of $O^\CSEQ_f$ because the value $b_0 = 1$ is deterministic. Step \ref{step:reflection-uncompute-x-0-query} uncomputes step \ref{step:reflection-x-0-query}. Steps \ref{step:reflection-x-0-query} and \ref{step:reflection-uncompute-x-0-query} act only on $\calQ_0$ and the internal state of $O^\CSEQ_f$, and the state on these registers is unchanged by step \ref{step:reflection-copy-b-0}. Therefore the state of $\calQ_0$ and the internal state of $O^\CSEQ_f$ is the same at the end of step \ref{step:reflection-uncompute-x-0-query} as it was at the start of step \ref{step:reflection-x-0-query}. The only difference in the program's state is that at the end of step \ref{step:reflection-uncompute-x-0-query}, $\calB_0$ contains $\ket{1}$.

        Next, steps \ref{step:reflection-x-1-query} - \ref{step:reflection-uncompute-x-1-query} are the same as steps \ref{step:reflection-x-0-query} - \ref{step:reflection-uncompute-x-0-query} except the query register is $\calQ_1 = \ket{x_1,0,0}$. By the same argument as before, the state at the end of step \ref{step:reflection-uncompute-x-1-query} is the same as the state at the start of step \ref{step:reflection-x-0-query} except that $\calB_0 = \ket{1}$ and $\calB_1 = \ket{1}$.

        Next, step \ref{step:reflection-bitflip} applies $\mathsf{X}$ to $\calE$ with certainty. This step does not change the state on the program's registers because the application of $\mathsf{X}$ occurs with probability $1$. Finally, step \ref{step:reflection-uncompute-all-previous-steps} uncomputes all steps except for the application of $\mathsf{X}$. At the end of $R$, the program state is the initial state $\ket{\psi}$, and the only change from the beginning of $R$ is the application of $\mathsf{X}$ to $\calE$.\\
        
        Next, let us consider the case where the program state is orthogonal to $\ket{\psi}$. We only need to consider states that are reachable from a sequence of $\Eval$ and $R$ operations. We will show that in this case, $R$ acts as the identity.

        If the program state is orthogonal to $\ket{\psi}$, then the database register $\calD$ must contain a non-empty database $D$. The program state comprises the registers: $\calW \times \calD$, which are the internal state of $O^\CSEQ_f$, as well as $\calQ_0 \times \calQ_1 \times \calB_0 \times \calB_1$. After any query to $O^\CSEQ_f$, the $\calW$ register is in the $\ket{0}$ state because any values written to it have been uncomputed during the query. Likewise, the state of $\calQ_0 \times \calQ_1 \times \calB_0 \times \calB_1$ at the start of any $R$ operation is the same as its initial state (\cref{thm:work-registers-in-initial-state}). After any number of queries to $\Eval$ and $R$, the only register that might not be in its initial state is $\calD$. If the program state is orthogonal to $\ket{\psi}$, then $\calD$ must contain a non-empty database. We will consider two types of non-empty databases: (1) every entry of $D$ (the only entry, really) is of the form $(x_0, r)$ for some $r \in \calR$, and (2) $D$ contains $(x', r)$ for some $x' \neq x_0$ and some $r \in \calR$.
    
        Now let us step through $R$ in the first case, where every entry of $D$ is of the form $(x_0, r)$ for some $r \in \calR$. Step \ref{step:reflection-x-0-query} queries $O^\CSEQ_f$ on input $\ket{x_0, 0, 0}$, and the query is answered ($b_0=1$). Step \ref{step:reflection-copy-b-0} copies $b_0$ to $\calB_0$, and step \ref{step:reflection-uncompute-x-0-query} uncomputes step \ref{step:reflection-x-0-query}. The state of the program at the end of step \ref{step:reflection-uncompute-x-0-query} is the same as the state at the beginning of step \ref{step:reflection-x-0-query} except that $\calB_0 = \ket{1}$. Next, step \ref{step:reflection-x-1-query} queries $O^\CSEQ_f$ on $x_1$. This query is rejected ($b_1=0$) because $D$ contains $(x_0,r)$. By step \ref{step:reflection-bitflip}, $\calB_0 \times \calB_1 = \ket{1,0}$, so step \ref{step:reflection-bitflip} acts as the identity. Finally, step \ref{step:reflection-uncompute-all-previous-steps} uncomputes steps \ref{step:reflection-x-0-query} - \ref{step:reflection-uncompute-x-1-query}, so $R$ acts as the identity on the program state.

        Finally, let us step through $R$ in the second case, where $D$ contains $(x', r)$ for some $x' \neq x_0$ and some $r \in \calR$. Then the query to $O^\CSEQ_f$ in step \ref{step:reflection-x-0-query} will be rejected ($b_0 = 0$). Skipping ahead, at the start of step \ref{step:reflection-bitflip}, $\calB_0 = \ket{0}$, so this step acts as the identity. Finally, step \ref{step:reflection-uncompute-all-previous-steps} uncomputes steps \ref{step:reflection-x-0-query} - \ref{step:reflection-uncompute-x-1-query}, so $R$ acts as the identity on the program state.

        In summary, we've shown that for every program state that is reachable by a sequence of $\Eval$ and $R$ operations, if the program state is orthogonal to $\ket{\psi}$, then $R$ acts as the identity on the program state.
    \end{proof}

        \begin{claim}\label{thm:work-registers-in-initial-state}
        After any number of $\Eval$ and $R$ operations, the work registers are in their initial state:
        \[\calQ_0 = \ket{x_0,0, 0}, \calQ_1 = \ket{x_1, 0, 0}, \calB_0 = \ket{0}, \calB_1 = \ket{0}\]
    \end{claim}
    \begin{proof}
        Before any $\Eval$ and $R$ operations have been executed, the work registers $\calQ_0 \times \calQ_1 \times \calB_0 \times \calB_1$ are in their initial state. Next, $\Eval$ does not modify the work registers. Finally, it suffices to analyze an $R$ operation and show that if the work registers are in their initial state at the start of the operation, then they will return to their initial state at the end of the operation. 

        Let us split up the program's registers into two groups. Let register $\calR$ comprise all components of the work registers $\calQ_0 \times \calQ_1 \times \calB_0 \times \calB_1$ except the values $x_0$ and $x_1$ written on $\calQ_0$ and $\calQ_1$ respectively. Let $\calS$ comprise the internal registers of $O^\CSEQ_f$ and the values $x_0$ and $x_1$ written on $\calQ_0$ and $\calQ_1$.
        
        Next, no step of $R$ or $O^\CSEQ_f$ modifies the computational basis value of $x_0$ and $x_1$ written on $\calQ_0$ and $\calQ_1$. Therefore these values remain in unchanged by $R$.

        Steps \ref{step:reflection-x-0-query} - \ref{step:reflection-uncompute-x-1-query} of $R$ may modify register $\calR$, but they only do so by CNOT-ing a value onto $\calR$ that was computed from the state on register $\calS$. Next, step \ref{step:reflection-bitflip} applies an $\mathsf{X}$ operation to $\calE$ controlled on the value of $\calR$. However, step \ref{step:reflection-bitflip} does not change the state of $\calS$. Finally, step \ref{step:reflection-uncompute-all-previous-steps} uncomputes steps \ref{step:reflection-x-0-query} - \ref{step:reflection-uncompute-x-1-query}. Again, this entails CNOT-ing a value onto $\calR$ that is computed from the state on $\calS$. The value that is CNOT-ed in step \ref{step:reflection-uncompute-all-previous-steps} is the same as the value that was CNOT-ed during steps \ref{step:reflection-x-0-query} - \ref{step:reflection-uncompute-x-1-query} because the state on $\calS$ is unchanged. Therefore step \ref{step:reflection-uncompute-all-previous-steps} returns the state of $\calR$ to its initial state.

        In summary, all components of $\calQ_0 \times \calQ_1 \times \calB_0 \times \calB_1$ are returned to their initial state at the end of $R$.
    \end{proof}
\end{proof}

\section{Impossibility of one-time correct sampling QSIO}
\label{sec:sampling-qsio}

In this section, we show that our impossibility for a best-possible one-time compiler in \Cref{sec:impossibility-of-bp-otp} rules out the existence of obfuscation schemes that satisfy a natural definition of quantum state indistinguishability obfuscation for \textit{sampling} programs. Previously, quantum state indistinguishability obfuscation has been studied only in the setting where the quantum program implements a deterministic classical functionality with negligible error~\cite{BKNY23-qio, GM24-otp}.

A natural generalization of the definition of quantum state IO to quantum programs designed for sampling tasks is the following. It requires obfuscations of any two programs that implement the same sampling task to be indistinguishable. Crucially, for two programs to be equal in this sense, they need to evaluate the same distributions on the first query, but there are no guarantees on their equivalence on subsequent queries. A priori, one can image that it might be possible to indistinguishably obfuscate such a pair of programs by enforcing some kind of one-time guarantee, so both obfuscations stop working after one evaluation. The impossibility in this section (\Cref{cor:no-qsio-sampling}) says that this is not possible. Indeed, if it was possible, then it would be a best-possible one-time program. This is formalized in \Cref{lem:qsio-sampling-implies-bpotp}.

We leave open the question of coming up with a feasible notion of quantum state IO for sampling programs. Our notion of stateful obfuscation in \Cref{sec:stateful-io-evidence} (a notion of quantum state IO where the two programs must have the same behavior on polynomially many (forward and inverse) queries, instead of just the first one) can be seen as an attempt to make progress in this direction.

\begin{definition}[Quantum State Obfuscation for Quantum Sampling Programs]\label{defn:qsio-sampling}
    A quantum state obfuscator for quantum sampling programs is a q.p.t.~algorithm $\mathsf{QObf}$ with the following syntax:
    $$\mathsf{QObf}(1^\lambda, \rho, C)\rightarrow (\tilde{\rho}, \tilde{C}).$$
    The obfuscator takes as input a security parameter $1^\lambda$ and a quantum sampling program $(\rho, C)$ and outputs an obfuscated sampling program $(\tilde{\rho}, \tilde{C})$. Both program states may be mixed, and $C$ and $\tilde C$ are unitary evaluation circuits, as in \Cref{def:quantum-samp-program}.
    \begin{itemize}
        \item \textbf{Correctness:} Suppose the program family $\{(\rho_\lambda, C_\lambda)\}_{\lambda \in \mathbb{N}}$ implements a family of sampling functionalities $\{D_\lambda: \{0,1\}^{n(\lambda)}\rightarrow \{0,1\}^{m(\lambda)}\}_{\lambda \in \mathbb{N}}$
        for every $x\in\{0,1\}^{n(\lambda)}$ up to negligible error $\negl[\lambda]$. Then, the family of quantum sampling programs $\{(\tilde{\rho}_\lambda, \tilde{C}_\lambda)\}$ where
        $$(\tilde{\rho}_\lambda, \tilde{C}_\lambda)\leftarrow \mathsf{QObf}(1^\lambda, \rho_\lambda, C_\lambda)$$
        also implements the same sampling functionality up to negligible error.
        \item \textbf{Indistinguishability Obfuscation:} For every pair of families of quantum sampling programs $\{(\rho_{\lambda, 0}, C_{\lambda, 0})\}_{\lambda \in \mathbb{N}}$ and $\{(\rho_{\lambda, 1}, C_{\lambda, 1})\}_{\lambda \in \mathbb{N}}$ that both implement the same sampling functionality $\{D_\lambda: \{0,1\}^{n(\lambda)} \rightarrow \{0,1\}^{m(\lambda)}\}_{\lambda \in \mathbb{N}}$ up to negligible error in $\lambda$, and furthermore satisfy that the program descriptions have the same length $|\rho_{\lambda, 0}, C_{\lambda, 0}| = |\rho_{\lambda, 1}, C_{\lambda, 1}|$ for every $\lambda$, for every q.p.t.~adversary $\{A_\lambda\}_{\lambda \in \mathbb{N}}$ and every family of polynomial-length classical auxiliary inputs $\{z_\lambda\}_{\lambda\in\mathbb N}$, there is a negligible function $\negl$ such that
        \begin{align*}
            \Big|&\Pr\left[1 \leftarrow A_\lambda\left(1^\lambda, z_\lambda, \mathsf{QObf}(1^\lambda, \rho_{\lambda, 0}, C_{\lambda, 0})\right)\right] \\
            &- \Pr\left[1 \leftarrow A_\lambda\left(1^\lambda, z_\lambda, \mathsf{QObf}(1^\lambda, \rho_{\lambda, 1}, C_{\lambda, 1})\right)\right]\Big| \le \negl[\lambda].
        \end{align*}
        
    \end{itemize}
\end{definition}

\begin{lemma}\label{lem:qsio-sampling-implies-bpotp}
    Suppose there exists a quantum state indistinguishability obfuscator for quantum sampling programs for a class of classical randomized functionalities $\calF$. Then, there must exist a one-time compiler that satisfies the best-possible one-time security guarantee of \Cref{def:bp-sim} for $\calF$.
\end{lemma}

\begin{proof}
By assumption, let $\QObf$ be an obfuscator that satisfies quantum state indistinguishability obfuscation for quantum sampling programs for $\calF$. Consider the following quantum sampling program $P_f = (\ketbra{\psi_f}, C_f)$ that implements classical randomized functions $f:\calX \times \calR \rightarrow \calY$ for $f\in\calF$: the program state is a description of $f$, along with a uniform superposition over $\calR$,
\begin{align*}
    \ket{\psi_f} = \ket{f} \otimes \frac{1}{\sqrt{|\calR|}}\sum_{r\in\calR} \ket{r}.
\end{align*}
The circuit $C_f$ maps
\begin{align*}
    \ket{x, u}_\calQ \otimes \ket{f, r}_\calP \mapsto \ket{x, u \oplus f(x; r)}_\calQ \otimes \ket{f, r}_\calP.
\end{align*}
Fix a public polynomial $L$ such that $L(\lambda)\geq\lambda$ and $|P_f|\leq L(\lambda)$ for every $f\in\calF_\lambda$. Such a polynomial exists because this canonical implementation has size polynomial in $|f|$. Let $\mathsf{Pad}_{L(\lambda)}$ pad any shorter program to length $L(\lambda)$ without changing its functionality.
Define $\OTP^*(f) := \QObf(1^\lambda, \mathsf{Pad}_{L(\lambda)}(P_f))$ for $f \in \calF_\lambda$. Correctness follows from the correctness of $\QObf$. We claim that $\OTP^*$ satisfies the best-possible one-time security guarantee of \Cref{def:bp-sim}.
For every q.p.t.\ adversary $\calA$ and negligible correctness error $\mu$, define
\begin{align*}
    \Sim(1^\lambda, P) := \calA(1^\lambda, \QObf(1^\lambda, \mathsf{Pad}_{L(\lambda)}(P))).
\end{align*}
Every admitted competitor $P$ satisfies $|P|=|f|=\lambda$, so the two inputs to $\QObf$ have the same length $L(\lambda)$. They also implement the same sampling functionality up to error $\mu$. Thus, for every q.p.t.~distinguisher $D$, there is a negligible function $\nu$ such that for every $\lambda$, every $f\in\calF_\lambda$, and every such $P$,
\begin{align*}
    \left|\Pr[1 \leftarrow D(1^\lambda, f, \calA(1^\lambda, \OTP^*(f)))] - \Pr[1 \leftarrow D(1^\lambda, f, \Sim(1^\lambda, P))] \right| \leq \nu(\lambda).
\end{align*}
Otherwise, choosing violating pairs $(f_\lambda,P_\lambda)$ along an infinite subsequence of parameters, and taking identical programs at the remaining parameters, would give two equal-length program families contradicting quantum state indistinguishability obfuscation. The corresponding adversary composes $\calA$ with $D$ and uses $z_\lambda=f_\lambda$ as its classical auxiliary input.
\end{proof}

\begin{corollary}\label{cor:no-qsio-sampling}
    Assuming $\mathsf{SZK}\not\subseteq\mathsf{BQP}$, there exists a class of sampling functionalities $\calD$ such that there is no quantum state indistinguishability obfuscator (as defined in \Cref{defn:qsio-sampling}) for $\calD$.
\end{corollary}
This follows from \Cref{lem:qsio-sampling-implies-bpotp,thm:impossibility-of-bp-otp-from-szk}.

\section{One-Time VBB Security When Possible}
\label{app:vbb-when-possible}

We formalize the weakening discussed in \Cref{rem:intro-vbb-when-possible} and explain how the proofs of \Cref{thm:impossibility-of-bp-otp-from-szk,thm:impossibility-of-bp-otp-from-hvpzk} also rule it out.
The phrase ``when possible'' refers to security on a family of functionalities: a simulator must work throughout that family, rather than being chosen separately for each circuit.

\paragraph{One-query simulation.}
Fix the public input and output spaces at each security parameter.
For a classical sampling functionality $f$, let $\mathcal O_f^{(1)}$ allow at most one classical query $x$ and return a classical sample $f(x;r)$ for fresh internal randomness $r$.
It supplies neither a coherent extension nor the circuit description of $f$.
We use the sampling-program syntax and public size convention of \Cref{def:bp-sim}; in particular, a competing program must satisfy $|P|=|f|$.
Compiler randomness is included in the output sampling program, so correctness can hold for that mixed program as a whole.

\begin{definition}[One-time VBB security]
\label{def:one-time-vbb}
A one-time compiler $\mathsf K$ has \emph{one-time VBB security} on a family $\mathcal G$ if for every QPT adversary $\mathcal A$ there is a QPT simulator $\mathsf S$ making at most one classical oracle query such that
\[
    \mathcal A(1^\lambda,\mathsf K(f))
    \approx_c
    \mathsf S^{\mathcal O_f^{(1)}}(1^\lambda,1^{|f|})
\]
uniformly over $f\in\mathcal G_\lambda$.
The simulator may depend on $\mathcal A$, $\mathsf K$, and $\mathcal G$, but not on the particular $f$.
Indistinguishability uses the same distinguisher convention as \Cref{def:bp-sim}.
\end{definition}

\begin{definition}[One-time VBB when possible]
\label{def:one-time-vbb-wp}
A one-time compiler $\mathsf K$ for $\mathcal F$ achieves \emph{one-time VBB when possible} with statistical/perfect equivalence if the following holds for every subfamily $\mathcal G\subseteq\mathcal F$.
If there is a one-time VBB-secure compiler $\mathsf J$ on $\mathcal G$ whose output $\mathsf J(f)$ is an admissible competitor for $f$ in \Cref{def:bp-sim}, then $\mathsf K$ has one-time VBB security on $\mathcal G$.
Here admissibility means $|\mathsf J(f)|=|f|$ and correctness with a common negligible/zero error bound, respectively.
\end{definition}

The definition imposes correctness on all of $\mathcal F$, while its security requirement applies only on subfamilies admitting one-query simulation.
It does not require the compiler to recognize those subfamilies efficiently.
The public size bound prevents the comparison from providing additional program-length information; all circuits and sample-containing programs in the argument below can be padded to satisfy it.

\begin{lemma}
\label{lem:bp-implies-vbb-wp}
A best-possible one-time compiler with statistical/perfect equivalence achieves one-time VBB when possible with the corresponding equivalence notion.
\end{lemma}
\begin{proof}
Fix a subfamily $\mathcal G$, a witnessing compiler $\mathsf J$, and an adversary $\mathcal A$ against $\OTP^*$.
Best-possible security gives a QPT simulator $\mathsf B$ such that
\[
    \mathcal A(1^\lambda,\OTP^*(f))
    \approx_c \mathsf B(1^\lambda,\mathsf J(f))
    \approx_c \mathsf S^{\mathcal O_f^{(1)}}(1^\lambda,1^{|f|})
    \qquad(f\in\mathcal G_\lambda).
\]
The second relation applies the one-time VBB security of $\mathsf J$ to the adversary $\mathsf B$.
Both simulators work uniformly over $\mathcal G$, and the resulting $\mathsf S$ makes at most one classical query.
\end{proof}

\begin{theorem}
\label{thm:vbb-wp-impossible}
If $\mathsf{SZK}\not\subseteq\mathsf{BQP}$ (respectively, $\left.\mathsf{HVPZK}\right|_{\mathrm{am}}\not\subseteq\mathsf{BQP}$), no one-time compiler for all classical sampling functionalities with a single input bit achieves one-time VBB when possible with statistical (respectively, perfect) equivalence.
\end{theorem}
\begin{proof}[Proof sketch]
Suppose $\mathsf K$ is a correct compiler with the claimed security.
For statistical equivalence, use the polarized sampling circuits $C_0,C_1$, functionality $f(b;r)=C_b(r)$, constants $c_y(b)=y$, and EPR adversary $\mathcal A$ from \Cref{thm:impossibility-of-bp-otp-from-szk}.
Use the same common polynomial padding, so all public lengths and input/output spaces agree.
The correctness arguments in that proof give
\[
    \Pr[\mathcal A(\mathsf K(c_y))=1]\leq\negl,
    \qquad
    \Pr[\mathcal A(\mathsf K(f))=1]\geq\tfrac12-\negl
    \quad\text{in the far case}.
\]

Let $\mathcal G$ contain all the close-case functionalities and all the constants used in the reduction.
This subfamily admits a one-time VBB-secure compiler $\mathsf J$: sample $y\gets f(0)$ and output a padded program that always returns $y$.
The common negligible bound on the distance between $D_0$ and $D_1$ gives correctness, and one classical query to $f$ exactly simulates $\mathsf J(f)$.
Thus \Cref{def:one-time-vbb-wp} supplies a single one-query simulator $\mathsf S$ for $\mathcal A$ throughout $\mathcal G$.

Fix a close instance.
For either possible query $b$, the response $f(b)$ is negligibly close to $c_y(b)$ averaged over $y\gets D_0$.
Because $\mathsf S$ makes at most one classical query and does not receive the circuit description, its final states in these two experiments are negligibly close.
Suppressing the common public length and security parameter, we obtain
\begin{align*}
    \Pr[\mathcal A(\mathsf K(f))=1]
    &\approx \Pr[\mathsf S^{\mathcal O_f^{(1)}}=1]\\
    &\approx \mathbb E_{y\gets D_0}\Pr[\mathsf S^{\mathcal O_{c_y}^{(1)}}=1]\\
    &\approx \mathbb E_{y\gets D_0}\Pr[\mathcal A(\mathsf K(c_y))=1]
    \leq \negl[\lambda].
\end{align*}
Each approximation here has negligible additive error, uniformly over the indicated programs.
The first and third use the same simulator on $\mathcal G$; the middle one uses only the one-query response distributions.
Consequently, a constant number of repetitions with freshly compiled programs distinguishes close from far distributions with bounded error for all sufficiently large parameters, giving $\mathsf{SZK}\subseteq\mathsf{BQP}$.

For perfect equivalence, use the product construction for $\mathsf{SD}_{0,1/2}$ in \Cref{thm:impossibility-of-bp-otp-from-hvpzk} and let $\mathcal G$ contain the identical-distribution functionalities and the constants.
The compiler $\mathsf J$ is now perfectly correct, and the middle comparison above is exact.
The same test decides $\mathsf{SD}_{0,1/2}$, so \Cref{lem:perfect-sd-hard} gives $\left.\mathsf{HVPZK}\right|_{\mathrm{am}}\subseteq\mathsf{BQP}$.
In either case, the finitely many smaller instances can be hardwired.
\end{proof}

\section{SEQ Is Occasionally Strictly More Secure than Testable Security}
\label{app:testability-adaptations}

\subsection{A Separation from Testable Implementations}
\label{sec:seq-testability-separation}

SEQ security implies best-possible testable security, but the SEQ interface need not itself be testable. We give a simple pair of classical sampling functionalities whose SEQ implementations reveal strictly less information than any testable implementations.
The channels use only an input qubit and an output qubit. We then show how a fixed output marker makes the SEQ interface testable, whereas one-time VBB security survives such output padding.

Consider the sampling functionalities
\begin{equation}
    f_0(x;r)=0,
    \qquad
    f_1(x;r)=(1-x)r,
    \qquad x,r\in\{0,1\},
    \label{eq:seq-separating-functions}
\end{equation}
where $r$ is uniform. Thus both functionalities output zero at $x=1$, while at $x=0$ they output zero and a fair bit, respectively. Let $\calQ=\calX\otimes\calY$ consist of an input qubit and an output accumulator, and define
\begin{equation}
    W\ket{x,y}=\ket{x,y\oplus(1-x)},
    \qquad
    \Phi_0(\rho)=\rho,
    \qquad
    \Phi_1(\rho)=\frac{\rho+W\rho W^\dagger}{2}.
    \label{eq:seq-separating-channels}
\end{equation}
The first channel does nothing; the second applies the controlled bit flip $W$ with probability $1/2$. On input $\ket{x,0}$, measuring $\calY$ after one channel evaluation gives exactly $f_0(x;r)$ or $f_1(x;r)$.

\begin{theorem}[SEQ can reveal less than testable implementations]
\label{prop:seq-testability-gap}
Given the SEQ interface for an unknown channel from \eqref{eq:seq-separating-channels}, every adversary has distinguishing advantage at most $2/3$, even with arbitrarily many adaptive, controlled, or inverse queries. In contrast, any individually correct testable implementation of the corresponding sampling functionality in \eqref{eq:seq-separating-functions} permits distinguishing advantage $1$.
\end{theorem}

\begin{proof}
The SEQ bound follows from a private subspace on which every query acts as the identity. In both experiments, the initial private state has a constant component in such a subspace, and these components give the adversary exactly the same interaction.

Write $\ket\eta=\ket{0}_{\calP}\ket0_{\calC}$ for the initial private state, and let
\[
    \ket a=\ket0_{\calP}\ket1_{\calC},
    \qquad
    \ket b=\ket1_{\calP}\ket1_{\calC}.
\]
Consider the following isometries:
\begin{align*}
    J\ket\varphi&=\ket\varphi\ket\eta,\\
    V_0\ket\varphi&=\ket\varphi\ket a,\\
    V_1\ket\varphi&=\frac{\ket\varphi\ket a+W\ket\varphi\ket b}{\sqrt2}.
\end{align*}
Tracing out the private register in $V_i$ gives $\Phi_i$. Since the ranges of $J$ and $V_i$ are orthogonal, the SEQ oracle of \Cref{def:CSEQ-oracle} is precisely
\begin{equation}
    O_i=I-2M_iM_i^\dagger,
    \qquad M_i=\frac{J-V_i}{\sqrt2}.
    \label{eq:seq-separation-reflection}
\end{equation}
Indeed, this unitary swaps $J\ket\varphi$ with $V_i\ket\varphi$ and fixes the orthogonal complement of their ranges.

Now define the private states
\[
    \ket{\zeta_0}=\frac{\ket\eta+\ket a}{\sqrt2},
    \qquad
    \ket{\zeta_1}=\frac{\ket\eta+\sqrt2\ket a}{\sqrt3}.
\]
For every public state $\ket\varphi$, direct substitution gives
\[
    V_i^\dagger(\ket\varphi\ket{\zeta_i})
    =\braket{\eta|\zeta_i}\ket\varphi
    =J^\dagger(\ket\varphi\ket{\zeta_i}).
\]
Thus $M_i^\dagger(\ket\varphi\ket{\zeta_i})=0$, so $O_i$ acts as the identity on $\calH_{\calQ}\otimes\operatorname{span}\{\ket{\zeta_i}\}$. Its orthogonal complement is also invariant. Every operation on public registers, including controlled or inverse oracle queries, preserves this private-subspace decomposition. Consequently, the coherence between the two private subspaces vanishes when the private registers are traced out at the end of any interaction.

Fix an arbitrary adversary. Let $\sigma$ be its final state when every oracle query acts as the identity, and let $\tau_0,\tau_1$ be its final states conditioned on the complementary private subspaces in the two experiments. Since
\[
    |\braket{\zeta_0|\eta}|^2=\frac12,
    \qquad
    |\braket{\zeta_1|\eta}|^2=\frac13,
\]
the adversary's actual final states satisfy
\begin{align*}
    \rho_0&=\frac12\sigma+\frac12\tau_0
            =\frac13\sigma+\frac23\left(\frac14\sigma+\frac34\tau_0\right),\\
    \rho_1&=\frac13\sigma+\frac23\tau_1.
\end{align*}
They share the same component $\sigma/3$. Their trace distance, and hence the distinguishing advantage, is at most $2/3$. The argument places no bound on the adversary's running time or number of queries.

For a testable implementation, initialize the evaluator on $x=0$, with fresh zero output and work registers and its designated pure program state. Call the joint initial state $\ket\Omega$. Apply $\Eval$, apply $Z$ to the output bit, apply $\Eval^\dagger$, and test the projector $\ketbra\Omega$ using the reflection and the known input/workspace initialization. If the output-one probability is $p$, this test accepts with probability
\begin{equation}
    \bigl|\bra\Omega\Eval^\dagger \cdot Z \cdot \Eval\ket\Omega\bigr|^2
    =(1-2p)^2.
    \label{eq:testable-bias-test}
\end{equation}
For perfect correctness this is $1$ for $f_0$ and $0$ for $f_1$.
\end{proof}

The example separates the SEQ interface from every testable implementation even though its honest output distributions are elementary. Its privacy bound concerns the entire interaction with one initialized oracle, rather than only one oracle query.

\subsection{Recognizable Outputs in the SEQ Interface}
\label{sec:seq-recognizable-outputs}

In this section, we establish a sufficient condition on the channel for the SEQ interface to be testable: if every channel output is orthogonal to some efficiently preparable pure state.
This condition also prevents a probe in that state from entering the computed subspace and causing a reverse evaluation.

\begin{lemma}[Reflection from a missing output direction]
\label{lem:seq-recognizable-output}
Let $\Phi$ be a channel on $\calQ$. Suppose we are given an efficient preparation circuit for a pure state $\ket\chi\in\calQ$ satisfying
\begin{equation}
    \Phi^\dagger(\ketbra\chi)=0,
    \label{eq:seq-missing-output}
\end{equation}
where $\Phi^\dagger$ is the adjoint map. Then two queries to $O_\Phi^{\SEQ}$ suffice to implement the coherent test for its designated initial private state, with all probe registers restored. In particular, the SEQ interface is testable.
\end{lemma}
\begin{proof}
Let $\ket\eta=\ket{0^m}_{\calP}\ket0_{\calC}$ be the initial private state. Define isometries
\[
    J\ket\varphi=\ket\varphi\ket\eta,
    \qquad
    V\ket\varphi=U_\Phi(\ket\varphi\ket{0^m})\ket1_{\calC}.
\]
Their ranges are orthogonal, and the SEQ unitary swaps $J\ket\varphi$ with $V\ket\varphi$, acting as the identity on the orthogonal complement of these two ranges.

Condition~\eqref{eq:seq-missing-output} is equivalent to
\[
    V^\dagger(\ketbra\chi\otimes I)V=0.
\]
Consequently, $\ket\chi\otimes\calH_{\calP\calC}$ is orthogonal to the entire computed subspace. For any private vector $a\ket\eta+\ket h$, where $\braket{\eta|h}=0$, a query on the probe $\ket\chi$ gives
\[
    O_\Phi^{\SEQ}\bigl(\ket\chi(a\ket\eta+\ket h)\bigr)
    =aV\ket\chi+\ket\chi\ket h.
\]
The first term has its query register orthogonal to $\ket\chi$, and the second has query register exactly $\ket\chi$. Coherently flip a fresh test bit when the query register is orthogonal to $\ket\chi$, then uncompute the oracle query. The probe returns to $\ket\chi$, and the test bit flips precisely on the private component $\ket\eta$. The probe can be prepared and tested efficiently by hypothesis.
\end{proof}

For any channel $\Phi$ on $\calQ$, one can enforce this condition by extending the query register with a success bit whose incoming value is ignored:
\begin{equation}
    \widehat\Phi(\rho_{\calQ b})
    =\Phi\bigl(\operatorname{Tr}_b\rho_{\calQ b}\bigr)\otimes\ketbra1_b.
    \label{eq:seq-reset-marker}
\end{equation}
Since the success bit is independent of $\Phi$ itself, intuitively this means that the extra security one may expect from SEQ (compared to testable programs) is simply more often preventing the adversary from knowing when the evaluation succeeded.
Here $\widehat\Phi$ is a channel on the entire enlarged query register, including arbitrary incoming contents of $b$.

Finally, we point out that one-time VBB security (\Cref{def:one-time-vbb}) is preserved when a success bit is appended to the output, unlike SEQ.
Given a one-time VBB compiler for $f$, keep its compiled program and wrap its evaluator to append $1$, thereby implementing $\widehat f(x;r)=(f(x;r),1)$. An attack on this program is also an attack on the original program with the public wrapper applied. Its one-time VBB simulator can answer its single query to $f$ by querying $\widehat f$ and discarding the bit.
Thus the same output padding that makes the SEQ interface testable does not affect one-time VBB security.

\end{document}